%% file: letter_format.tex
\documentclass[twocolumn]{article}
\usepackage[affil-it]{authblk}
\usepackage[text={6.8in,8in},centering]{geometry}
\usepackage{graphicx} 
\usepackage{amsfonts}
\usepackage{amsthm}
\usepackage{amsmath}
\usepackage[english]{babel}
\usepackage{algorithm}
\usepackage{algpseudocode}
\usepackage{makecell}

\usepackage{tikz}
\usepackage{pgf}
\usepackage{amssymb}
\usepackage{graphicx}
\usepackage{subfig} 
\usepackage{url}
\usepackage{ulem}
\usepackage{bbold}
\usepackage{hyperref}
\usepackage{authblk}
\usepackage{color}
\usepackage{qcircuit}
\usepackage{siunitx}
\usetikzlibrary{fadings}

\definecolor{ibmblue}{HTML}{002D9C}
\definecolor{ibmred}{HTML}{da1e28}
\definecolor{ibmpurple}{HTML}{8A3FFC}
\definecolor{ibmcyan}{HTML}{00539A}
\definecolor{ibmgreen}{HTML}{6fdc8c}

\newcommand{\ket}[1]{|#1\rangle}

\title{Low-overhead error detection with spacetime codes}
\author[1]{Simon Martiel\footnote{martiel@ibm.com}}
\author[2]{Ali Javadi-Abhari\footnote{ali.javadi@ibm.com}}
\affil[1]{IBM Quantum, IBM France Lab, Saclay, France}
\affil[2]{IBM Quantum, IBM T.J. Watson Research Center, Yorktown Heights, NY}

\date{}
\begin{document}

\maketitle

\begin{abstract}
We introduce a low-overhead approach for detecting errors in arbitrary Clifford circuits on arbitrary qubit connectivities. Our method is based on the framework of spacetime codes, and is particularly suited to near-term hardware since it has a much milder overhead in qubits and gates compared to error correction, while  achieving a better sampling overhead than existing error mitigation methods. We present efficient algorithms for finding valid checks that are simultaneously low weight, satisfy connectivity constraints, and cover large detecting regions within the circuit. Using this approach, we experimentally demonstrate error detection on circuits of up to 50 logical qubits containing 2450 CZ gates, and show physical to logical fidelity gains of up to $236\times$. Furthermore, we show our algorithm can efficiently find checks in universal circuits, but the space of valid checks diminishes exponentially with the non-Cliffordness of the circuit. These theoretical and experimental results suggest that Clifford-dominated circuits are promising candidates for near-term quantum advantage.
\end{abstract}

 \section*{Introduction}
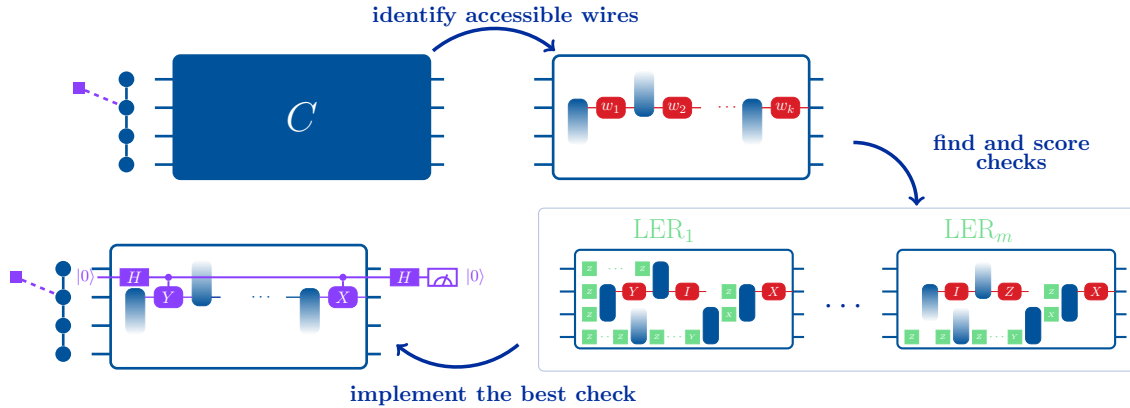
\begin{figure*}[h]
\centering
  \hspace{-1.2cm}
  \raisebox{-1cm}{
    \scalebox{0.84}{
        \input{figures/overview}
    }
    }
  \caption{\textbf{Efficient spacetime check construction in a constrained qubit layout.} (top left) The payload circuit is transpiled onto a subset of physical qubits (blue) with a neighboring ancilla available for use (purple). (top right) The selected ancilla defines a set of accessible wires $w_1, \cdots, w_k$ within the payload circuit. These wires correspond to spacetime locations between gates and will serve as a check's support. (bottom right) We use a decoding heuristic to enumerate multiple valid checks. For each check, Monte Carlo sampling is used to estimate the logical error rate (LER) after postselecting on that check's measurement.
 (bottom left) The check with the lowest LER is implemented into the payload circuit using the chosen ancilla. By confining the check's support to accessible wires, all necessary controlled-Pauli gates remain local in the hardware connectivity.}\label{fig:overview}
\end{figure*}

Quantum computers have advanced to the stage where they can perform computations beyond the reach of exact classical simulations, raising the likelihood of achieving quantum advantage in the near future. However, quantum systems are inherently susceptible to quantum noise, which introduces errors that can compromise computational accuracy. Large-scale quantum computing is untenable without robust techniques to either mitigate or correct these errors. The ultimate objective is to develop a large-scale, error-corrected quantum processor capable of executing significantly larger circuits with reliability. Error mitigation techniques, alongside hardware improvements, have been crucial in managing these errors and enabling progress toward this goal.

Quantum error mitigation (QEM) and quantum error correction (QEC) have different overheads. In QEM, we use little to no extra qubits or gates, but instead take exponentially many samples from the quantum processor in order to counter the effect of noise in postprocessing~\cite{temme2017error, cai2023quantum}. This makes it an appealing approach in the near term due to limitations on the number of available qubits and depth of computation, and especially if the processor can perform fast sampling. On the other extreme, QEC uses a polylogarithmic or even constant qubit overhead~\cite{gottesman2013fault}, but instead does not require any extra samples. However, the qubit overhead is large even for the most efficient error correcting codes~\cite{bravyi2024high}, which severely limits the number of logical qubits that are available in early fault-tolerant quantum computers.

To bridge the gap from error-mitigated to error-corrected quantum computers, it is imperative to consider trade-offs between extra qubits and extra samples. Notably, quantum error \textit{detection} (QED) emerges as a promising approach that can use fewer extra qubits than full-blown error correction, yet achieve \textit{quartically} lower sampling overhead than standard QEM methods such as probabilistic error cancellation~\cite{tsubouchi2023universal, barron2024provable}. What's more, QED gives single-shot access to the quantum state, as opposed to QEM which is limited to expectation values (see Supplementary Information, Section I).

A central problem in designing error detection protocols is that the process of detecting errors can itself introduce new errors. It is therefore critically important to design low-overhead circuits for this task, otherwise both the postselection rate (the fraction of samples that are accepted) and the logical error rate (the fraction of accepted shots that are wrong) will be severely impacted. We use the framework of spacetime codes to achieve this~\cite{bacon_sparse_2017,gottesman2022opportunities,delfosse_spacetime_2023}. The spacetime approach is a circuit-centric approach that considers both spatial and temporal aspects of a circuit to gain flexibility, compared to a code-centric approach where circuit dynamics are often neglected. The spacetime approach is also useful for studying existing codes in a new light~\cite{mcewen2023relaxing, bombin2024unifying}.

Our approach to error detection can be considered a generalization of coherent Pauli checks (CPC)~\cite{roffe_protecting_2018,debroy_extended_2020,van2023single,gonzales_quantum_2023}. By entangling the ``payload'' to some ancillary qubits, these methods are able to ``check'' certain invariants in Clifford circuits, in particular that they map known Paulis to other known Paulis. Any hardware run of the circuit that does not satisfy all checks is flagged as erroneous and discarded. However, these methods are plagued by scalability issues because the measured Paulis are generally high weight, which is only exacerbated when qubit connectivity is constrained. The resulting check circuitry becomes large and noisy, quickly overwhelming the signal it seeks to boost.

Spacetime codes relax the constraint of having time-concentrated and space-delocalized stabilizers by allowing stabilizers to be arbitrarily spread in both space and time. Our central idea is that by employing spacetime checks --- i.e., measuring syndromes associated with a spacetime code --- we can circumvent the need to measure high-weight or non-local operators. After identifying ancilla qubits near the payload, we search for checks supported within a constant spatial radius --- in the extreme case, on the ancilla’s nearest neighbors. Such checks would be concentrated on the same set of qubits but spread over time, and their implementation incurs overhead only proportional to their weight, substantially reducing added noise.

A priori it is not clear that good checks of this form even exist, let alone whether we can find them efficiently. We answer both of these questions affirmatively. The primary contribution of our work is to reduce the problem of finding valid checks located on an arbitrary subset of spacetime coordinates to a well-known linear code decoding problem. This reduction provides us with good heuristics to generate check candidates, which can then be scored using standard simulation techniques. Figure~\ref{fig:overview} illustrates the general approach.

Our method is well-suited to circuits dominated by Clifford gates, although we also discuss extensions to universal circuits composed of Cliffords and Pauli rotations. While purely Clifford circuits are not universal for quantum computation and can be classically simulated~\cite{gottesman_heisenberg_1998}, they appear prominently within quantum algorithms and error correction protocols~\cite{huang2020predicting}, and have played an important role in proposals for quantum advantage~\cite{jozsa2013classical,bouland2017complexity,ghosh2023complexity}. Furthermore, any circuit can be transformed into purely Clifford gates acting on magic state inputs~\cite{gottesman1999demonstrating}. More importantly, the possibility of classically verifying these circuits is useful for building trust in the method and experimental data.

We experimentally demonstrate our construction on IBM Quantum Heron processors by preparing complex stabilizer states and estimating their fidelity. Single-shot access to these states is a useful resource for computation. For example, it enables universal MBQC~\cite{PhysRevLett.97.150504}, and is a candidate for quantum advantage since sampling from their output in a non-Clifford basis is believed to be classically hard~\cite{ghosh2023complexity}. We show that increasing number of checks can successfully increase the fidelity of these states, at the expense of decreasing postselection rates. That is, the quantum computer is slowed down but produces higher quality samples. 

Our experimental results show that we can improve single-shot state fidelity of random stabilizer states on up to 50 logical qubits prepared using 2,450 logical CZ gates. We do so by employing up to 68 physical qubits and 2,718 physical CZ gates. We measure fidelity gains of up to $236\times$ when using spacetime checks compared to the bare circuit. Lastly, we lower-bound the amount of entanglement in our target states by computing the rank-width of the corresponding graph states, and find that they have close to maximal entanglement for those that we can compute in reasonable time.  

Our results indicate that a significant boost in fidelity is possible using sampling and qubit overhead that is orders of magnitude smaller than QEM and QEC respectively. Yet, our protocol is compatible with both due to its single shot nature. This indicates that a serious candidate for near-term quantum advantage can be Clifford-dominated or other circuits for which efficient error detection protocols can be designed.

\section*{Efficient Spacetime checks design}\label{sec:method}

\paragraph{Spacetime codes.} Traditional quantum error correction encodes quantum memory using a stabilizer code and implements fault-tolerant operations on the encoded data. In this framework, the system's state between logical gates satisfies stabilizing constraints defined by the code. This invariant structure can be relaxed by alternating between different codes, thereby modifying the set of invariants after each logical operation. This occurs commonly in lattice surgery~\cite{horsman2012surface} and Floquet codes~\cite{hastings2021dynamically}. Although the invariants change over time, they remain time-local, meaning the valid invariant at any moment can be explicitly identified.

Spacetime codes extend this concept by generalizing these invariants to be spread across both space and time. This approach, as discussed in Refs.~\cite{bacon_sparse_2017, gottesman2022opportunities,delfosse_spacetime_2023}, specifies stabilizers that are not confined to a single moment but are distributed over a temporal and spatial domain. In spacetime circuit codes, syndrome measurements are performed on individual ``wires'' of the circuit, i.e. the segment connecting the output of one gate to the input of another (see Supplementary Information Section II-A). An \( n \)-qubit circuit organized into \( d\) layers of non-overlapping gates features up to \( nd \) wires, representing qubit timelines between each gate layer. Previous work has shown how a code can be constructed from a circuit already equipped with intermediate measurements on its wires~\cite{delfosse_spacetime_2023}. By analyzing the correlations between measurement outcomes, a set of linear equations over $\mathbb{F}_2$ is derived, which should hold true in the absence of noise. Similar to purely spatial codes, any deviation from these equations indicates the presence of errors in the circuit.

Despite this formalism, the problem of how to actually pick the Pauli measurements given a bare circuit has largely been unaddressed. It has also not been clear whether this approach is useful for near-term quantum computing. In this paper we present an efficient approach for constructing good spacetime syndrome measurements given a bare Clifford circuit. 
These syndromes define an \textit{unstructured} error detection code with no particular distance guarantees.
We present algorithms for dressing an arbitrary Clifford payload with Pauli measurements in order to minimize the extra circuit overhead from those Pauli measurements, while still maximizing their error detection capabilities. 

\begin{figure*}[h]
    \centering
    \includegraphics[scale=.5]{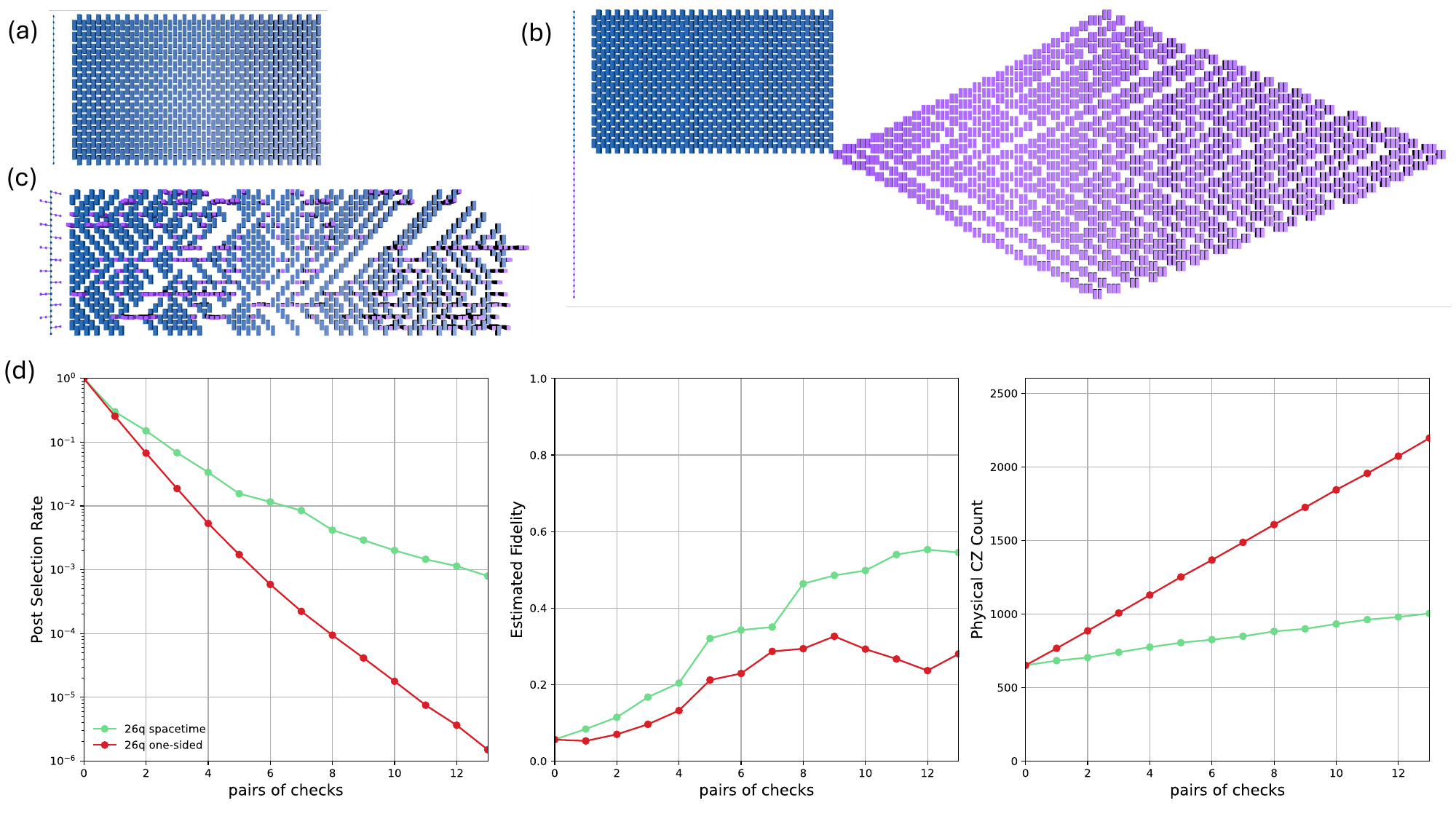}
\caption{\textbf{Comparing error detection in a Clifford circuit using spacetime checks vs. traditional stabilizer checks.} (a) The payload (blue) is a 26-qubit circuit with 52 entangling layers arranged in a linear brickwork pattern, preparing a highly entangled stabilizer state. (b) Errors can be detected by measuring stabilizers of the state at one timeslice, but the overhead (purple) of measuring multiple high-weight Paulis on restricted connectivity can be very large. Pictured are 26 checks. (c) Detecting errors using spacetime codes amounts to “sprinkling” of valid Pauli checks inside the payload, and is vastly more efficient. Pictured are 13 pairs of checks probing the payload from the sides. (d) Noisy simulation showing how improved circuit efficiency impacts error detection. As rounds of checks are added, the circuit footprint in spacetime checks grows much more mildly (right), meaning it removes errors without introducing too many new ones. This has a dramatic effect, improving postselection rate by 3 orders of magnitude (left) and logical error rate by 70\% (middle) in this example.
}
    \label{fig:lr_vs_spacetime}
\end{figure*}

Figure~\ref{fig:lr_vs_spacetime} shows a comparison between our spacetime checks approach vs. previous work that relied on sandwiching the payload between checks in a linear nearest neighbor setting~\cite{van2023single}. The smaller footprint of the spacetime approach is evident, leading to much better performance under similar noise. In this example, checks were picked in a state preparation setting where only right checks are necessary, effectively checking stabilizers of the prepared state.
This difference will only be amplified if we attempt to detect errors on unitaries rather than states, since the one-sided right checks will have to become two-sided left and right checks, effectively doubling the gate and depth overhead.

\paragraph{Spacetime coherent Pauli checks.}
Our construction involves strategically incorporating Pauli operators within the Clifford payload such that the resulting circuit remains equivalent to the original. This process assigns a (possibly trivial) Pauli operator \( P_w \) to each wire \( w \) in the payload (see Figure~\ref{fig:overview}). The goal is to ensure that when these operators are pulled to the front of the circuit, they collectively form the identity:

\begin{align*}
    \prod_w B(P, w) &= I
\end{align*}

Here, \( B(P, w) \) (called the {\it back-propagator}) represents the Pauli operator obtained after moving operator \( P \) from wire \( w \) to the start of the payload by conjugation through preceding Cliffords. This condition can be relaxed to finding operators $P_w$ such that:
\begin{align*}
    \prod_w B(P, w) \in S
\end{align*}
where $S$ is a stabilizer group describing some constraints on the input state. For example, when a circuit is applied to a stabilizer state, $S$ represents the full stabilizer group of that state. If the state is not a stabilizer but possesses certain symmetries, $S$ will be a non-maximal Abelian subgroup of the Pauli group, also known as a code.

In both cases, identifying such a set of Pauli operators provides a direct method for implementing a coherent check: introduce an ancilla qubit in state \( \ket{+} \), apply a controlled-$P$ gate from the ancilla to each wire \( w \) with a non-trivial Pauli \( P \), and finally measure the ancilla in the \( X \) basis. Ideally, the ancilla ends up in a product state with the payload qubits, yielding a deterministic measurement. However, in the presence of noise, this invariant is disrupted through phase kickback onto the ancilla, leading to random measurement outcomes. For Pauli noise, the {\it back-cumulant} associated with a check~\cite{delfosse_spacetime_2023} can precisely characterize the set of Pauli errors that flip the measurement outcomes and can thus be detected. The back-cumulant of a Pauli $P$ on wire $w$ is a Pauli operator supported across all wires of the circuit, corresponding to the trace obtained by pulling $P_w$ to the front of the circuit. This operator fully captures how a Pauli operator propagates through various spacetime locations (or wires) within the circuit (see Supplementary Information Section II for more details). By defining the back-cumulant of a check as the back-cumulant of a Pauli $Z$ operator located on the output of the ancilla qubit, we can precisely characterize the set of Pauli errors that will anti-commute with the ancilla measurement and thus be detected by the check. 

Figure~\ref{fig:rendering} shows a visualization of circuits protected with a spacetime check. Intuitively, we can probe into specific points in the spacetime volume of a quantum circuit using Pauli measurements that are spatially local yet temporally distributed. Errors happening at various points in this spacetime may trigger one or more of these checks and be discarded. This approach is especially suited for planar nearest-neighbor architectures, because the Pauli checks remain low-overhead. Instead of reaching deep into the circuit (and potentially adding more errors that they can detect), these checks rely on the circuit itself to propagate errors to the detectors. 

\begin{figure}[h]
    \centering
    \includegraphics[scale=0.45]{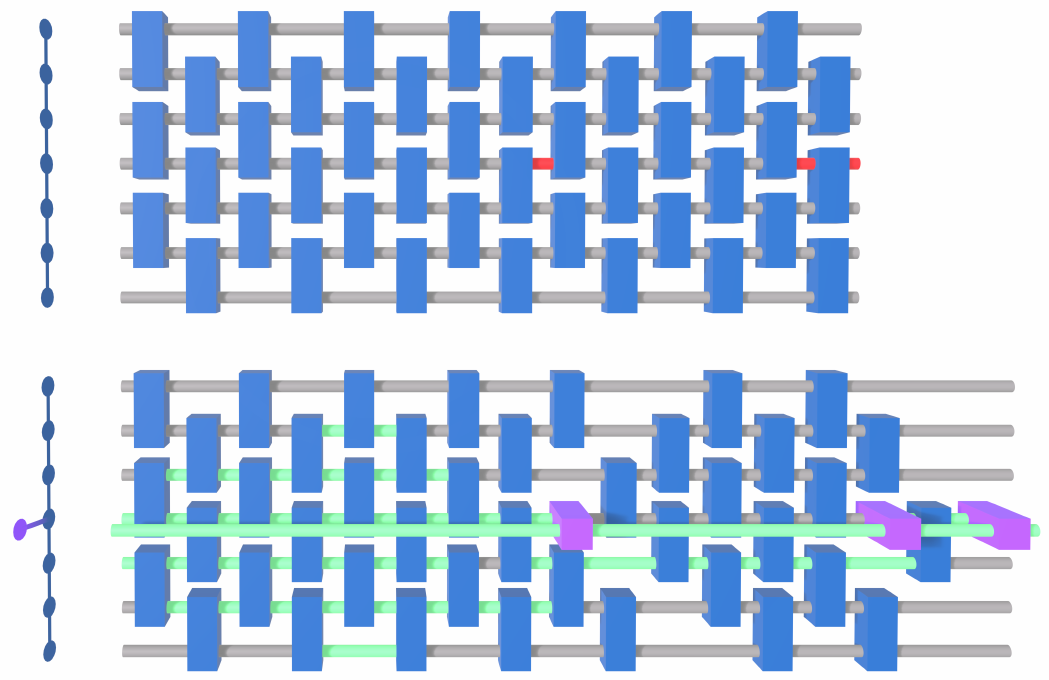}
    \caption{\textbf{3D rendering of a single spacetime check for qubits on a 2D lattice and time going left to right.} (top) the bare Clifford circuit with a valid check's support highlighted in red. Only two-qubit gates are shown. (bottom) the actual implementation of the check using an adjacent ancilla. Each red wire now receives a controlled-Pauli gate from the ancilla. Wires highlighted in green are the detecting region. They represent the non-trivial support of the check's back-cumulant, which is the operator obtained by pulling the final $Z$ measurement of the ancilla backwards and tracing the resulting Pauli path through the circuit. A low-overhead check yields a large detecting region.}
    \label{fig:rendering}
\end{figure}

Using adjoint-based code simulations \cite{delfosse_simulation_2023,gidney2021stim}, we can efficiently estimate the logical error rate after postselecting on the joint $0$ outcomes of a set of checks. These methods involve preprocessing a circuit with syndrome measurements by computing the back-cumulants of each measurement. This preprocessing enables efficient computation of syndrome outcomes for a given error, thereby accelerating any Monte Carlo estimation of the error detection rate or logical error rate.
Moreover, this technique allows us to score checks and compare their effectiveness (see Supplementary Information Section III). 

The success of this approach hinges on our ability to design checks that detect more errors than they introduce. In most quantum hardware, implementing a check incurs a significant entangling gate overhead due to limited qubit connectivity. To address this challenge, we must focus on optimizing the design of checks to minimize the entangling gate overhead, which leads us to the problem of finding valid checks within a subset of accessible wires in the payload: given a subset of wires \( L \) in the input payload (typically the set of spacetime wires accessible to an ancilla), find a valid and low-weight check whose support is within \( L \). This problem can be reduced to a decoding problem for a linear code, for which effective heuristics exist (see Supplementary Information Section IV).

Our framework also naturally extends to Clifford + Pauli rotation circuits. In this extension, one can formalize additional constraints on valid checks to ensure that adding the check to the payload will not interfere with the rotations. These constraints can be expressed as additional linear equations within our decoding framework, and solved efficiently in the same way. Naturally, those additional constraints further restrict the set of possible checks, leading to an exponential reduction in the size of the group of valid checks (see Supplementary Information Section V). This suggests that, although non-Clifford circuits are not fundamentally precluded, our method is most effective when applied to circuit segments that are Clifford or predominantly Clifford. Several classically hard circuit families of this form are known~\cite{jozsa2013classical,bouland2017complexity,ghosh2023complexity}.

Our approach can be summarized as follows (see Figure \ref{fig:overview}):

\begin{enumerate}
    \item Transpile the target payload onto a subset of qubits on the device.
    \item Identify unused ancilla adjacent to one or more payload qubits, and determine the set of payload wires accessible to each ancilla.
    \item Use a decoding heuristic to find multiple valid checks located on these wires.
    \item Score each check with Monte Carlo simulation. Typically, a good score will be inversely proportional to the logical error rate after postselection.
    \item Implement the highest-scoring check and repeat from step 2.
\end{enumerate}

The key ingredient of step 4 is to generate the full circuit, including the check we want to score and the previously picked checks, and only then estimate the logical error rate after postselection. This ensures that the gate overhead is taken into account when picking a check. In particular, we can also use multiple neighboring checks in such a way that they detect each other's errors (see Supplementary Information Section VI).   By following this structured approach, we can systematically identify and implement effective checks, optimizing their design to minimize entangling gate overhead while maximizing detection coverage, all achieved through a polynomial-time heuristic.



\section*{Preparing highly-entangled stabilizer states}\label{sec:exp}

\begin{figure*}[h]
    \centering
    \includegraphics[scale=.5]{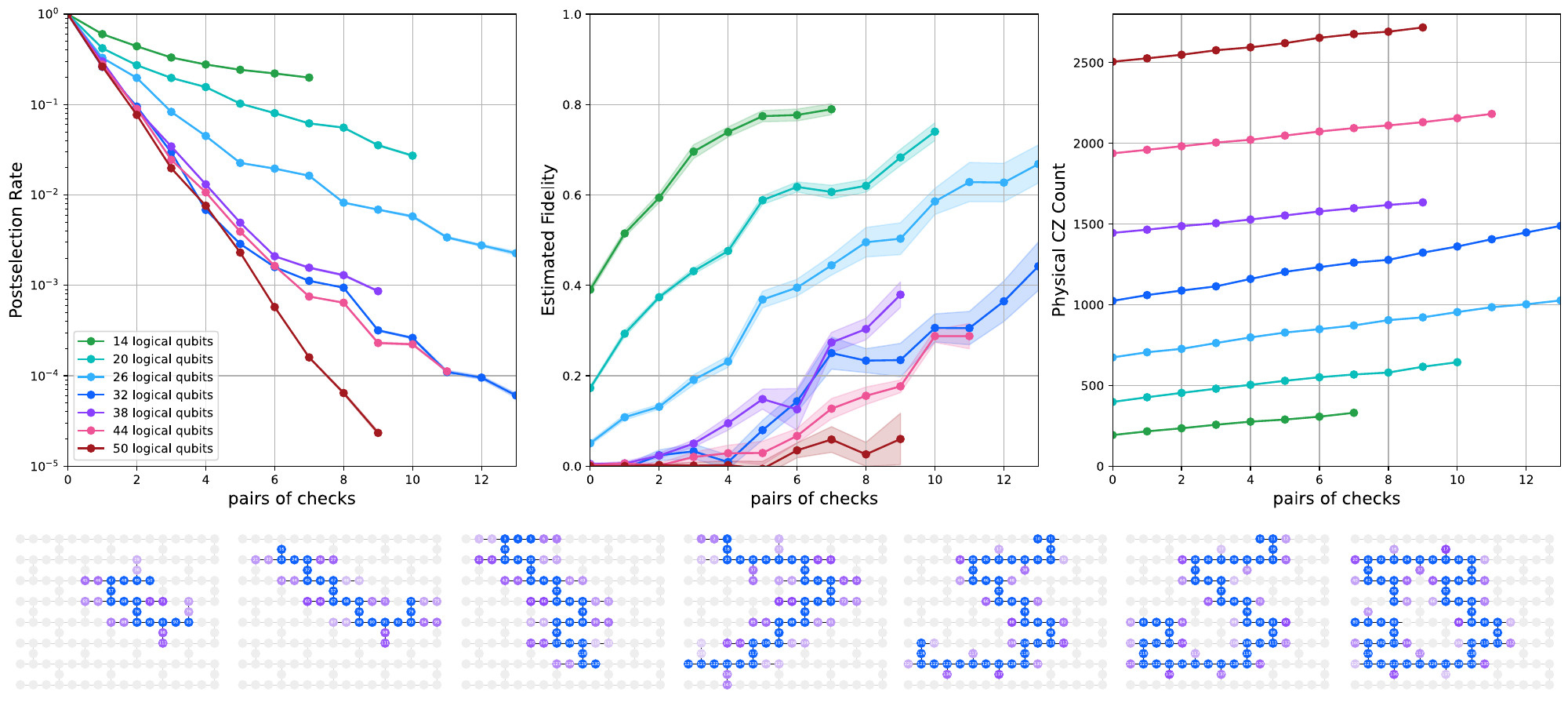}
    \caption{\textbf{Preparing high-fidelity and highly-entangled stabilizer states on superconducting qubits by error detection.} As more checks are added, the logical state fidelity goes up, at the expense of lower postselection rate. The circuits prepare a random stabilizer state over $n$ logical qubits with an entangling depth of $2n$, enough to saturate the entanglement upper bound. The number of logical gates is the number of physical gates at round $0$. From there, each round of check adds $2$ physical qubits and some number of physical gates, shown by the slopes in the right plot. 
    Fidelities are estimated by sampling a constant number of stabilizers, in this case 5. 
    Shaded regions indicate the standard error of the mean ($68\%$ confidence intervals).
    The layout of qubits on the physical hardware (\textit{ibm\_kingston}) are also depicted. The payload is supported on the blue qubits and checks on the purple qubits. Darker to lighter shades indicate the order in which checks are added.
    }
    \label{fig:expt-paths}
\end{figure*}

\begin{table*}[ht]
\centering
\begin{tabular}{|c|c|c|c|c|c|c|c|c|c|}
\hline
\multicolumn{4}{|c|}{\textbf{Bare circuit}} & \multicolumn{2}{c|}{\textbf{Sampling overhead}} & \multicolumn{4}{c|}{\textbf{This work}} \\ \hline
\textbf{Qubits} & \textbf{\makecell{Gates}} & \textbf{\makecell{Ent.\\width}} & \textbf{\makecell{LF}} & \textbf{PS} & \textbf{PEC} & \textbf{\makecell{Qubit\\overhead}} & \textbf{\makecell{Gate\\overhead}} & \textbf{\makecell{Sampling\\overhead}} & \textbf{\makecell{Fidelity\\gain}} \\ \hline
14   & 182  & $[4,5]$ & 0.960    & \num{2.9e0}  & \num{6.9e+01} & 14 & 148  & \num{5.0}   & $2\times$     \\ \hline
20   & 380  & $[5,6]$ & 0.937  & \num{1.4e1}  & \num{3.3e+04} & 20 & 263  & \num{3.7e1}  & $4\times$     \\ \hline
26   & 650  & $[7,9]$ & 0.908  & \num{1.6e2}  & \num{5.8e+08} & 26 & 376 & \num{4.4e2} & $13\times$    \\ \hline
32   & 992  & $[9,11]$ & 0.869  & \num{8.2e3}  & \num{4.4e15}  & 26 & 496 & \num{1.7e4} & $109\times$   \\ \hline
38   & 1406 & $[10,13]$ & 0.884  & \num{1.2e4}  & \num{1.9e+16} & 18 & 228 & \num{1.2e3} & $83\times$    \\ \hline
44   & 1892 & $[9,15]$ & 0.875  & \num{1.3e5}  & \num{2.8e+18} & 22 & 288 & \num{8.9e3} & $104\times$   \\ \hline
50   & 2450 & $[8,17]$ & 0.870  & \num{1.2e6}  & \num{1.9e+24} & 18 & 268 & \num{4.3e4}  & $236\times$    \\ \hline
\end{tabular}
\caption{\textbf{Summary of experiments performed on $ibm\_kingston$}. We characterized layer fidelity (LF) for one layer of each circuit, which (due to the brickwork nature of the circuit) lets us infer the sampling overhead needed to remove all errors using either idealized postselection (PS) or probabilistic error cancellation (PEC). ``Gates'' refers to number of CZ gates. Sampling overhead for error detection is the inverse of postselection rate. The difficulty of tensor-based simulation is exponential in the entanglement width of the prepared state, for which we show known lower and upper bounds. The fidelity gains are obtained using orders of magnitude lower sampling overhead and qubit overhead than required by PEC and QEC respectively.}\label{tab:experiments}

\end{table*}

We performed experiments on random Clifford circuits of depth $2n$ with number of logical qubits ranging from $n=14$ to $n=50$, arranged on linear nearest-neighbor connectivity (a path of qubits). Such circuits are capable of preparing any stabilizer state~\cite{maslov2018shorter}. Consequently, the resulting states exhibit maximal entanglement, as measured by the entanglement width~\cite{PhysRevLett.97.150504}. Nevertheless, the experimental fidelity remains efficiently verifiable (see Supplementary Information Section VII). For comparison, a tensor-network-based simulation of a stabilizer state preparation followed by local unitaries would require approximately $(n-1)2^{5r}$ operations, where $r$ is the rank-width of the corresponding graph state (see \cite{ghosh2023complexity} for more details). In our setting, $r$ is maximal and close to $n/3$. For our 50-qubit experiment, this would require approximately $6 \times 10^{26}$ FLOPS to compute a single amplitude, making the entanglement generated in our experiment non-trivial to simulate using tensor-network-based methods. Similarly, the same task could be tackled by near-Clifford simulation techniques in time $O(n^2 2^{\alpha\cdot c\cdot n})$ with $\alpha $ as low as $0.3963$ (asymptotically) if the local unitaries can be implemented with a constant number $c$ of $T$ gates \cite{Kissinger_2022}. With $c=2$, this induces a complexity of roughly $10^{16}$ FLOPS, and pushing to $c=4$ would induce a complexity of $10^{27}$ FLOPS, similar to that of a tensor-based approach.

Figure~\ref{fig:expt-paths} shows the experimental results and qubit layouts on an IBM Heron quantum processor. We estimate the fidelity of stabilizer states by measuring in random stabilizer basis~\cite{PhysRevLett.106.230501} and postselecting on shots with trivial syndromes. We introduce a method that lets us certify single-shot state fidelity with reduced sensitivity to readout errors. Immediately before measurement, we fold the stabilizer onto a smaller set of qubits with the highest readout fidelity. This tradeoff is favorable when gates have lower errors than readouts, as is the case in our experiments. Importantly, all checks are performed prior to the folding gates, ensuring that the resulting fidelity is a true lower bound. In all experiments, we find that increasing rounds of checks successfully detect errors and discard bad shots, a result made possible by the mild qubit and gate overhead of spacetime checks. See Supplementary Information Sections VIII and IX for additional numerical and experimental data.

Table~\ref{tab:experiments} summarizes our experimental results including the overheads associated with each error detection experiment. One can better understand the sampling overheads by contrasting them with PEC on one hand and ideal post selection on another. Both of these can be found from the circuit's layer fidelity (LF)~\cite{mckay2023benchmarking, barron2024provable} and are summarized in the Table.  It confirms the fact that spacetime checks have far lower qubit overhead than error correction and far lower sampling overhead than PEC.
However, it is important to keep in mind that unlike error correction or PEC, our method is a heuristic one which does not catch all errors and is thus not an unbiased estimator. Nevertheless, a single-shot protocol such as this can be used in conjunction with other QEM or QEC techniques.

\section*{Discussion}

We have introduced a low-overhead heuristic for detecting errors in Clifford and near-Clifford circuits, and used it to experimentally prepare some of the most highly entangled quantum states reported to date. 


A more traditional approach to quantum error detection uses error detecting codes. These are codes that can detect but not correct errors. Examples include the $[[4, 2, 2]]$ and $[[8, 3, 2]]$ code, which have been demonstrated in various experiments~\cite{harper2019fault, gupta2024encoding, Bluvstein_2023, reichardt2024demonstration, reichardt2024logical, self2024protecting}. While these codes use fault-tolerant gadgets and offer detection guarantees, they have limitations regarding the set of logical gates that can be fault-tolerantly implemented. For example the $[[4, 2, 2]]$ code struggles to efficiently implement phase ($S$) gates and $CZ$ between arbitrary qubits in different code blocks. These limitations render them unsuitable for running complex circuits such as those in our experiments, unless one accepts using a single qubit per code block, which results in a $4\times$ space overhead. Even in that case, the circuit would need to be rewritten using the $H+CZ+SWAP$ gate set, significantly increasing the count and depth of entangling gates. In contrast, spacetime codes encode the entire circuit and not individual blocks of qubits without regard to the computation performed.


The obvious downside of error detection is that it can discard too many shots. As we discussed, an exponential slowdown in the rate of computing is unavoidable in the absence of error correction. However, we argue that through optimizations such as those presented here, coupled with fast quantum hardware, classically challenging computations are foreseeable. In particular, in these experiments we performed error detection down to postselection rates of $10^{-5}$ and saw fidelities continuing to improve. On IBM's superconducting processors, $10^5$ samples correspond to approximately 25 seconds of machine time, allowing the production of low-error samples at a rate that is comparable with the raw sampling rate of other technologies such as neutral atoms \cite{Bluvstein_2023, PhysRevResearch.6.043020} or trapped ions~\cite{decross2024computational}.

Several intriguing avenues remain open for further investigation. Although our current framework focuses on error detection, a natural extension of this work would be to investigate the requirements and modifications needed to transition to error correction. This would involve developing strategies to not only detect the presence of errors but also to localize them~\cite{gottesman2022opportunities}, which may be approached from the perspective of flag fault tolerance~\cite{chao2018quantum, anker2024flag}.
Akin to assertions in classical computing, error detection verifies certain invariants in quantum programs~\cite{li2020projection}, albeit in a more sophisticated way by revealing syndromes and not the quantum state itself. It can be interesting to look beyond Clifford circuits and for other symmteries that can be checked~\cite{tsubouchi2025symmetric, bonet2018low}. Future research could also explore more precise and thorough code design tailored to specific structured circuits, such as IQP circuits or circuits with fixed structures. This could yield even greater efficiency.


\paragraph{Data availability} All circuits used in this work and corresponding experimental data can be accessed at the following link: 
\href{github.com/ajavadia/spacetime-checks}{github.com/ajavadia/spacetime-checks}

\section*{Acknowledgments}
We are grateful to Maika Takita and Abhinav Kandala for their insights on the experiments, and acknowledge valuable discussions with Ted Yoder, Ewout van den Berg, Nobuyuki Yoshioka, Andrew Cross, and Sergey Bravyi. We thank Alireza Seif for feedback on an earlier draft, and thank Sarah Sheldon and Jay Gambetta for their support throughout.
This work was partially supported by the U.S. Department of Energy, Office of Science, National Quantum Information Science
Research Centers, Co-design Center for Quantum Advantage (C2QA) under contract number DE-SC0012704.

\bibliography{biblio}
\bibliographystyle{naturemag}

\end{document}


\maketitle
\tableofcontents

\section{Relation between error detection and error mitigation/correction}

In this section we compare error detection to error mitigation and correction in terms of various parameters that affect their applicability and overhead. Table~\ref{tab:comparison} summarizes this discussion.

Consider $F$ to be the fidelity of a quantum circuit when executed on a noisy quantum computer. One can efficiently estimate $F$ for a Clifford circuit~\cite{PhysRevLett.106.230501}, but in many instances for non-Clifford circuits as well~\cite{merkel2025clifford}. Conceptually, $F$ can be thought of as the probability of no error happening in an execution of the circuit (a ``shot''), and it falls exponentially as the circuit volume grows. In the absence of active error correction (i.e. extracting error syndromes during the circuit and correcting them before they can accumulate), there is a fundamental $1/F$ lower bound to the sampling overhead: one needs that many more shots to obtain the same number of good samples as a noiseless circuit~\cite{dutkiewicz2024error,barron2024provable}. In contrast, standard error mitigation methods that invert the noise have a sampling overhead lower bounded by $1/F^2$. This lower bound is achievable in certain settings, such as under global depolarizing noise~\cite{tsubouchi2023universal} or by utilizing large classical resources~\cite{filippov2023scalable}.

Perhaps more importantly, error detection gives ``single-shot'' access to the output of a quantum circuit. A defining feature of all QEM strategies, instead, is that they reconstruct the average of an observable in post-processing. That is, the quantum computer never actually prepares the noise-free state. This is the case for example in probabilistic error cancellation, zero-noise extrapolation~\cite{temme2017error}, tensor error mitigation~\cite{filippov2023scalable} and circuit cutting~\cite{peng2020simulating}.
However, single-shot access to quantum states is important for many quantum algorithms that work with samples, such as Shor's, Grover's, QAOA, and measurement-based quantum computing (MBQC). Single-shot sampling of quantum states is also a cornerstone of many hardness results in classical simulation of quantum circuits~\cite{bremner2017achieving, bouland2018quantum}, whereas expectation values appear easier for classical computers to simulate~\cite{bravyi2021classical}.

Error correction is in principle possible with constant overhead~\cite{gottesman2013fault}, but the barrier of entry can be quite high. Even the most efficient codes require 1-2 orders of magnitude more qubits at non-trivial code distances~\cite{bravyi2024high}, and efficient codes typically have more demanding hardware requirements~\cite{bravyi2010tradeoffs}. Error detection however can be achieved with fewer qubits as we have demonstrated, because it only needs to detect the presence of an error but not resolve its type or location. However once hardware capabilities reach that point, large-scale fault tolerant quantum computing is expected to achieve much larger universal quantum computations.

It must be noted however that full PEC and high-distance QEC practically remove all errors from the computation. While this is also true for error detection in the ideal case, the heuristics we have presented in this paper do not detect all errors (for example errors that occur towards the end of the computation will not trigger the checks and will go undetected). It is therefore important to keep in mind that current heuristics provide a low-overhead boost to the signal, and they must be combined with other error mitigation approaches to obtain unbiased estimates of a circuit's observables.

\begin{table}[h!]
\centering
\begin{tabular}{|c|c|c|c|}
\hline
\textbf{Method} & \textbf{Sampling overhead} & \textbf{Qubit overhead} & \textbf{Single shot} \\ \hline
Error mitigation / PEC & $\Omega(1/F^2)$ / $\Omega(1/F^4)$ & 0 & $\times$  \\ \hline

Error detection & $\Omega(1/F)$ & small constant & \checkmark \\ \hline

Error correction & $\Omega(1)$ & large constant & \checkmark \\ \hline
\end{tabular}
\caption{Comparing the properties and overhead of error detection vs. error mitigation and correction}
\label{tab:comparison}
\end{table}

\section{Spacetime Pauli checks}
This section expands on the concept of spacetime coherent Pauli checks within Clifford circuits, detailing how these checks are implemented to detect errors by non-destructively measuring Pauli operators at specific spacetime locations in the circuit. 

\subsection{Circuit wires as spacetime coordinates}

The main idea behind spacetime codes is first to notice that a circuit, described as a sequence of gates acting on a set of qubits, naturally induces a set of space and time coordinates. Literature slightly differs on the details of this construction (see \cite{gottesman2022opportunities, bacon_sparse_2017, delfosse_spacetime_2023}). We propose here yet another construction that is built directly on the circuit's wires, when expressed as a directed acyclic graph (DAG) of gates.

\paragraph{Circuit DAG and coordinates.} Let $C$ be a circuit over $n$ qubits. Its DAG $\mathcal{G}(C)$ has vertices for the unitary gates of $C$, and edges for the input/output relations between these gates. Moreover, we extend this DAG with $n$ dummy input gates and $n$ dummy output gates, one for each qubit in the circuit (see Figure~\ref{fig:dag}). Note that if the circuit is correctly formed, the resulting graph is indeed acyclic.
Consider now the line graph $\mathcal{G}(C)^\star$ of  $\mathcal{G}(C)$. $\mathcal{G}(C)^\star$ has the edges of $\mathcal{G}(C)$ as a set of vertices. Two vertices $uv$, $wx$ in 
$\mathcal{G}(C)^\star$ are connected if and only if $v = w$. In other words, edges of $\mathcal{G}(C)^\star$ are exactly the directed paths of length $2$ in $\mathcal{G}(C)$. 
Note again that the graph $\mathcal{G}(C)^\star$ is itself acyclic.
As such, it induces a partial order $\leq_C$ on its set of vertices, defined as $ e_1 \leq_C e_2$ if and only if there exists some directed path from $e_1$ to $e_2$ in $\mathcal{G}(C)^\star$.
In this work, we call the vertices of $\mathcal{G}(C)^\star$ wires. This set of wires will be denoted $\mathcal{W}_C$ or $\mathcal{W}$ when $C$ is clear from context. They will be our set of spacetime coordinates.
The relation $\leq_C$ lets us naturally decide if two wires are causally related or independent. In particular, a maximal set of causally independent coordinates will be referred to as a \textit{spacelike}. When $C$ is clear from context we might use $\leq$ to refer to $\leq_C$.

\paragraph{Spacetime Hilbert space.} Each wire carries the state of a qubit between two gates in the circuit. We can thus make an abstraction of the relations between those states and define a global spacetime Hilbert space $\mathcal{H}(C)$ composed of qubit spaces attached to each wire: $$\mathcal{H}(C) = \bigotimes_{w\in\mathcal{W}(C)} \mathbb{C}^2$$
We will later introduce Pauli operators defined over this Hilbert space. Let $P$ be some single-qubit Pauli operator and $w\in \mathcal{W}(C)$. We denote by $P_w$ the Pauli operator $P\otimes \left(\bigotimes_{x\in\mathcal{W}, x\neq w} I\right)$, or, in other words, the Pauli operator acting as $P$ on qubit $w$ and as the identity on every other qubit. To simplify notations in some settings, we might interchangeably use $(P, w)$ or $P_w$ to refer to this same operator.

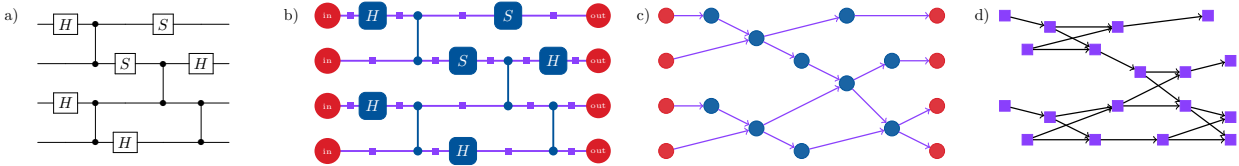
\begin{figure}[h]
    \centering
    \hspace{-4.5cm}\scalebox{0.6}{
        \input{figures/spacetime}
    }
    \caption{\textbf{Wires as spacetime coordinates of a Clifford circuit.} (a) A quantum circuit $C$. (b) This same circuit with $8$ additional dummy gates representing inputs and outputs (in red). Wires are highlighted in purple. (c) The underlying directed acyclic graph (DAG) $\mathcal{G}(C)$. Its set of vertices is composed of all the gates, including the dummy gates. Its edges are exactly the input/output relations between the circuit's gates. (d) The line graph $\mathcal{G}(C)^\star$ of the underlying DAG. This graph is another DAG. Its vertices are the edges of $\mathcal{G}(C)$. Two vertices are connected if and only if the corresponding edges form a directed path of length 2 in $\mathcal{G}(C)$.}
    \label{fig:dag}
\end{figure}

\subsection{Spacetime checks}

\begin{figure}[h]
    \centering
    \scalebox{0.4}{
        \input{figures/backpropagator}
    }
    \caption{\textbf{Construction of the back-propagator \( B(P, w) \)}. Wire \( w \) determines a bi-partition of the circuit \( U \) into \( U = B_w A_w \). The back-propagator is a Pauli operator obtained from ``pulling'' \( P\) backwards through \(A_w\).}
    \label{fig:backprop}
\end{figure}
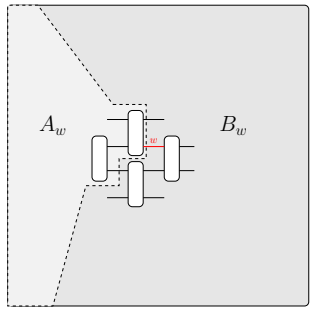

Consider the coherent Pauli check framework~\cite{debroy_extended_2020, gonzales_quantum_2023, van2023single,tsubouchi2025symmetric}, where a Clifford circuit \( C \) is checked by left and right Pauli measurements \( P \) and \( P' \). When \( C.P.C^\dagger = P' \), the measurements are expected to agree. Disagreement indicates a logical error, which can be detected. Since adding many noisy checks in this way can be expensive, we seek to generalize this framework to insert checks {\it inside} the payload as well. We also wish to reduce the cost of checks when the payload is applied to a set of qubits in a stabilizer state.


Let \( C \) be a Clifford circuit applied to a stabilized state \( \ket{\psi} \) with stabilizer group \( S \). Let the set of wires \( \mathcal{W} \) of \( C \) be the spacetime locations of \( C \) in-between its gates, and consider a wire \( w \) in \( \mathcal{W} \). We decompose \( C \) as \( C = B_w A_w \), where \( A_w \) and \( B_w \) are Clifford circuits containing gates before and after \( w \), respectively (see Figure \ref{fig:backprop}). In practice, this bipartition $A_w,B_w$ of the circuit can be directly inferred by using a topological ordering of $\mathcal{G}(C)^\star$. Taking all the gates adjacent to wires $x$ such that $x$ lies before $w$ in a topological ordering will produce a valid $A_w$. 
Let $P\in\{X, Y, Z\}$, $w\in \mathcal{W}(C)$. We define the back-propagator \( B(P, w) \) as the Pauli operator \( A_w^\dag P A_w \) for \( P \in \{X, Y, Z\} \). A valid check is a collection of Paulis on some wires \( \{(P_1, w_1), \ldots, (P_k, w_k)\} \) such that:

\begin{align} 
    \prod_i B(P_i, w_i) &= I \label{eq:reduces_to_i}
\end{align}

up to a global phase\footnote{See Section~\ref{subsec:phase}}. This condition can be relaxed to:

\begin{align}
    \prod_i B(P_i, w_i) \in S\label{eq:reduces_to_s}
\end{align}
in the case where the stabilizer group $S$ is non-trivial.

This formulation encompasses the left-right setting and can be expressed at any spacelike slice in the circuit.
In other words, spacetime checks are a set of Pauli operators spread throughout a Clifford circuit and such that their overall contribution cancels out (up to some stabilizer group $S$).

\subsection{On global phases of checks}\label{subsec:phase}

It is convenient to restrict the set of injected Paulis to $\{X, Y, Z\}$. However, by doing so, one might induce a global phase on the circuit when injecting the check. In other words, one needs to relax Eq.~\ref{eq:reduces_to_i} and Eq.~\ref{eq:reduces_to_s} to:
\begin{align} 
    \prod_i B(P_i, w_i) &\in \{\pm I, \pm i I\}\label{eq:reduces_to_i_with_phase}
\end{align}
and
\begin{align}
    \prod_i B(P_i, w_i) \in \pm S \cup \pm i S\label{eq:reduces_to_s_with_phase}
\end{align}

In practice, this global phase is ignored when picking checks and only considered when the check is actually implemented in the circuit, because at that point the controlled-Pauli turns it into a relative phase which has an observable effect. At that stage, we infer the phase by computing the left hand side of Eq.~\ref{eq:reduces_to_i_with_phase} (or Eq.~\ref{eq:reduces_to_s_with_phase}). The phase is then compensated by injecting a single qubit phase gate on the ancilla's wire (see Figure~\ref{fig:phase} for an example).

\begin{figure}[h]
    \centering
   \[  \Qcircuit @C=1.0em @R=1.2em @!R { 
       a)&   & \gate{Y} & \gate{H} & \gate{Y} &\qw \\
        b)  && \gate{-I} & \gate{H} & \qw &\qw
     }  ~~~~ \Qcircuit @C=1.0em @R=1.2em @!R { 
     c)  &         & \qw     & \gate{Y}  & \gate{H} & \gate{Y}  &  \qw        & \qw&\\
       & \gate{H}&\gate{Z} & \ctrl{-1} & \qw      & \ctrl{-1} &  \gate{H}   & \qw&
     } 
     \]
    \caption{\textbf{Phase-correct implementation of checks}. (a-b) Example global phase induced by a check. (c) This global phase introduces a relative phase when implementing the check, which must be compensated by adding a Clifford gate in $\{Z, S, S^\dagger\}$ on the ancilla.}
    \label{fig:phase}
\end{figure}
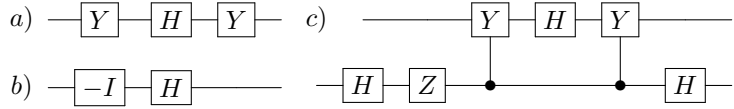

\section{Computing check coverage} \label{sec:cumulant}
Here we detail our methodology for evaluating the effectiveness of spacetime checks in detecting errors. We introduce the concept of a check's back-cumulant~\cite{delfosse_spacetime_2023}, which captures the error detection capability of a check, and describe how Monte Carlo sampling is used to estimate the logical error rate of a checked payload.

\subsection{Check back-cumulant}
For simplicity, let us assume that the circuit undergoes Pauli noise located on each wire in between gates. More general noise channels can be expressed as Pauli noise after twirling; nevertheless, we find Pauli noise approximates experimental data well (see Section~\ref{sec:sim}). Let \( \mathcal{E}_x \) be the noise channel on wire \( x \). 
A single Pauli error \( P \) from \( \mathcal{E}_x \) is detected by a check \( \{(Q_1, w_1), \ldots, (Q_k, w_k)\} \) if and only if \( |\{i \in [k], x \leq w_i \text{ and } [B(Q_i, w_i), B(P, x)] \neq 0\}| \) is odd. 
Similarly, multiple errors \( P_1, \ldots, P_l \) occurring in channels $\mathcal{E}_{x_1}, ..., \mathcal{E}_{x_l}$ are detected if and only if \( |\{(i, j) \in [k]\times [l], x_j \leq w_i \text{ and } [B(Q_i, w_i), B(P_j, x_j)] \neq 0\}| \) is odd. This condition utilizes the back-propagator operator to pull checks and errors to the beginning of the circuit and express their (anti-)commutativity in that frame. It is possible to simplify this expression by considering an operator called the \textit{check's back-cumulant}, as introduced in Ref.~\cite{delfosse_spacetime_2023} (also related to the ``spackle'' operator in Ref.~\cite{bacon_sparse_2017} and ``detecting regions'' in Ref.~\cite{mcewen2023relaxing}).

For a single Pauli \( (P, w) \), its back-cumulant \( \overleftarrow{B}(P, w) \) is the Pauli operator over $\mathcal{H}(C)$ obtained by propagating \( P_w \) to the front of the circuit, leaving a trace throughout the circuit's wires. Contrary to the back-propagator $B(P, w)$ which acts on a set of $n$ qubits, this operator is supported on the full set of wires $\mathcal{W}$. The back-cumulant of a check \( C = \{(P_1, w_1), \ldots, (P_k, w_k)\} \) is:

\begin{align*}
    \overleftarrow{B}(C) = \prod_i \overleftarrow{B}(P_i, w_i)
\end{align*}

An error \( P \) on wire \( w \) is caught by \( C \) if and only if \( [\overleftarrow{B}(C), P_w] \neq 0 \).


Using the back-cumulant, we can determine if a set of physical errors induces a logical error. For our experiment, a logical error is any final Pauli error outside the expected stabilizer group. Given errors \( P_1, \ldots, P_n \), we check if there exists a stabilizer \( S \in \mathcal{S}_{final} \) such that \( [\overleftarrow{B}(S), \prod (P_i, w_i)] \neq 0 \).

In practice, we compute the back-cumulant by first building the circuit that implements the check, then computing the back-cumulant of a single \( Z \) operator on the check qubit's output wire, corresponding to the check qubit's measurement. This requires conjugating a Pauli operator through all gates, which can be done in \( O(nm) \) time, where \( n \) and \( m \) are the number of qubits and gates, respectively. Building the check circuit first ensures that we also consider the extra gates and wires that come from the checks themselves.

Figure~\ref{fig:full_example} shows an example of how a low-overhead check can detect errors on a large region of the circuit. 
Counterintuitively, our spacetime check design excels when applied to deep circuits. For a space-local check to effectively detect errors across a significant portion of the payload, the payload must be sufficiently deep to ``spread'' the check's back-cumulant over a large number of internal wires. For shorter circuits, we can increase the radius of checks to achieve better coverage, but this must be balanced against the overhead of probing deeper into the circuit.

\begin{figure}[h]
    \centering
    \hspace*{-3cm} 
    \includegraphics[scale=0.85]{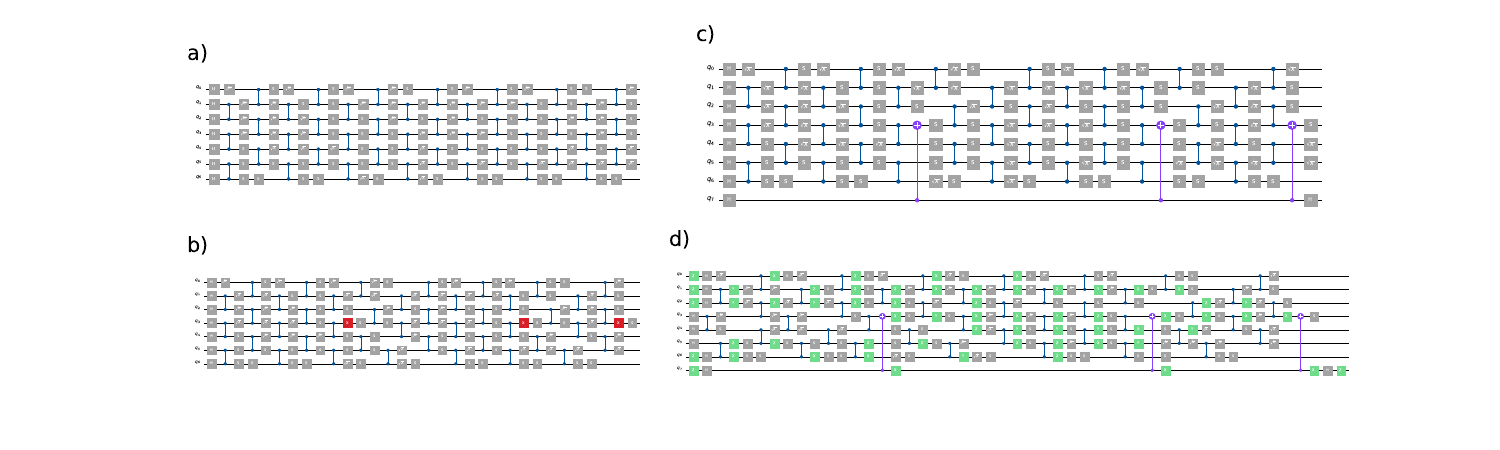}
    \caption{\textbf{Large circuit coverage with a low-overhead check.} (a) An example circuit with entangling gates highlighted. (b) Circuit with highlighted Pauli operators inside it, forming a valid check of the circuit. (c) Check circuitry added: an ancilla in \( \ket{+} \) state is entangled to the circuit according to the valid check, and is measured in the \( X \) basis. (d) Check's back-cumulant highlighted in green. The check will detect any fault that anti-commutes with its back-cumulant.}
    \label{fig:full_example}
\end{figure}

\subsection{Estimating logical error rate} \label{sec:coverage}

In order to know how effective any particular check is and to pick good checks, we need to be able to estimate the logical error rate of a checked payload. To do so we can Monte Carlo sample the sum:

\begin{align}
   LER(C_1, \ldots, C_l) &= \sum_{k=1}^{\infty} \sum_{\substack{E=(P_1, w_1)\ldots(P_k, w_k)\\\textrm{not detected}\\\textrm{logical error}}} \mathbb{P}(E) \\
   &= \sum_{k=1}^{\infty} \sum_{\substack{E=(P_1, w_1)\ldots(P_k, w_k)\\\textrm{not detected}\\\textrm{logical error}}} \prod_{i=1}^k  \mathbb{P}((P_i, w_i))
   \label{eq:coverage}
\end{align}

Physical errors can be efficiently drawn according to $\mathbb{P}(E)$ by going through all wires $w$ and drawing a Pauli according to the error rate of the channel attached to it.
The overall error is then discarded if it does not trigger a logical error or if it is caught by a check. We accumulate the probabilities of non discarded errors and renormalize this quantity by the total sum of probabilities of all sampled errors.

This quantity serves as a scoring function for checks. Given a circuit with checks \( C_1, \ldots, C_l \), a new check \( C_{l+1} \) will have score \( -LER(C_1, \ldots, C_l, C_{l+1}) \).

\paragraph{Error Model.}

To evaluate Eq.~\ref{eq:coverage}, we need a model for \( \mathbb{P}((P, w)) \). We fix a base error probability \( \varepsilon = 8e^{-4} \) for wires after 2-qubit gates, setting \( \mathbb{P}((P, w)) = \frac{\varepsilon}{3} \). For other wires, \( \mathbb{P}((P, w)) = 0 \). We also model idling noise with a depolarizing channel following each 2-qubit gate, using an ALAP scheduling with entangling gate time of \( 60ns \) and single-qubit gate time of \( 50ns \) for non-diagonal gates and \( 0ns \) for diagonal gates. The idling error probability for duration \( \tau \) is \( \varepsilon_\tau = 1-e^{-\tau/T} \) with \( T = 100\mu s \). We merge this channel with the previous gate's channel, setting \( \mathbb{P}(P, w) = \frac{\varepsilon + \varepsilon_\tau - \frac{4\varepsilon \varepsilon_\tau}{3}}{3} \) for wires after 2-qubit gates followed by idling time \( \tau \).
Even though this error model does not capture all physical sources of error, it remains qualitatively fair for our purpose: wires followed by longer delays are subject to stronger noise and it is more important to cover them with checks.

\section{Finding low-weight checks}

Having discussed what constitutes valid spacetime checks and how effective a particular check is, we now address the issue of how to pick a good check. This section formalizes the challenge of identifying efficient, low-weight Pauli checks within a given subset of wires in a circuit. We discuss the use of back-propagators and decoding heuristics to find checks that minimize overhead while maximizing error detection, ensuring that the checks are both effective and practical to implement. 

\subsection{Reduction to syndrome decoding}

 We aim to select effective checks that are cost-efficient to implement. The complexity of implementing a check is influenced by its weight and the proximity of the ancilla qubit to the data qubits with which it needs to interact. To minimize overhead, we can limit the check to a specific set of target wires that are close to the ancilla qubit. This can be wires at a constant physical radius away from the ancilla, or in the extreme case, only its nearest neighbor.
In practice, this requires us to efficiently solve the following problem:
\medskip

\noindent\textbf{Problem 1: Finding a low weight valid check}
\smallskip

\noindent\textbf{Input:} A subset of (accessible) wires \( L \subseteq \mathcal{W} \)
\smallskip

\noindent\textbf{Output:} Pauli assignments  \( P_w \in \{I, X, Y, Z\} \) for each wire \( w \in L \), with a minimal number of non$-I$ Paulis and such that:

\begin{align}
    \prod_{w \in L} B(P_w, w) &= I  \label{eq:reduces_to_i_second}
\end{align}
\medskip

Using standard Pauli encoding and dropping phases, notice that Eq.~\ref{eq:reduces_to_i_second} can be expressed as a set of linear constraints on some solution vector. In particular, Problem 1
is equivalent to finding a codeword of minimal weight in a linear code:
\medskip

\noindent\textbf{Problem 2: Finding a minimal weight code word}
\smallskip

\noindent\textbf{Input:} A parity check matrix $B$
\smallskip

\noindent\textbf{Output:} 
\begin{equation}
\begin{aligned}
\min_{x} \quad & |x|\\
\textrm{s.t.} \quad & xB = 0
\end{aligned}
\end{equation}
\medskip

In our case, $B$ is built by stacking the Boolean Pauli encodings of the back-propagators $B(P, w)$ for all $P \in \{X, Y, Z\}$ and $w\in L$.
This problem is notoriously hard to solve optimally for a generic parity check matrix $B$ \cite{641542}. 

In our setting, it can also be useful to look for \textit{completions} of a partial collection of Paulis into a valid check, which will allow us to force checks to span certain windows of wires.
\medskip

\noindent\textbf{Problem 3: Completion into a valid check}
\smallskip

\noindent\textbf{Input:} A subset of wires \( L \subseteq \mathcal{W} \), and a partial assignment of Pauli operators $Q_w$ for $w \in J \subset L$
\smallskip

\noindent\textbf{Output:} Pauli assignments  \( P_w \in \{I, X, Y, Z\} \) for each wire \( w \in L \setminus J \), with a minimal number of non$-I$ Paulis and such that:

\begin{align*}
    \prod_{w \in L \setminus J} B(P, w) &= \prod_{w \in J} B(P, w)
\end{align*}
\medskip

By the same argument, using Boolean encoding of Pauli operators and dropping phase, this problem reduces to the following decoding problem:
\medskip

\noindent\textbf{Problem 4: Syndrome decoding}
\smallskip

\noindent\textbf{Input:} A parity check matrix $B$, a vector $q$
\smallskip

\noindent\textbf{Output:} 
\begin{equation}
\begin{aligned}
\min_{x} \quad & |x|\\
\textrm{s.t.} \quad & xB = q
\end{aligned}
\end{equation}
\medskip

Here $B$ is built by stacking the encoding of the back-propagators $B(P, w)$ for all $P \in \{X, Y, Z\}$ and $w\in L \setminus J$ and $q$ is the encoding of $\prod_{w \in J} B(P, w)$.
This problem is once again NP-Hard in the general case \cite{1055873}.

\subsection{Generalization to the stabilized case}
When the circuit input is stabilized by some input Abelian subgroup $S$, we are interested in the following problem:
\medskip

\noindent\textbf{Problem 5: Finding a low weight valid check up to stabilizer}
\smallskip

\noindent\textbf{Input:} A subset of wires \( L \subseteq \mathcal{W} \), an Abelian subgroup $S$
\smallskip

\noindent\textbf{Output:} Pauli assignments  \( P_w \in \{I, X, Y, Z\} \) for each wire \( w \in L \), with a minimal number of non$-I$ Paulis and such that:

\begin{align*}
    \prod_{w \in L} B(P_w) &\in S 
\end{align*}
\medskip

By the same means described above, this problem reduces to the following decoding problem:
\medskip

\noindent\textbf{Problem 6: Finding a minimal weight code word up to a subspace}
\smallskip

\noindent\textbf{Input:} A parity check matrix $B$, a matrix $T$
\smallskip

\noindent\textbf{Output:} 
\begin{equation}
\begin{aligned}
\min_{x, y} \quad & |x|\\
\textrm{s.t.} \quad & xB = yT
\end{aligned}
\end{equation}
\medskip

where $T$ is obtained by stacking the generators of $S$. We can eliminate the dependency on an extra variable $y$ by computing the nullspace $N$ of $T$ and solving the problem:
\medskip

\noindent\textbf{Problem 7: Finding a minimal weight code word up to a subspace}
\smallskip

\noindent\textbf{Input:} A parity check matrix $B$, a matrix $T$ of nullspace $N$
\smallskip

\noindent\textbf{Output:} 
\begin{equation}
\begin{aligned}
\min_{x, y} \quad & |x|\\
\textrm{s.t.} \quad & xBN = 0
\end{aligned}
\end{equation}
\medskip

which is as hard as Problem 2.

Similarly, one can consider the problem of finding a low weight completion of a check up to some stabilizer group. This problem reduces to the following problem:
\medskip

\noindent\textbf{Problem 8: Syndrome decoding up to a subspace}
\smallskip

\noindent\textbf{Input:} A parity check matrix $B$, a matrix $T$ of nullspace $N$, a vector $q$
\smallskip

\noindent\textbf{Output:} 
\begin{equation}\label{eq:to_identity}
\begin{aligned}
\min_{x, y} \quad & |x|\\
\textrm{s.t.} \quad & xBN = qN
\end{aligned}
\end{equation}
\medskip

This last problem is the one we will actually attempt to solve, that is to find a low-weight completion of checks in the presence of an input stabilizer group. It captures all the other problems as special cases and can be adapted to fit most applications.

\subsection{A heuristic for the syndrome decoding problem}

We established that, in the general case, finding a low weight check on a target support is hard. In practice, we are not necessarily interested in the optimal solution to that problem.
Indeed, it often happens that the optimal solution to this problem actually yields a poor quality check. Consider, for instance, the case of a brickwork circuit where $L$ is the set of wires of some target qubit.
In that setting, it can happen that the minimal check is simply a weight $2$ check ``sandwiching'' a single entangling gate. Such a check would only catch errors concentrated on that gate, and it would be very likely to introduce more errors than it catches. The hardness of the problem, coupled with this remark, leads us to consider a suboptimal but fast approach that works well in practice.
\smallskip

\noindent In order to solve Problem 8, we take the following greedy approach:
\begin{enumerate}
    \item find $ \operatorname{min}_{i} |B[i] \oplus q| $
    \item update $q \gets B[i] \oplus q$
    \item if $|q| > 0$ go to $(1)$
\end{enumerate}

In other words, we greedily add to the solution the row of $B$ that lowers its Hamming weight the most. Notice that this method might not converge if no solution exists. We solve this by raising an error if no row can be found that reduces the weight of $q$.
We couple this greedy approach with a randomization of $B$. The rows of $B$ are randomly permuted, and $B$ is brought into column echelon form, replicating column operations onto $q$. This ensures that if the canonical vectors were present in $B$, they will still be present in the newly produced matrix. It is easy to notice that a solution to this new problem can be permuted back into a solution of the original problem.
This technique was previously utilized successfully to synthesize efficient Clifford circuits (see \cite{10.1007/978-3-030-52482-1_11, Goubault_de_Brugi_re_2025}).
Notice that this algorithm runs in polynomial time. $B$ can be generated in time $O(|L|mn)$ where $n$ (resp. $m$) is the number of qubits (resp. gates) in the payload. The greedy decoding algorithm has runtime $O(|L|n)$.
 









\subsection{Characterizing the group of valid checks}\label{subsec:group_size}

It is straightforward to notice that valid spacetime checks, as defined by checks satisfying Eq.~\ref{eq:reduces_to_i} or Eq.~\ref{eq:reduces_to_s} form a group. Moreover, the decoding framework introduced above gives us tools to directly access elements in that group.

Consider Problem 2 for instance. The valid checks are exactly specified by elements $x \in \mathbb{F}_2^{3|L|}$ such that $xB=0$. This is equivalent to saying that valid checks are characterized by the elements of the (left) nullspace of $B$.
This remark leads to the following Lemma.

\begin{lemma}
    Let $L$ be a subset of wires of an $n$ qubit Clifford payload. We have that the number of distinct valid checks supported on $L$ is $2^{2|L| - \operatorname{rk}(B)}$.
\end{lemma}
\begin{proof}
    It is obvious that the left nullspace of $B$ has dimension $2^{3|L| - \operatorname{rk}(B)}$. However, in our formulation, every $3^{rd}$ row in $B$ corresponds to a combination of the previous two rows. This is equivalent to saying that in a check, having Paulis $(X, w), (Z, w)$ is equivalent to having only Pauli $(Y, w)$. Removing all the $Y$ rows in $B$ still yields a valid matrix $B'$. This time, any vector $x'$ such that $x' B' = 0$ uniquely defines a valid check. The size of its left nullspace is exactly $2^{2|L| - \operatorname{rk}(B)}$.
\end{proof}

This Lemma trivially generalizes in the presence of a non-trivial stabilizer group $S$, in which case the number of checks becomes $2^{2|L| - \operatorname{rk}(BN)}$.
In a random payload, when considering the full set of valid checks, we expect the number of valid checks to grow exponentially with the size of the support $L$. Figure~\ref{fig:scaling_check_dim} shows numerical evidences of this scaling.

\begin{figure}[h]
    \centering
    \includegraphics[width=0.45\linewidth]{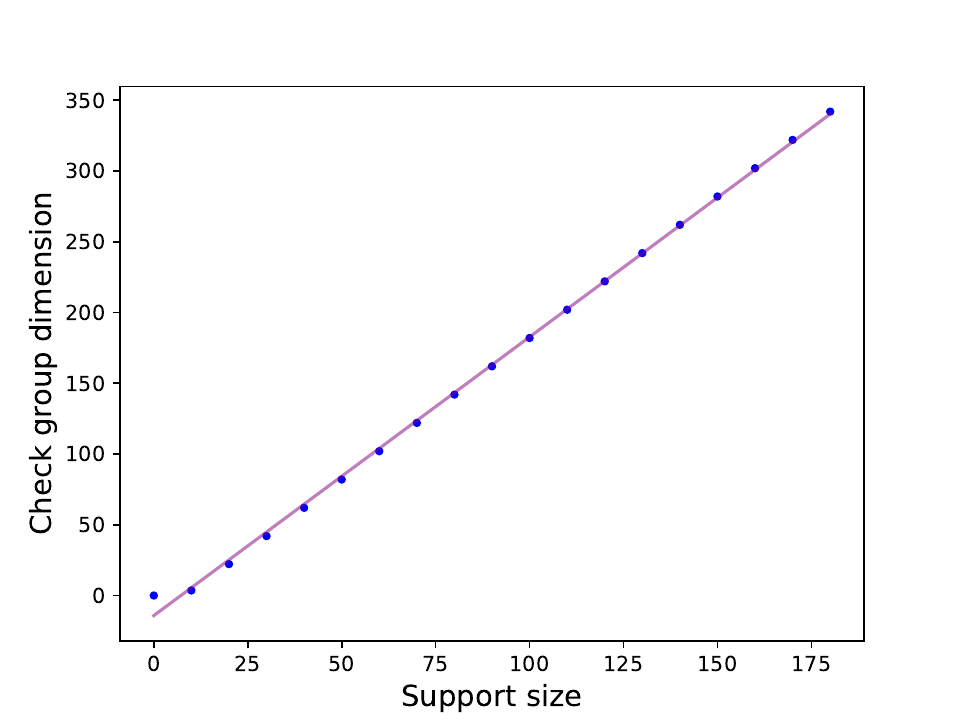}
    \caption{\textbf{The size of the group of valid checks.} Scaling of the check group's dimension as a function of the check's support size. In this example, the payload is a random brickwork Clifford circuit over a path of 20 qubits and of depth $40$. The support is grown by iteratively adding random wires to it. Each point is averaged over 10 runs. The fit is near perfect, supporting the assertion that the check group size grows exponentially with increased support size.}
    \label{fig:scaling_check_dim}
\end{figure}

 The previous remark also gives a means to uniformly sample valid checks. One can simply compute a set of generators of the left nullspace of the appropriate matrix (removing the $Y$ rows of $B$), and pick random combinations of those generators.

\section{Extending spacetime checks to Clifford + Pauli rotation circuits}
While coherent Pauli checks were first designed to detect errors within a pure Clifford payload, it is possible to extend them to a universal class of circuits, namely, circuits composed of Clifford gates and non-Clifford single-qubit Pauli rotations. In \cite{gonzales_quantum_2023}, the authors propose that a left/right check can be used in the presence of some non-Clifford rotations provided that the check can be commuted through. However, their trial-and-error solution to finding checks suffers from an exponential slowdown with increasing number of rotations in the circuit.

In this section we show that the framework of spacetime Pauli checks provides us with efficient tools to fully characterize the set of valid checks in the presence of non-Clifford rotations, by using the fact that additional validity constraints are linear and can be enforced in our decoding formalism. This keeps check picking efficient for universal circuits as well, despite the shrinking space of valid checks. Furthermore, we show that the space of checks can be fully restored at the expense of some overhead during the implementation of those checks that do not satisfy the new constraints.

\subsection{Validity of a check in the presence of rotations}

Before we present this generalization, consider a Clifford payload and a single valid check $\mathcal{C} = \{(P_1, w_1), ..., (P_k, w_k)\}$ (in the sense of Eq.~\ref{eq:reduces_to_i}). It is straightforward to notice that for any bi-partition $\mathcal{L}\subseteq\mathcal{C}, \mathcal{C}\setminus \mathcal{L}$ of $\mathcal{C}$, we have:
$$\prod_{(P, w)\in \mathcal{L}} B(P_i, w_i) = \prod_{(P, w)\in \mathcal{C}\setminus \mathcal{L}} B(P_i, w_i)$$
Notice also that this equation holds at every space-like in the circuit: one can pick any set of independent wires in the circuit, propagate all Paulis onto those wires, and express the same equality.

Consider such a circuit with a unique non-Clifford single-qubit rotation.
This circuit can be seen as a Clifford circuit with a single-qubit rotation $R_P(\theta)$ located on some internal wire $w_r$, and consider the circuits $A_{w_r}$ and $B_{w_r}$ as defined in Figure~\ref{fig:backprop}.
Our payload can be written as $B_{w_r}\cdot  R_{P_{w_r}}(\theta)\cdot A_{w_r}$. This bi-partition of the circuit induces a natural bi-partition of any check $\mathcal{C}$ into two half-checks $\mathcal{L}= \{(P, w) \in \mathcal{C}, w \leq w_r \}$ and $\mathcal{R}= \{(P, w) \in \mathcal{C}, w > w_r\}$. Those two half-checks, when injected in the circuit, will effectively change $A_{w_r}$ and $B_{w_r}$ into $QA_{w_r}$ and $B_{w_r} Q$ for some operator $Q$ obtained by pulling either half of the check to any spacelike containing $w_r$. Figure~\ref{fig:clifford_t_example_single} depicts a circuit equivalent to the one we would get when actually implementing the check.

\begin{figure}[h]
    \centering
    \begin{tikzpicture}
        \draw[ibmcyan, thick] (-0.2, -0.3) -- (6.4, -0.3);
        \draw[ibmcyan, thick] (-0.2, -0.8) -- (6.4, -0.8);
        \draw[ibmcyan, thick] (-0.2, -1.5) -- (6.4, -1.5);
        \draw[ibmcyan, thick] (-0.2, -2.3) -- (6.4, -2.3);
        \draw[ibmcyan, thick] (-0.2, -2.8) -- (6.4, -2.8);
        \draw[ibmpurple, thick] (-0.2, -3.4) -- (6.4, -3.4);
        \draw (-0.7, -3.4) node[ibmpurple] {$\ket{+}$};
        \draw (1.8, -3.4) node[ibmpurple, fill=ibmpurple, circle, inner sep=2pt] {};
        \draw (4.4, -3.4) node[ibmpurple, fill=ibmpurple, circle, inner sep=2pt] {};
        \draw[thick, ibmpurple] (4.4, -3.4) -- (4.4, - 3);
        \draw[thick, ibmpurple] (1.8, -3.4) -- (1.8, - 3);
        \draw[fill=ibmcyan, rounded corners, ibmcyan] (0, 0) rectangle node[white]{$A_{w_r}$} (1, -3);
        \draw[fill=ibmpurple, rounded corners, ibmpurple] (1.3, 0) rectangle node[white]{$Q$} (2.3, -3);
        \draw[fill=ibmteal, rounded corners, ibmteal] (2.6, -1 - 0.15) rectangle node[white]{$R_{P}(\theta)$} (3.6, -2 + 0.15);
        \draw[fill=ibmpurple, rounded corners, ibmpurple] (3.9, 0) rectangle node[white]{$Q$} (4.9, -3);
        \draw[fill=ibmcyan, rounded corners, ibmcyan] (5.2, 0) rectangle node[white]{$B_{w_r}$} (6.2, -3);
    \end{tikzpicture}
    \caption{\textbf{Picking checks in the presence of non-Clifford rotations.} Any check in a payload composed of Clifford gates and a single non-Clifford rotation can be put into this form (in the case of a trivial stabilizer group).}
    \label{fig:clifford_t_example_single}
\end{figure}
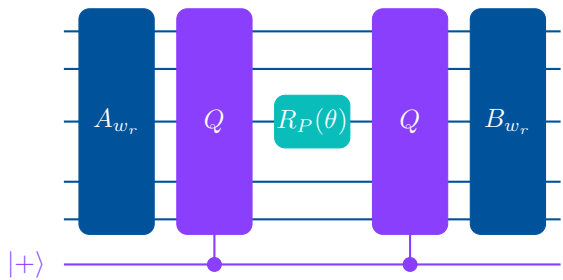

Now, it is easy to notice that if $[P, Q] \neq 0$, then the resulting circuit is not equivalent to the original one, and the ancilla qubit becomes entangled with the payload qubit, failing to create an invariant that can be efficiently checked. In other words, the notion of a valid check is now more constrained: each rotation present in the circuit will induce new constraints on the checks.

Luckily, it is possible to modify the formulation of Problem 7 to take these additional constraints into account. Indeed, consider $R$, the $3|L| \times m$ table where $m$ is the number of rotations
such that the row indexed by $(P, w)$ has a $1$ in column $(Q, x)$ if and only if :
    \begin{itemize}
        \item $[B(P, w), B(Q, x)] \neq 0$
        \item $x < w$
    \end{itemize}
In other words, $R$ summarizes all the commutation relations between all rotations and any Pauli that would potentially belong to our check, coming \textit{after} the rotation. 
Checking that a check, corresponding to a set of rows of $R$, commutes with a given rotation, corresponding to a column of $R$, now boils down to considering the corresponding entries in $R$ and checking that they sum to $0$ mod 2 (see Figure~\ref{fig:build_matrix_c} for a detailed example).

\begin{figure}[h]
    \centering
    \begin{tikzpicture}[baseline=(anchor)]
    \coordinate (anchor) at (0, -1);
        \draw[ibmcyan, thick] (-0.2, 0) -- (3.6 + 1.5, 0);
        \draw[ibmcyan, thick](-0.2, -1) -- (3.6+ 1.5, -1);
        \draw (0.5, 0) node[ibmcyan, fill=ibmcyan, circle, inner sep=2pt] {};
        \draw[ibmcyan, thick] (0.5, -1) node[circle, inner sep=0.15cm, draw]{};
        \draw[ibmcyan, thick] (0.5, 0) -- (0.5, -1 - 0.23);

        \draw[fill=ibmteal, rounded corners, ibmteal] (1, 0.35) rectangle node[white]{$R_{X}(\theta_1)$} (2.2, -0.35);
        \draw[fill=ibmteal, rounded corners, ibmteal] (1, 0.35 - 1) rectangle node[white]{$R_{Z}(\theta_2)$} (2.2, -1.35);

        \draw[fill=ibmcyan, rounded corners, ibmcyan] (2.5, 0.35 - 1) rectangle node[white]{$\sqrt{X}$} (3.7, -1.35);
        \draw (2.7+ 1.5, 0) node[ibmcyan, fill=ibmcyan, circle, inner sep=2pt] {};
        \draw[ibmcyan, thick] (2.7+ 1.5, -1) node[circle, inner sep=0.15cm, draw]{};
        \draw[ibmcyan, thick] (2.7+ 1.5, 0) -- (2.7+ 1.5, -1 - 0.23);

        \draw (-0.1, 0.3) node {$w_1$};

        \draw (3.1+ 1.5, 0.3) node {$w_2$};
        \draw (3.1+ 1.5, 0.3 - 1) node {$w_3$};

        \draw (1.6, 0.5) node {$w_a$};
        \draw (1.6, -1.5) node {$w_b$};
    \end{tikzpicture}~~~~~
        \begin{tabular}{cc|c|c}
             & & $(Z, w_a)$ & $(X, w_b)$ \\
            \hline \multirow{3}{0.7em}{$w_1$} & $X$ & - & -\\
                                            & $Y$ & - & -\\
                                            & $Z$ & - & -\\
            \hline \multirow{3}{0.7em}{$w_2$} & $X$ & 0 & 1 \\
                                            & $Y$ & 1 & 1 \\
                                            & $Z$ & 1 & 0 \\
            \hline \multirow{3}{0.7em}{$w_3$} & $X$ & 0 & 1\\
                                            & $Y$ & 1 & 0\\
                                            & $Z$ & 1 & 1\\
        \end{tabular}
    \caption{\textbf{Example of extra constraints on checks due to rotations.} In this example, $L = \{w_1, w_2, w_3\}$. Since $w_1$ comes before $w_a$ and $w_b$, the corresponding table entries are set to $0$ (here $-$ for clarity). From the constructed table, we can observe that picking the half-check $(Y, w_2), (Z, w_3)$ would be a valid choice. Indeed, an operator $YZ$ at the end of the payload will propagate as a $XZ$ operator on wires $(w_a, w_b)$, which commutes with both rotations. On the other side, picking check $(Y, w_2), (X, w_3)$ is an invalid choice. The operator would propagate as $Y$ on $w_a$ which anti-commutes with the first rotation. }
    \label{fig:build_matrix_c}
\end{figure}
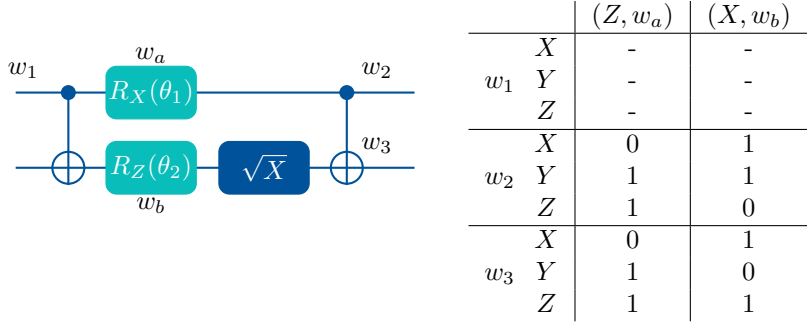

Once $R$ is constructed, one can modify Problem 4 as follows:
\medskip

\noindent\textbf{Problem 9: Syndrome decoding skipping rotations}
\smallskip

\noindent\textbf{Input:} A parity check matrix $B$, a matrix $R$
\smallskip

\noindent\textbf{Output:} 
\begin{equation}
\begin{aligned}
\min_{x, y} \quad & |x|\\
\textrm{s.t.} \quad & x \left(B\ R\right) = \left(q\ r\right)
\end{aligned}
\end{equation}
\medskip

where $\left(B\ R\right)$ is the horizontal stacking of $B$ and $R$, and $r$ is the anti-commutation vector obtained for the fixed portion of the check. Of course, all this construction can be generalized in the case where $S$ is non-trivial, leading to the more general formulation:
\medskip

\noindent\textbf{Problem 10: Syndrome decoding up to subspace, skipping rotations}
\smallskip

\noindent\textbf{Input:} A parity check matrix $B$, a matrix $R$, and a matrix $T$ of nullspace $N$
\smallskip

\noindent\textbf{Output:} 
\begin{equation}
\begin{aligned}
\min_{x, y} \quad & |x|\\
\textrm{s.t.} \quad & x \left(BN\ R\right) = \left(qN\ r\right)
\end{aligned}
\end{equation}
\medskip

Any solution to this problem would be a valid check that commutes with all rotations present in the circuit. Even though this approach is mathematically sound, it is unlikely to scale well in practice since the set of valid checks would shrink dramatically when adding new rotations.
To illustrate this behavior, we ran simulations of near-Clifford payloads with an increasing number of non-Clifford rotations.
Simulation is performed by propagating Pauli errors through the payload and checking if the error commutes with the injected rotations. Runs are considered faulty if the accumulated error anti-commutes with any rotation in the circuit or if the final error acts non-trivially on the payload qubits. Results are depicted in Figure~\ref{fig:non_clifford_simulation}.

\begin{figure}
    \centering
    \includegraphics[width=0.85\linewidth]{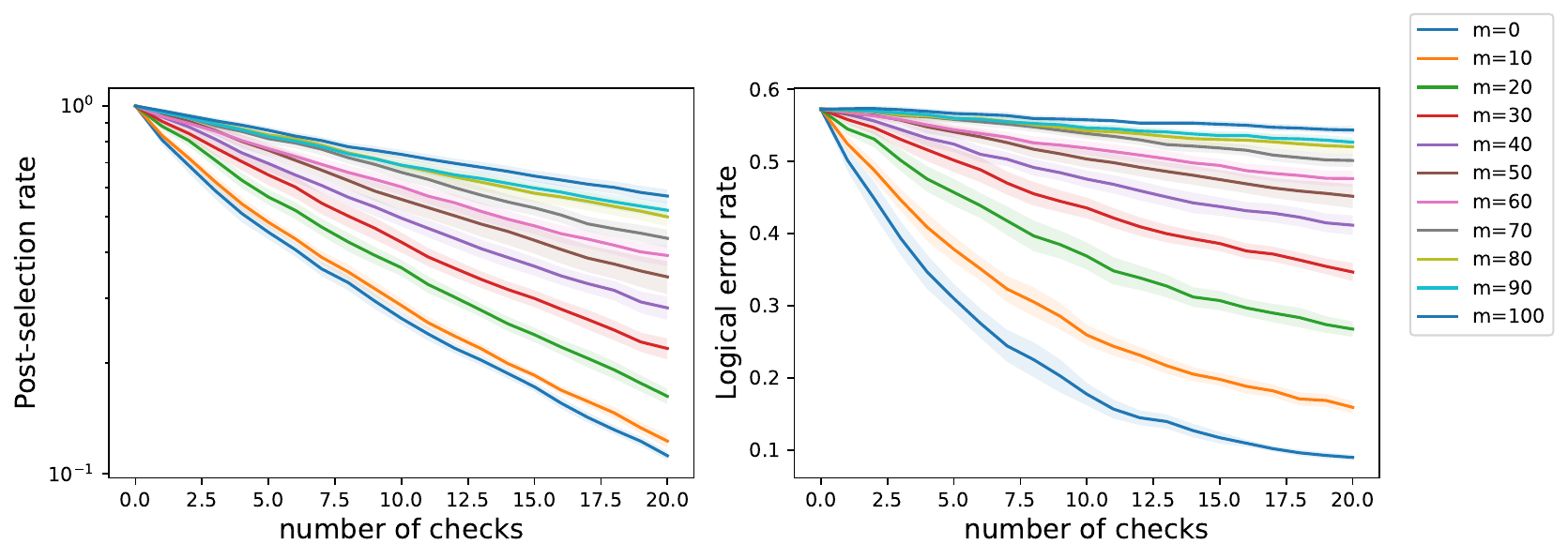} \includegraphics[width=0.45\linewidth]{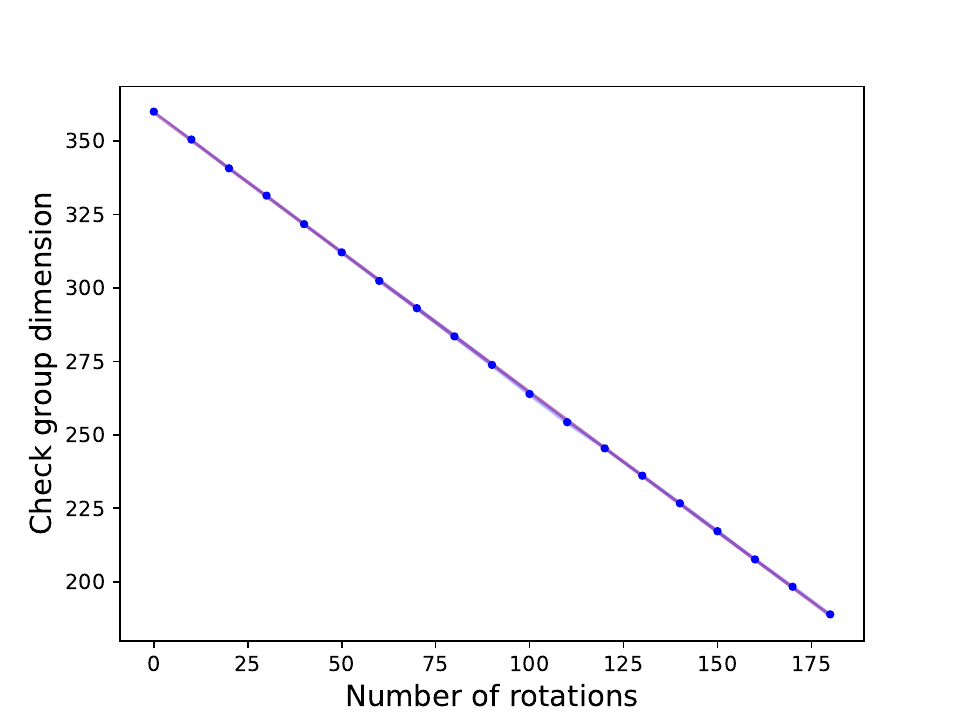}
    \caption{\textbf{Effect of increasing non-Clifford rotations on the performance of checks.} Payloads are generated by constructing a random brickwork Clifford circuits of depth 40 over 20 qubits and injecting $T$ gates on random wires.  (Top) We report post-selection rate and logical error rate as a function of the number of checks in the presence of a variable number of rotations ($m$). Each point is averaged over 20 runs with different rotation injected in the same initial circuit. Shaded regions correspond to $95\%$ confidence interval. Each check probes a single data qubit. The noise model consists of single-qubit uniform Pauli channels with rate $0.1\%$ attached to all wires directly following $2$-qubits gates. Data is presented for various rotation counts ($m$) ranging from $0$ to $100$. (Bottom) Scaling of the check group dimension as a function of the number of rotations. In this setting, rotations are injected at random positions. Each point is averaged over $10$ different rotation emplacements and axes. As expected, the number of valid checks decreases exponentially with the number of rotations.}
    \label{fig:non_clifford_simulation}
\end{figure}

Our decoding method remains applicable to Problem 10, and all the matrices used can be constructed in polynomial time from the input circuit's description. This approach offers an efficient heuristic for identifying checks in the presence of rotations, unlike the trial-and-error method described in \cite{gonzales_quantum_2023}, which has a runtime that scales exponentially with the number of rotations.

It is straightforward to notice that any check that is compatible with some rotation of axis $P$ on wire $w$ cannot detect an error $P$ on that wire. This is due to the fact that the back-cumulant of the check will commute with $P_w$. This can seem like a reasonable drawback until considering practical use cases. When considering Clifford + T circuits for instance, noisy $T$ gates can be twirled, yielding a diagonal Pauli channel \cite{PhysRevLett.127.200505}. This entails that any error due to the $T$ gates remains undetected by valid checks, which would be detrimental in regimes where Clifford gates have a lower error rate than $T$ gates. 

 We can extend the results of Section \ref{subsec:group_size} in the presence of non-Clifford rotations. In this setting, the number of checks becomes $2^{2|L| - \operatorname{rk}(B N\ R)}$.
In a random payload, when considering the full set of valid checks,  we expect the number of valid checks to decrease exponentially with the number of rotations.
Figure~\ref{fig:non_clifford_simulation} (Bottom) shows numerical evidence of these scalings.

The fact that the check group size decreases exponentially suggests that we need to remove such constraints on valid checks. In the next section we show how this can be done at the expense of some extra gate overhead. Even though, this approach will let us recover the full set of checks, those checks won't be able to catch errors due to the rotations.

\subsection{Fixing invalid checks in the presence of rotations}
In our framework, we have the freedom of picking any Pauli located in the target support $L$ to implement our checks.
We can exploit this freedom to relax the constraints described in the previous section, and consider a check that might not commute with some rotation of axis $P$ on wire $w$ when $w\in L$.

Indeed, consider some check $\mathcal{C}$ that anti-commutes with some rotation $R_{P_w}(\theta)$. Now consider any Pauli $V \in \{X, Y, Z\}$ such that $[V, P] \neq 0$. We claim that the check $\mathcal{C} \cup \{(V, w-), (V, w+)\}$, where $w\pm$ denotes the wires right before and after the rotation in the full circuit, is a compatible check.
Indeed, consider the half-checks $\mathcal{R}$ and $\mathcal{C}\setminus \mathcal{R}$ as described in the previous section. When propagating $\mathcal{R}$ (resp. $\mathcal{C}\setminus \mathcal{R}$) right before (resp. after) the rotation, they account for the same operator $Q$ such that $[Q, P_w]\neq 0$, since $\mathcal{C}$ is not compatible with the rotation. When adding $(V, w-)$ and $(V, w+)$ to the check, $Q$ becomes $Q\cdot V_w$ and we have $[Q\cdot  V_w, P_w] = 0$, hence the compatibility of the new check. Moreover, this new check is obviously valid in the sense of Eq.~\ref{eq:reduces_to_i} (or Eq.~\ref{eq:reduces_to_s} in the presence of a non-trivial $S$).

In practice, this means that one can remove some columns of $R$ when solving the decoding problem. Namely, any column corresponding to some rotation $(Q, x)$ where $x\in L$ can be safely removed. The resulting check might be incompatible with the corresponding rotations. We will then add in the correct $(V, w)$ Paulis to the check to make it compatible with those rotations (see Figure~\ref{fig:clifford_t_fixing}). Note that the newly added controlled-V gates are Cliffords. Therefore, we do not add new rotations to the circuit, which can be useful in settings where noise is predominantly due to rotations (e.g. $T$ gates) and Cliffords are less noisy.

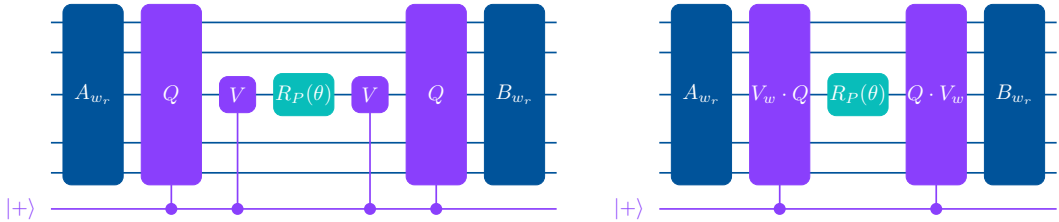
\begin{figure}[h]
    \centering
    \scalebox{0.8}{\begin{tikzpicture}
        \draw[ibmcyan, thick] (-0.2, -0.3) -- (8.2, -0.3);
        \draw[ibmcyan, thick] (-0.2, -0.8) -- (8.2, -0.8);
        \draw[ibmcyan, thick] (-0.2, -1.5) -- (8.2, -1.5);
        \draw[ibmcyan, thick] (-0.2, -2.3) -- (8.2, -2.3);
        \draw[ibmcyan, thick] (-0.2, -2.8) -- (8.2, -2.8);
        \draw[ibmpurple, thick] (-0.2, -3.4) -- (8.2, -3.4);
        \draw (-0.7, -3.4) node[ibmpurple] {$\ket{+}$};
        \draw (1.8, -3.4) node[ibmpurple, fill=ibmpurple, circle, inner sep=2pt] {};
        \draw[thick, ibmpurple] (1.8, -3.4) -- (1.8, - 3);
        \draw (6.2, -3.4) node[ibmpurple, fill=ibmpurple, circle, inner sep=2pt] {};
        \draw[thick, ibmpurple] (6.2, -3.4) -- (6.2, - 3);
        \draw (2.9, -3.4) node[ibmpurple, fill=ibmpurple, circle, inner sep=2pt] {};
        \draw[thick, ibmpurple] (2.9, -3.4) -- (2.9, -2+0.2);
        \draw (5.1, -3.4) node[ibmpurple, fill=ibmpurple, circle, inner sep=2pt] {};
        \draw[thick, ibmpurple] (5.1, -3.4) -- (5.1, - 2 +0.2);
        \draw[fill=ibmcyan, rounded corners, ibmcyan] (0, 0) rectangle node[white]{$A_{w_r}$} (1, -3);
        \draw[fill=ibmpurple, rounded corners, ibmpurple] (1.3, 0) rectangle node[white]{$Q$} (2.3, -3);
        \draw[fill=ibmpurple, rounded corners, ibmpurple] (2.6, -1 - 0.2) rectangle node[white]{$V$} (3.2, -2+0.2);
        \draw[fill=ibmteal, rounded corners, ibmteal] (3.5, -1-0.15) rectangle node[white]{$R_{P}(\theta)$} (4.5, -2+0.15);
        \draw[fill=ibmpurple, rounded corners, ibmpurple] (4.8, -1-0.2) rectangle node[white]{$V$} (5.4, -2+0.2);
        \draw[fill=ibmpurple, rounded corners, ibmpurple] (5.7, 0) rectangle node[white]{$Q$} (6.7, -3);
        \draw[fill=ibmcyan, rounded corners, ibmcyan] (7, 0) rectangle node[white]{$B_{w_r}$} (8, -3);
    \end{tikzpicture}~~~~~~~\begin{tikzpicture}
        \draw[ibmcyan, thick] (-0.2, -0.3) -- (6.4, -0.3);
        \draw[ibmcyan, thick] (-0.2, -0.8) -- (6.4, -0.8);
        \draw[ibmcyan, thick] (-0.2, -1.5) -- (6.4, -1.5);
        \draw[ibmcyan, thick] (-0.2, -2.3) -- (6.4, -2.3);
        \draw[ibmcyan, thick] (-0.2, -2.8) -- (6.4, -2.8);
        \draw[ibmpurple, thick] (-0.2, -3.4) -- (6.4, -3.4);
        \draw (-0.7, -3.4) node[ibmpurple] {$\ket{+}$};
        \draw (1.8, -3.4) node[ibmpurple, fill=ibmpurple, circle, inner sep=2pt] {};
        \draw (4.4, -3.4) node[ibmpurple, fill=ibmpurple, circle, inner sep=2pt] {};
        \draw[thick, ibmpurple] (4.4, -3.4) -- (4.4, - 3);
        \draw[thick, ibmpurple] (1.8, -3.4) -- (1.8, - 3);
        \draw[fill=ibmcyan, rounded corners, ibmcyan] (0, 0) rectangle node[white]{$A_{w_r}$} (1, -3);
        \draw[fill=ibmpurple, rounded corners, ibmpurple] (1.3, 0) rectangle node[white]{$V_w\cdot Q$} (2.3, -3);
        \draw[fill=ibmteal, rounded corners, ibmteal] (2.6, -1 - 0.15) rectangle node[white]{$R_{P}(\theta)$} (3.6, -2 + 0.15);
        \draw[fill=ibmpurple, rounded corners, ibmpurple] (3.9, 0) rectangle node[white]{$Q\cdot V_w$} (4.9, -3);
        \draw[fill=ibmcyan, rounded corners, ibmcyan] (5.2, 0) rectangle node[white]{$B_{w_r}$} (6.2, -3);
    \end{tikzpicture}}
    \caption{Any incompatible check can be fixed into a valid compatible check at the cost of two additional CV gates where $[V_w, P_w] \neq 0$. The circuit can be brought into the form depicted here (left). This circuit is equivalent to the circuit on the right. Since $[Q, P_w] \neq 0$ and $[V_w, P_w]\neq 0$, we have that $[V_w \cdot Q, P_w] = 0$}
    \label{fig:clifford_t_fixing}
\end{figure}

\section{Picking multiple checks around a payload} \label{sec:multicheck}

In this section we outline our strategy for selecting good checks among many valid checks, and how we efficiently implement them in hardware-constrained environments. We describe how to associate ancilla qubits with data qubits to perform checks and discuss the process of iteratively picking and scoring checks to optimize error detection while managing noise and routing overhead.

Figure~\ref{fig:layout_example} shows an example layout of data and check qubits. The payload occupies some physical region of the hardware, while the checks probe it from the periphery to protect it from errors.
Given a Clifford circuit transpiled onto an arbitrary subgraph of the connectivity graph, we find data qubits (green) that are adjacent to a path of ancilla qubits (purple). Each ancilla in the path implements a check on the corresponding data qubit's wires. Since only one ancilla in the path has direct access to the data qubit, further ancilla are swapped closer when needed.

\begin{figure}[h]
    \centering
    \includegraphics[width=0.5\linewidth]{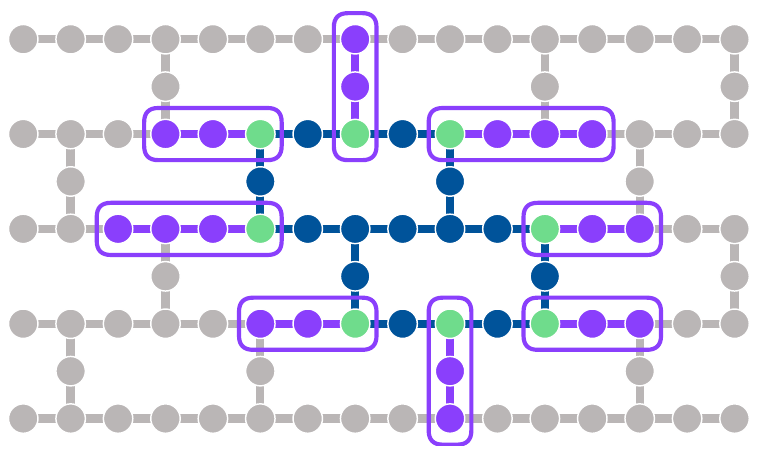}
    \caption{\textbf{Example physical layout of qubits}. The payload is supported on some contiguous set of physical qubits (green and blue), and free qubits nearby can be used as check ancilla. In our experiments we use a path of ancilla (purple), and restrict them to measure Paulis with support only on their nearest neighbor (green) to minimize their overhead, but this can also be relaxed to a short-range neighborhood.
    Half-SWAPS are inserted so that each qubit in the purple paths can be moved next to their corresponding green data qubit in order to implement the checks.
    }
    \label{fig:layout_example}
\end{figure}

\paragraph{Picking good check support.}
Within a path of ancilla qubits, each check qubit implements a Pauli check on the data qubit's wires. Since the group of valid checks is exponentially large, we need a way to restrict the search to non-trivial checks whose back-cumulants have a good chance of propagating far inside the payload. To force this behavior, we adopt the strategy of forcing some random Paulis on two (time) distant locations and use our decoding algorithm to complete this partial check into a valid check. Moreover, since we will pick several checks on (roughly) the same support, we want the check to be grouped on a restricted window of the support. Having a check supported on the full spread of the support might lead to a large half-SWAP overhead (see below).
We consider wires following 2-qubit gates, temporally sorting them to form the check's support. We pick a window size within \( 10\% \) to \( 30\% \) of the total support size and a random window position, giving us a window of wires \( [w_1, \ldots, w_k] \). We enumerate all pairs of Pauli operators \( (P_1, P_2) \in \{X, Y, Z\}^{\otimes 2} \), forcing values \( P_1 \) for \( w_1 \) and \( P_2 \) for \( w_k \), and use the decoding algorithm to find valid checks under these constraints. This process is repeated 15 times, leading to up to \( 135 \) valid checks. Each check is scored using the technique in Section~\ref{sec:coverage}, and the best check is selected. 

\paragraph{Picking multiple checks along a path.}

For a given path of check qubits, we iteratively pick checks. After finding the first best check, we commit it to the payload and proceed with the next checks. For the second check qubit, we first find a collection of checks and insert half-SWAP gates (see Figure~\ref{fig:half_swaps}) before scoring to account for additional noise. If a check's score decreases compared to the prior check, we abort check picking for that path.

\paragraph{Routing checks.}

\begin{figure}[h]
    \centering
\[ \begin{array}{ccc} \Qcircuit @C=1.0em @R=1.2em @!R { 
      & \qswap & \ctrl{1} & \qw\\
      & \qswap \qwx & \control\qw & \qw
 }  &\Qcircuit @C=1.0em @R=0.8em @!R { 
      & \\
      &=
 }&~~~\Qcircuit @C=1.0em @R=.2em @!R { 
      & \gate{H} & \ctrl{1} & \targ & \qw& \qw\\
      & \qw & \targ & \ctrl{-1} &\gate{H}&\qw
 } \end{array}\]

    \caption{\textbf{The half-SWAP identity used in check routing}. This simultaneously reduces the gate count and improves check coverage (see Figure~\ref{fig:half_swaps_vs_swaps}).}
    \label{fig:half_swaps}
\end{figure}
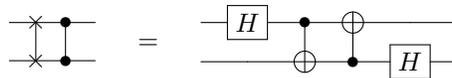

\begin{figure}[h]
    \centering
    \scalebox{0.7}{\input{figures/half_swaps}}
    \caption{\textbf{Two checks probing the same data qubit}. (Top) The two check are implemented, but the circuit is still incompatible with the hardware's connectivity (depicted on the left). The check back-cumulant propagates as a Pauli $X$ on the check qubit's wires, detecting $Z$ and $Y$ errors at those positions (here we left out the propagation of the cumulant inside the payload). (Middle) After routing with SWAP gates, the circuit is now compatible with the qubit's connectivity. The back-cumulant propagates similarly. The routing overhead is of $9$ entangling gates. (Bottom) When using half-SWAPs instead of SWAPs, we need to insert an additional CZ gate to have a total of $4$ equivalent CZ gates acting between the two ancilla. The gate overhead drops to $7$ entangling gates. Moreover, the back-cumulant now partially spreads to the other check qubit's wires. In this example, we can observe that $4$ additional wires are covered by Pauli $Z$s. This check will also detect $X$ and $Y$ errors on those $4$ additional wires. Notice that on those wires, the other check's cumulant will propagate an $X$ operator, thus catching all possible errors.}
    \label{fig:half_swaps_vs_swaps}
\end{figure}
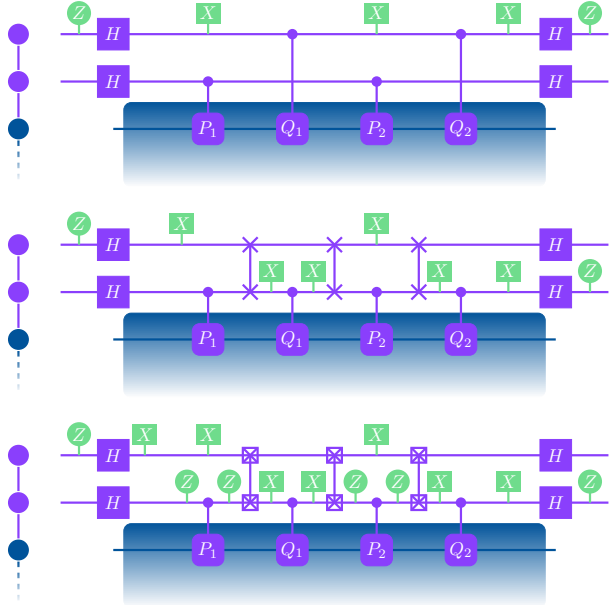

In our protocol, we are picking a check per available ancilla qubit along a path. Among all those ancilla qubits, only the first one is actually able to interact with the data qubits. To actually implement the full set of checks, we need to move the other ancillae into the first position using SWAP gates whenever they need to interact with the payload, incurring an additional gate overhead. 
This is done greedily by swapping the next ancilla that needs to interact with the data qubits to the first position in the path.

One can notice that given two check qubits implementing some separate checks on the payload, we can insert CZ gates between those ancilla qubits at any point during the check implementations (between the $H$ gates). As long as we insert an even number of them, they will commute with the check's controlled-Pauli gates and cancel out. We can use this trick to implement cheaper SWAPs consisting of a SWAP followed by a CZ gate. This circuit can be simplified into a cost 2 circuit (see Figure~\ref{fig:half_swaps}).
Overall, we make sure that each pair of ancilla along the path was subject to an even number of SWAPs, effectively ensuring that the additional CZ gates canceled out. In practice, we never use paths longer than $3$ in our experiments and are able to generate the optimal circuits implementing the $8$ possible different corrections.

Importantly, in addition to reducing the overall entangling gate count, this half-SWAP approach also has the advantage of propagating $X$ errors from one check to the others, leading to improved coverage (see Figure~\ref{fig:half_swaps_vs_swaps}).

\section{Stabilizer states, entanglement width, and fidelity estimation} \label{sec:stabilizer-state}

The main experiments we report in this work are the preparation of highly-entangled stabilizer states and measurement of their fidelity under error detection. This section formalizes the notion of entanglement width, a measure of the entanglement present in a quantum state. We discuss how stabilizer states can be characterized by their entanglement width and discuss how their experimental fidelity can be estimated using Monte Carlo state certification.

\subsection{Entanglement of a quantum state}

Entanglement is typically defined for bipartite systems. For a quantum state \( \ket{\psi} \) over two disjoint sets \( A \) and \( B \) of qubits, it can be expressed as \( \ket{\psi} = \sum_{i=1}^m \alpha_i \ket{\psi_i}_A \otimes \ket{\phi_i}_B \), where \( \ket{\psi_i} \) and \( \ket{\phi_i} \) are orthonormal families of states, and \( \alpha_i \) are real positive numbers. The bipartite entanglement entropy is defined as \( E_{A, B}(\ket{\psi}) = \sum_{i=1}^m -\alpha_i^2 \log (\alpha_i^2) \).

While this quantifies entanglement across a bisection, it does not capture the global entanglement in the system. The \textit{entanglement width}, introduced in \cite{PhysRevLett.97.150504}, is one way of characterizing the flat amount of entanglement present in a system.

\paragraph{Branch decomposition and entanglement width.}

A subcubic tree is a tree of degree at most 3. A branch decomposition of a finite set \( V \) is a pair \( (T, L) \) where \( T \) is a subcubic tree and \( L \) is a bijection between the leaves of \( T \) and \( V \). Removing any edge \( e \) of \( T \) splits \( T \) into two connected components, inducing a bipartition of the leaves of $T$, and, thanks to $L$, a bipartition $A^e, B^e$ of  \( V \). Given a state \( \ket{\psi} \) over \( V \) and a branch decomposition \( (T, L) \), we associate each edge \( e \) of \( T \) with the number \( E_{A^e, B^e} \). The width of \( e \) is \( E_{A^e, B^e} \), and the width of \( T \) is the maximal width of its edges. The entanglement width of \( \ket{\psi} \) is the minimum width over all branch decompositions of \( V \):

\[ ew(\ket{\psi}) = \min_T \max_{e \in T} E_{A^e, B^e}(\ket{\psi}) \]

This metric can't be ``cheated'' in the sense that \textit{any} balanced bipartition $A/B, |A|, |B| \geq \left\lceil \frac{n}{3}\right\rceil$ of the qubits of $\ket{\psi}$ will verify \( E_{A, B}(\ket{\psi}) \geq ew(\ket{\psi}) \).
It is also straightforward to notice that, due to the subcubic structure of the decomposition, any quantum state over \( n \) qubits has an entanglement width bounded by \( \lceil n/3 \rceil \).

\paragraph{Stabilizer states and entanglement width.}

Stabilizer states can be expressed as graph states transformed via local Clifford gates \cite{PhysRevA.69.022316}. Graph states are stabilizer states described by undirected graphs. Given a graph \( G \), the graph state \( \ket{G} \) is the unique state stabilized by \( \langle X_u \prod_{v \in \mathcal{N}_G(u)} Z_v, u \in V(G) \rangle \).

It happens that there exists some graph parameter, called \textit{rank-width} that coincides exactly with the entanglement width of graph-states, in the sense that for any graph $G$:

\[ ew(\ket{G}) = rw(G) \]

Rank-width is defined similarly to entanglement width but uses the rank over \( \mathbb{F}_2 \) of the bipartite adjacency matrix induced by a bipartition of its vertices. This equality derives from the fact that cut-rank function and the bipartite entanglement entropy of graph states coincide exactly \cite{hein2006entanglementgraphstatesapplications}.

For instance, a GHZ state over \( n \) qubits is locally equivalent to the graph state $\ket{K_n}$ (the complete graph over $n$ vertices) which has \( rw(K_n) = 1 \). In contrast, a random graph state has entanglement width close to \( \lceil n/3 \rceil \) \cite{rank_width_random}.

In our experiments we used random hardware-efficient Clifford circuits with depth \( 2n \), producing random stabilizer states.
To find graph states locally equivalent to our stabilizer states, we use a technique described in \cite{Aaronson_2004, Goubault_de_Brugi_re_2025} consisting of finding a set of $H$ gates turning the $X$ part of the stabilizers into a full rank sub-matrix. The $X$ is then diagonalized, effectively symmetrizing the $Z$ part. At this stage the $Z$ corresponds to the adjacency matrix of the locally equivalent graph state. We then used a rank-width lower bound algorithm \cite{rank_width_lower_bound} to compute entanglement width lower bounds for those states. Results show that all considered stabilizer states have rank-width close to the optimal width. This result is not surprising, considering the fact that depth $2n+2$ on a linear nearest neighbor architecture is enough to be able to produce \textit{any} graph state~\cite{maslov2018shorter}.

\begin{table}[h]
    \centering
    \begin{tabular}{c|c|c}
        n & $ew$ lower bound & $ew$ upper bound (\(\lceil n / 3 \rceil\)) \\
        \hline
        14 & 4 & 5\\
        20 & 5 & 6\\
        26 & 7 & 9 \\
        32 & 9 & 11\\
        38 & $10^\star$ & 13\\
        44 & $9^\star$ & 15\\
        50 & $8^\star$ & 17
    \end{tabular}
    \caption{\textbf{Entanglement bounds for the stabilizer states used in this work.} For each stabilizer state in our experiments we find a lower bound for the entanglement width of the state by calculating a rank-width lower bound for the corresponding graph state. Numbers flagged with~$^\star$ indicate that the lower bound estimation algorithm timed out after 24 hours.}
    \label{tab:rank-width-lw}
\end{table}

\begin{figure}[h]
    \centering
    \includegraphics[scale=.5]{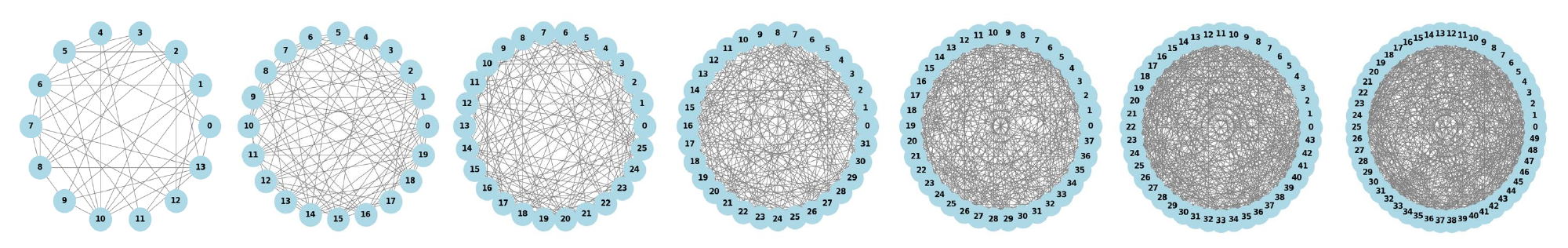}
    \caption{\textbf{Graphs corresponding to the stabilizer states produced in this work}. Each graph state is locally equivalent to the state we prepare via a $n \times 2n$ random Clifford circuit, for $n \in \{14, 20, 26, 32, 38, 44, 50\}$. Such circuits likely produce a random Erd\H{o}s--R\'enyi graph, which would saturate the entanglement upper bound and become challenging for tensor network simulators at large $n$.}
    \label{fig:fig-graph-states}
\end{figure}

\subsection{Estimating stabilizer state fidelity}

To evaluate the fidelity of our prepared state, we use direct fidelity estimation with Monte Carlo sampling\cite{PhysRevLett.106.230501,PhysRevLett.107.210404}. The fidelity \( F(\rho, \ket{\psi}) \) of the experimentally produced state \( \rho \) with respect to the target stabilizer state \( \ket{\psi} \) stabilized by \( \{S_1, \ldots, S_N\} \) is:

\[ F(\rho, \ket{\psi}) = \frac{1}{N} \sum_{i=1}^N \langle S_i \rangle_\rho \]

where \( \langle S_i \rangle_\rho = \text{Tr}(\rho S_i) \). This sum is subsampled by picking \( M \) random stabilizers \( S_1, \ldots, S_M \):

\[ F(\rho, \ket{\psi}) = \frac{1}{M} \sum_{i=1}^M \langle S_i \rangle_\rho + O(1/\sqrt{M}) \]

The standard error of this fidelity estimate is approximated using:

\[ SE \approx \sqrt{\frac{1}{M} \left( \hat{\sigma}_F^2 + \frac{1}{M} \sum_{i=1}^M \frac{\hat{\sigma}_i^2}{k_i} \right)} \]

where \( \hat{\sigma}_F \) is the standard deviation of the distribution of \( \langle S_i \rangle_\rho \), \( \hat{\sigma}_i \) is the standard deviation of the shot distributions used to evaluate \( \langle S_i \rangle_\rho \), and \( k_i \) is the number of shots used to sample \( \langle S_i \rangle_\rho \).

Since we are interested in single-shot error detection, we need to evaluate each bitstring for whether it has a logical error or not. This can be accomplished by computing the parity of the bitstring obtained from measuring the state in its stabilizer basis. The bitstring is logically correct if and only if it has even parity.

In our experiment, we chose to report the prepared state's fidelity as opposed to the logical error rate captured by our formalism. The two quantities are linearly related by the expression $F = 1-2\times LER$. Indeed, $LER$ captures the rate at which some stabilizer measurement fails. For a failure rate $r$, the corresponding stabilizer measurement expected value would be of $-r + (1-r) = 1 -2r$.

\section{Performance of checks in simulation and experiments} \label{sec:sim}

\begin{figure}
\centering
\includegraphics[width=\textwidth]{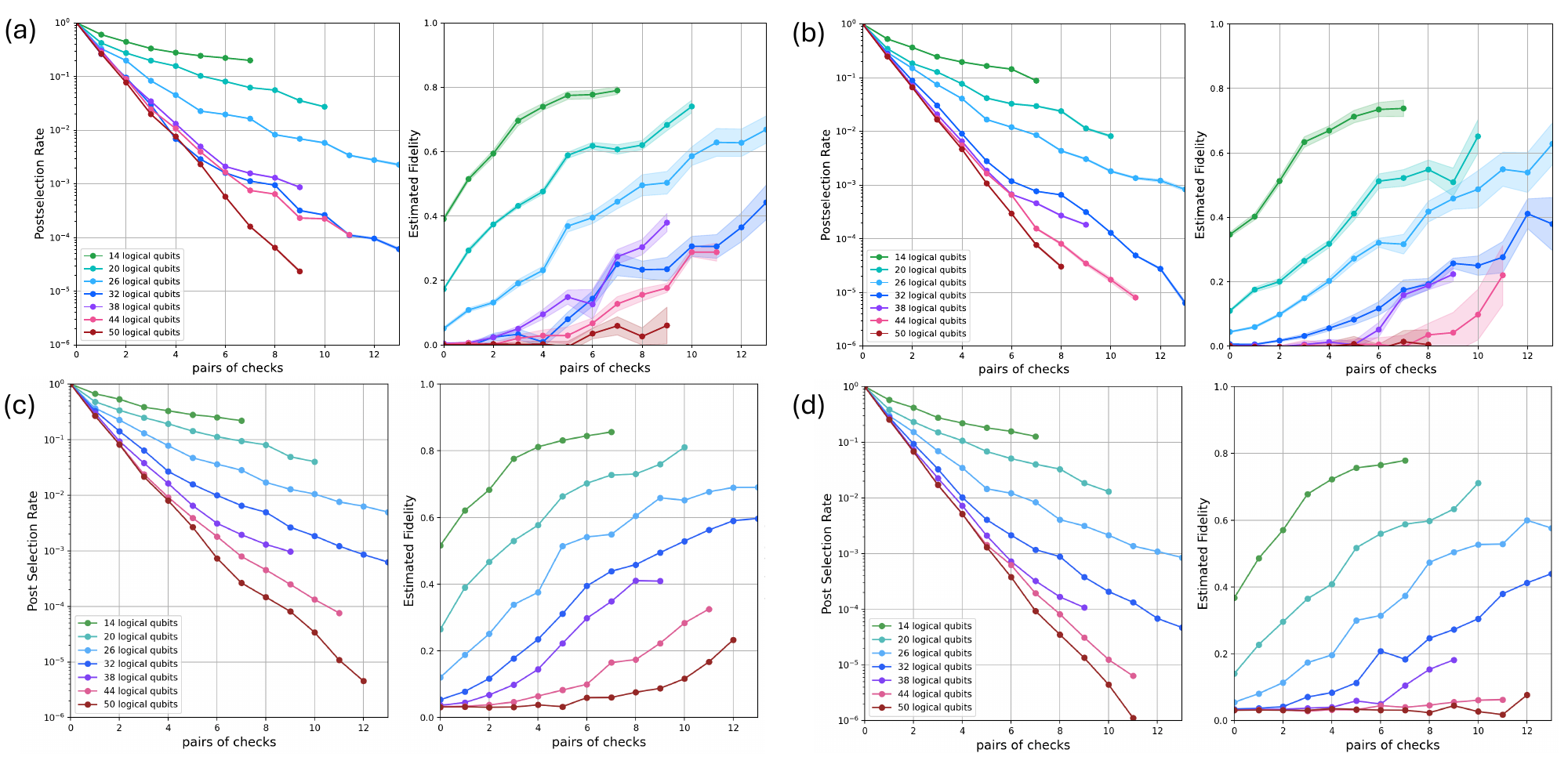}
\caption{\textbf{Experimental and numerical results.} The same circuits show better postselection rate and fidelity on $ibm\_kingston$ (a) compared to $ibm\_fez$ (b), which is consistent with the device parameters. A simple Pauli noise with gate error $0.3\%$ (c) and $0.5\%$ (d) and coherence $100~\mu\text{s}$ captures the experimental data well.}
\label{fig:kingston-vs-fez-vs-sims}
\end{figure}

In this section we show numerical simulations for the performance of spacetime Pauli checks, which exhibit good agreement with the experiments. We also present more experimental data from another IBM Heron r2 processor that shows how changes in device parameters change the performance of checks. Lastly we investigate how other qubit connectivities beyond what we had access to could change the performance of checks.

\paragraph{Simulation vs. experimental data} Figure~\ref{fig:kingston-vs-fez-vs-sims} shows experimental results obtained from $ibm\_kingston$ and $ibm\_fez$. The former was already presented in Figure 4 of the main text, but repeated here for comparison. We see that $ibm\_kingston$ has higher post-selection rate and better fidelity compared to $ibm\_fez$, which we attribute to better coherence, gates and readouts, see Table~\ref{tab:device}.

We also include numerical simulation for the same circuits with two different physical gate error rates. We see that a simple Pauli noise model has good qualitative agreement with the experimental data, further validating our Monte Carlo approach for scoring checks as detailed in Section~\ref{sec:cumulant}.

\begin{figure}[h!]
    \centering
    \begin{minipage}[c]{0.25\textwidth}
        \centering
        \begin{subfigure}[t]{\textwidth}
            \centering
            \includegraphics[width=0.65\textwidth]{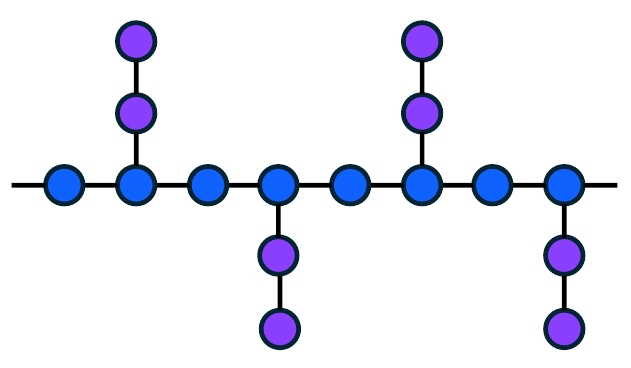}
            \caption{}
        \end{subfigure}

        \vspace{1em} 

        \begin{subfigure}[t]{\textwidth}
            \centering
            \includegraphics[width=0.65\textwidth]{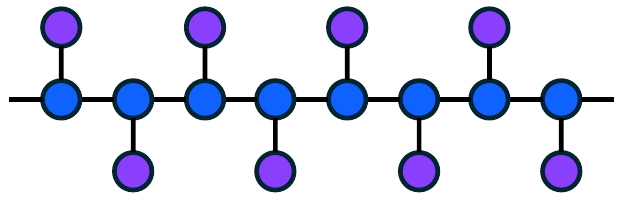}
            \caption{}
        \end{subfigure}

    \end{minipage}
    \hfill
    \begin{minipage}[c]{0.74\textwidth}
        \centering
        \begin{subfigure}[c]{\textwidth}
            \centering
            \includegraphics[width=\textwidth]{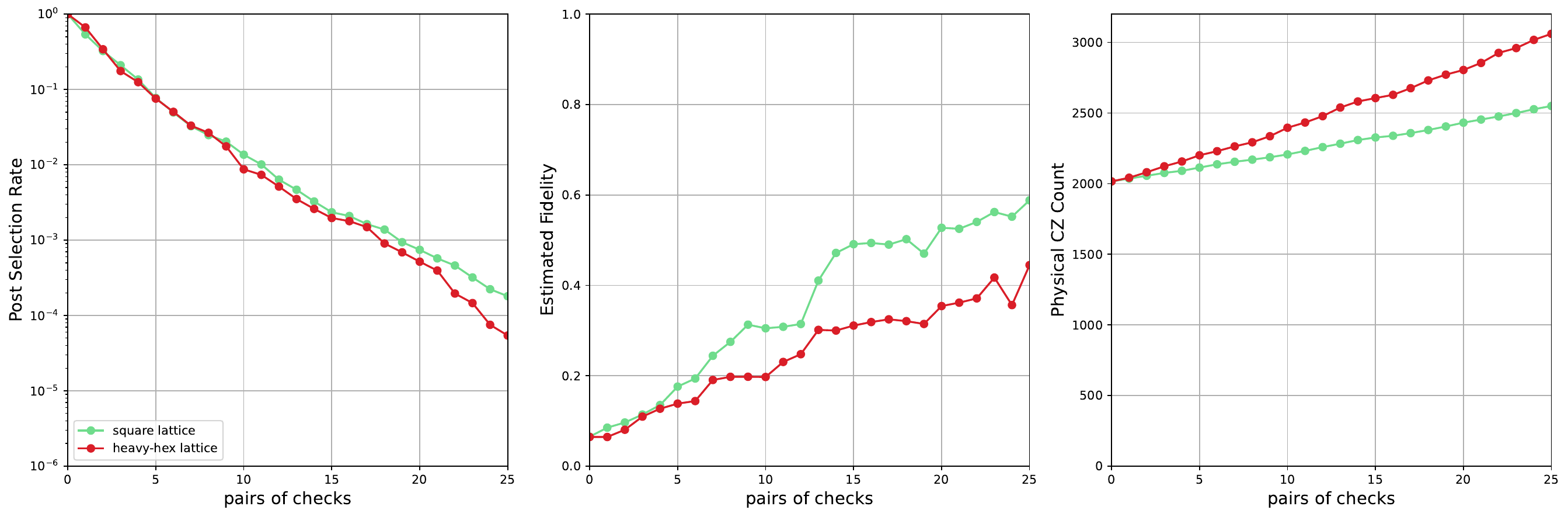}
            \caption{}
        \end{subfigure}
    \end{minipage}

    \caption{\textbf{Impact of connectivity on the performance of checks}. Error detection is simulated for a $64 \times 64$ circuit on two different qubit connectivities (a) and (b) which are compatible with a heavy hex lattice and square lattice respectively. Better connectivity results in better error detection performance.}
    \label{fig:square_vs_heavyhex}
\end{figure}

\paragraph{Impact of connectivity} The IBM Heron devices we used in this work have a heavy hex qubit connectivity, which lets us attach two ancilla to every other data qubit. In this way the ancilla pair can detect errors among themselves as well. However in simulation we also studied how higher connectivity would impact the performance of our method. For this purpose we simulate the same 64-qubit, depth-64 Clifford circuit once on the heavy hex lattice and once on a square lattice. As shown in Figure~\ref{fig:square_vs_heavyhex}, we see that better connectivity indeed improves the performance, even if now only a single ancilla is attached to each data. We attribute this to the increased number of accessible wires, which allows the checks to probe the payload at more points in space. These simulations were performed with a $0.1\%$ depolarizing error on $CZ$ gates and depolarizing error on idle wires of strength $1 - e^{(-t/T)}$ where $T=100\,\mu\text{s}$.

\section{Experimental methods}

In this section we provide an overview of the quantum devices used in our experiments, and we present the technique we employed to do single-shot fidelity estimation in a manner that is robust to state preparation and measurement (SPAM) errors.

\subsection{Device overview}
The primary device used in our experiments is $ibm\_kingston$, and we also report data from $ibm\_fez$. Both are 156-qubit Heron r2 processors from IBM, consisting of fixed-frequency transmon qubits with tunable couplers~\cite{stehlik2021tunable}.
Table~\ref{tab:device} summarizes the qubit and gate properties at the time of each experiment. The $ibm\_kingston$ device has an overall better gate error, readout error and coherence, which is also reflected in the experimental results in Figure~\ref{fig:kingston-vs-fez-vs-sims} and Figure~\ref{fig:m3-vs-nom3}. All experiments were conducted using Qiskit~\cite{javadi2024quantum} via the IBM Quantum Platform. We used the Qiskit Runtime Sampler, which returns a bitstring for each shot (sample), which we then accept or discard based on the syndrome that it reveals.

\begin{table}[h!]
\centering
\renewcommand{\arraystretch}{1.1}
\begin{tabular}{c|*{4}{*{4}{c}}}
\toprule
\multicolumn{13}{c}{\textbf{\textit{ibm\_kingston}}} \\
\midrule
\multirow{2}{*}{Exp.} & \multicolumn{3}{c}{CZ Error (\%)} & \multicolumn{3}{c}{Readout Error (\%)} & \multicolumn{3}{c}{T1 ($\mu\text{s}$)} & \multicolumn{3}{c}{T2 ($\mu\text{s}$)} \\
\cmidrule(lr){2-4} \cmidrule(lr){5-7} \cmidrule(lr){8-10} \cmidrule(lr){11-13}
& Med & Min & Max
& Med & Min & Max 
& Med & Min & Max 
& Med & Min & Max \\
\midrule
14 & 0.26 & 0.12 & 0.59 & 0.63 & 0.27 & 1.61 & 246.5 & 43.7 & 393.4 & 147.2 & 11.7 & 360.3 \\
20 & 0.26 & 0.12 & 0.61 & 0.72 & 0.17 & 5.93 & 247.5 & 120.8 & 435.4 & 159.2 & 10.1 & 305.2 \\
26 & 0.18 & 0.11 & 0.95 & 0.5 & 0.27 & 6.28 & 230.4 & 67.6 & 362.6 & 153.5 & 10.9 & 427.3 \\
32 & 0.27 & 0.14 & 3.67 & 0.57 & 0.15 & 5.93 & 257.9 & 21.6 & 406.4 & 145.6 & 10.1 & 458.4 \\
38 & 0.23 & 0.14 & 0.94 & 0.56 & 0.15 & 3.49 & 255.6 & 160.8 & 435.4 & 159.2 & 10.1 & 399.4 \\
44 & 0.25 & 0.14 & 1.32 & 0.61 & 0.2 & 4.93 & 220.7 & 43.3 & 381.7 & 113.8 & 10.3 & 320.7 \\
50 & 0.27 & 0.14 & 1.32 & 0.59 & 0.2 & 5.57 & 215.1 & 43.0 & 369.7 & 111.6 & 10.3 & 320.7 \\
\bottomrule

\toprule
\multicolumn{13}{c}{\textbf{\textit{ibm\_fez}}} \\
\midrule
\multirow{2}{*}{Exp.} & \multicolumn{3}{c}{CZ Error (\%)} & \multicolumn{3}{c}{Readout Error (\%)} & \multicolumn{3}{c}{T1 ($\mu\text{s}$)} & \multicolumn{3}{c}{T2 ($\mu\text{s}$)} \\
\cmidrule(lr){2-4} \cmidrule(lr){5-7} \cmidrule(lr){8-10} \cmidrule(lr){11-13}
& Med & Min & Max
& Med & Min & Max 
& Med & Min & Max 
& Med & Min & Max \\
\midrule
14 & 0.33 & 0.22 & 0.55 & 0.88 & 0.32 & 8.35 & 144.5 & 68.6 & 221.0 & 94.5 & 21.7 & 289.2 \\
20 & 0.39 & 0.23 & 100.0 & 1.0 & 0.3 & 6.7 & 158.3 & 61.5 & 238.0 & 85.1 & 11.2 & 251.6 \\
26 & 0.3 & 0.19 & 1.14 & 1.88 & 0.42 & 6.64 & 109.1 & 34.1 & 243.2 & 102.9 & 11.9 & 224.5 \\
32 & 0.38 & 0.24 & 3.68 & 0.93 & 0.37 & 11.65 & 150.5 & 85.5 & 240.7 & 98.3 & 13.2 & 256.0 \\
38 & 0.35 & 0.24 & 1.3 & 0.68 & 0.24 & 9.18 & 117.8 & 22.8 & 296.9 & 84.9 & 12.0 & 239.8 \\
44 & 0.34 & 0.21 & 0.82 & 0.68 & 0.15 & 10.6 & 145.4 & 83.8 & 291.1 & 81.3 & 14.5 & 263.2 \\
50 & 0.35 & 0.19 & 3.22 & 0.71 & 0.29 & 8.23 & 151.4 & 61.7 & 265.6 & 97.8 & 10.1 & 275.0 \\
\bottomrule
\end{tabular}
\caption{Device properties at the time each experiment reported in this paper was performed.}
\label{tab:device}
\end{table}

\subsection{Readout error mitigation}\label{sec:readout}

Readout errors are generally \( 5-10\times \) worse than gate errors, presenting a challenge in verifying the fidelity of prepared states without being overwhelmed by SPAM errors. Recall that stabilizer fidelity is determined by the parity of bits in the measured bitstring (see section~\ref{sec:stabilizer-state}). For an \( n \)-qubit stabilizer state with average stabilizer weight \( 3n/4 \), the logical error rate due to readout error rate \( p \) is:

\[ \sum_{k \text{ odd}} \binom{3n/4}{k} p^k (1-p)^{n-k} = \frac{1}{2}(1-(1-2p)^n) \]

which translates to a state fidelity of $(1-2p)^n$. For example, a 50-qubit graph state experiencing a \( 1\% \) readout error would show a fidelity of \( 36\% \), even in the absence of any gate errors. 

To estimate state fidelities in a SPAM-robust way, we reduce the number of measured qubits by performing part of the parity computation on the quantum circuit rather than in classical postprocessing.  By ``folding'' the stabilizer onto fewer qubits, we can simultaneously reduce the number of measurements and avoid bad readouts. To do this we divide the payload's support graph into $k$ subgraphs, picking the best readout within each subgraph as the root, and performing CNOTs from the leaves to the root in reverse breadth-first order. The measurements from all the roots are then used to compute the full parity classically. In our experiments we set $k = 4$, which creates 16 unique bitstrings in the output measurements. Figure~\ref{fig:folding} illustrates this process, as well as the tradeoff between circuit depth and number of measurements as a function of $k$ (number of qubits to fold onto). Note that the gates we use for folding are not themselves checked by detectors, so we do not artificially boost fidelity and therefore get a true lower bound on fidelity.

\begin{figure}
\centering
\begin{subfigure}{.18\textwidth}
\includegraphics[width=\textwidth]{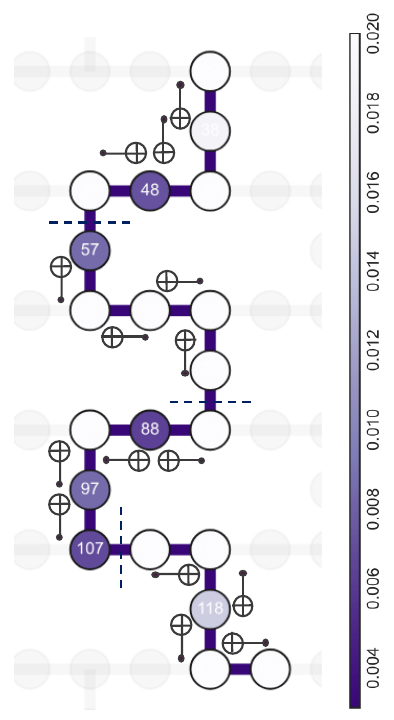}
\end{subfigure}
\hfill
\begin{subfigure}{.72\textwidth}
\includegraphics[width=\textwidth]{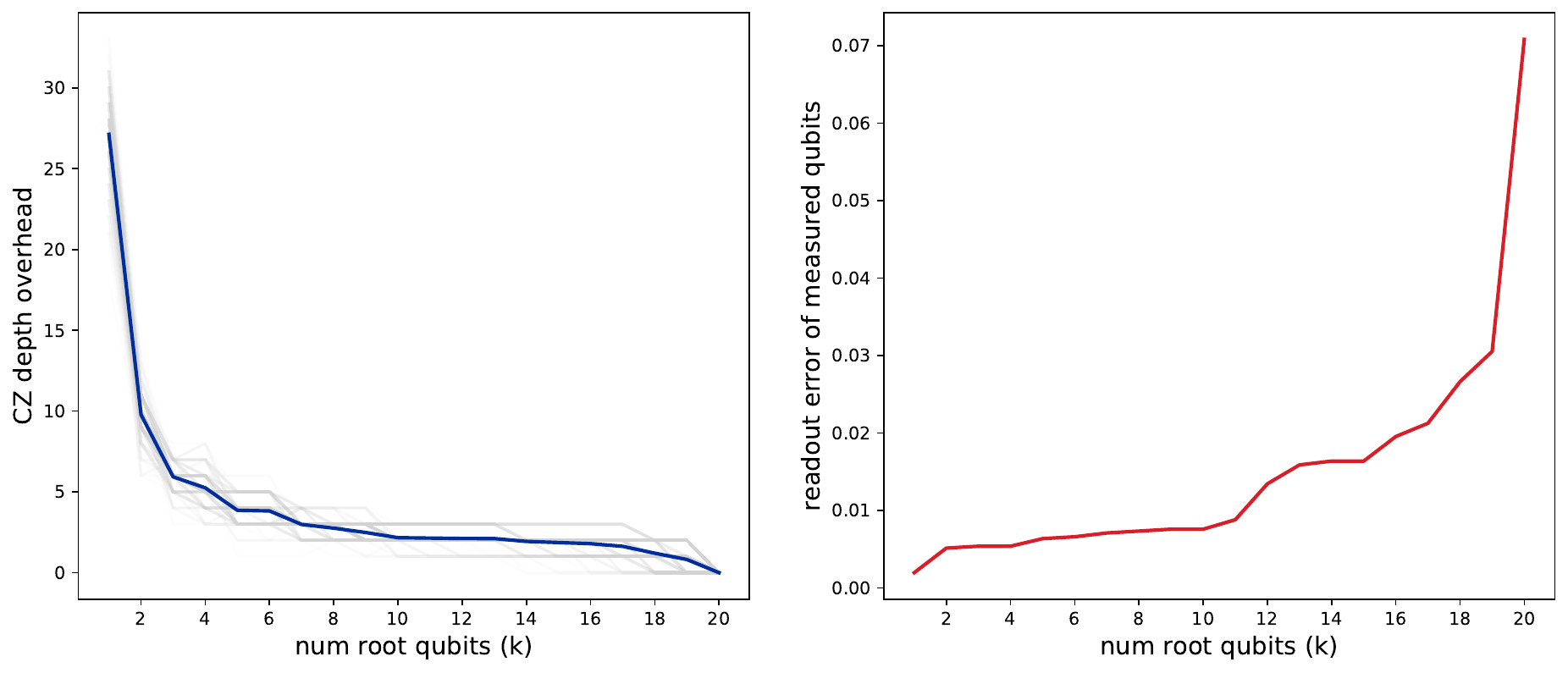}
\end{subfigure}
\caption{\textbf{Folding a stabilizer before measurement to reduce readout errors.} (a) A 20-qubit example where the graph is partitioned into 4 folds, and CNOTs are used to compute the parity within each subgraph onto the root with the best readout error (darker qubits). The state's parity is computed classically by measuring the roots only. (b) While this procedure removes bad readouts, it introduces some depth overhead to the measurement. The plot shows this tradeoff averaged over 200 random Pauli measurements for the readout errors in (a). Choosing a small number of roots to fold onto presents a good tradeoff in practice.}
\label{fig:folding}
\end{figure}

Having reduced the number of unique bitstrings, we can also perform statistical corrections on the output bitstring probabilities, which would otherwise be infeasible for a flat distribution with many unique bitstrings.  We use the M3 method~\cite{nation2021scalable} for this purpose which characterizes readout errors to build a model of readout transition probabilities and then statistically corrects the bitstring populations in post-processing. Due to the fact that this method is not compatible with our objective of single-shot error detection, we avoided reporting these results in the main text. However for completeness we report them here, which shows that these correction do produce some boost in the estimated fidelity. However this is not very significant since we avoided the worst readouts during stabilizer folding. We see the boost due to statistical corrections to be larger for $ibm\_fez$, which is consistent with its worse readout.

\begin{figure}
\centering
\includegraphics[width=\textwidth]{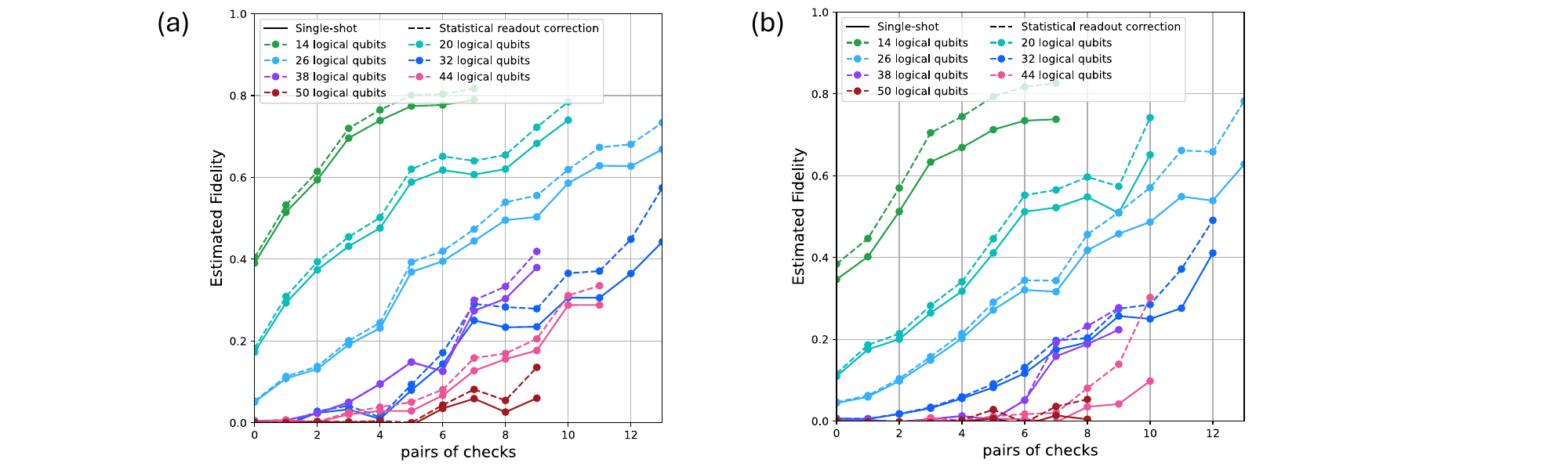}
\caption{\textbf{Single-shot readout (solid) vs. statistically corrected readouts (dashed).} Our stabilizer folding approach already reduces the impact of readout errors, so there is not a wide gap in fidelity when estimated on single-shot basis vs. statistically corrected distributions. On $ibm\_kingston$ (a) the readout errors are generally better than $ibm\_fez$ (b), leading to a smaller gap (as well as a better post-selection rate, see Figure~\ref{fig:kingston-vs-fez-vs-sims}.)}
\label{fig:m3-vs-nom3}
\end{figure}

\subsection{Further experimental details}

\paragraph{Circuit generation}
Bare payload circuits are generated by alternating layers of $CZ$ gates on a path of $n$ qubits with randomly-chosen single-qubit Cliffords. On a path there are two unique $CZ$ layers, and we repeat each $n$ times for a total 2-qubit depth of $2n$,  which is  asymptotically sufficient to implement any graph state in the LNN setting~\cite{maslov2018shorter}.

We pick ancilla for checks from the qubits that dangle from data qubits. On the heavy hex architecture, every other qubit is a degree-3 node and can thus be an anchor for a check. Therefore for $n$ logical qubits we have $\left\lfloor \frac{n}{2} \right\rfloor$ anchors. 

In the experiments we progressively add more ancilla and record postselection rate and estimate state fidelity. We pick the order of the ancilla in a binary tree fashion, first connecting to the middle of the payload and then recursively to the subgraphs induced. This maximizes the coverage of wires in the earlier rounds, allowing us to see more boost in the signal while the post-selection rate is still high, which is particularly important for large circuits whose postselection rate drops rapidly. In each experiment we add rounds of checks until either we run out of them, or until the postselection rate drops below $10^{-5}$. For each stabilizer state we estimate its fidelity by choosing 5 random stabilizers, which we measure at each round of checks.

\paragraph{Qubit selection}
Due to the size of circuits we consider, it is important that they are supported on physical qubits with low errors in order to have high postselection rates and low logical error rates. Prior to performing the experiments, we enumerate all subgraphs on the hardware that are isomorphic to the interaction graph of our circuits, which includes the payload as well as the connected ancilla. We then weigh the subgraphs by the average CZ error rate in them. We neglect readout error in this ranking, since readout errors on the payload can be mitigated as detailed in the next section. Moreover, readout errors on the checks primarily affect postselection, since the rejection criterion is essentially the {\tt OR} of check outcomes which is much more likely to reject good shots than to accept bad shots. 

Our experiments can be divided into a small regime ($n \in \{14, 20, 26, 32\}$) and a large regime ($n \in \{38, 44, 50\}$). In the small regime  we add checks in pairs (i.e. as an ancilla path of length 2 as detailed in Section~\ref{sec:multicheck}), while in the large regime we add them individually. This is to keep the number of qubits reasonable and compatible with the hardware at hand. In all plots, the X axis corresponds to a pair of checks. One can notice in Figure 4 of the main text that in the large regime the slope of number of gates grows more mildly because both checks are directly attached to the payload and no SWAPs are used. Since the small regime has more dangling qubits, it can suppress more errors (including in the ancilla qubits themselves as detailed before), but it becomes harder to avoid some bad qubits due to the more stringent connectivity requirements.

\paragraph{Sampling speed} 
Postselection rate in error detection is fundamentally upper bounded by the inverse of circuit fidelity, even in the ideal case where the check circuitry itself does not introduce new errors. Therefore to demonstrate error detection on the largest circuits of this work, we require fast circuit execution to generate many samples, and retain sufficiently many good samples after postselection. Fortunately superconducting quantum computers have high gate and reset speeds, which make them ideally suited for error detection experiments. At the time of our experiments, the $ibm\_kingston$ and $ibm\_fez$ quantum computers had a ``circuit layer operations per second'' (CLOPS) rate of 250K and 195K respectively~\cite{wack2021scale}. We take as many samples as necessary to reduce the error bar to a reasonable amount. This means more samples are needed at lower post selection rates. We performed up to 6 million shots per single stabilizer measurement, which requires about 25 minutes of QPU time.

\paragraph{Dynamical decoupling} The payload circuits that we consider are dense, however the inclusion of check qubits means that the payload is probed via controlled-Pauli gates at a few distinct points in time, causing long idling gaps on the control (check) qubit.  We use a simple \( X_+-X_- \) dynamical decoupling sequence~\cite{Hahn1950,Viola1998} to suppress dephasing errors during idle periods of the circuit. Crosstalk is suppressed at the physical level via tunable couplers, so we did not use more complex circuit-level DD sequences~\cite{seif2024suppressing,coote2025resource}.

\bibliography{biblio}
\bibliographystyle{naturemag}

%% file: figures/overview.tex
  \begin{tikzpicture}
    
    \newsavebox{\stageone}
    \savebox{\stageone}{
        \draw (0, 0) node[circle, fill=ibmcyan](n1) {};
        \draw (0, -1* 0.6) node[circle, fill=ibmcyan](n2) {};
        \draw (0, -2* 0.6) node[circle, fill=ibmcyan](n3) {};
        \draw (0, -3* 0.6) node[circle, fill=ibmcyan](n4) {};
        \draw[ultra thick, ibmcyan] (n1) -- (n2) -- (n3) -- (n4);
        
        \draw (-1, -0.3* 0.6) node[rectangle, fill=ibmpurple](na) {};
        \draw [ibmpurple, ultra thick, dashed] (na) -- (n2);

        \draw[rounded corners, ultra thick, ibmcyan, fill=ibmcyan] (1, 0.5) rectangle node[white]{\Huge$C$} (6.4, -3.5 * 0.6);
        
        \draw[ibmcyan, ultra thick] (1* 0.6, 0) -- (1, 0);
        \draw[ibmcyan, ultra thick] (1* 0.6, -1 * 0.6) -- (1, -1* 0.6);
        \draw[ibmcyan, ultra thick] (1* 0.6, -2* 0.6) -- (1, -2* 0.6);
        \draw[ibmcyan, ultra thick] (1* 0.6, -3* 0.6) -- (1, -3* 0.6);

        \draw[ibmcyan, ultra thick] (6.4, 0) -- (6.7, 0);
        \draw[ibmcyan, ultra thick] (6.4, -1* 0.6) -- (6.7, -1* 0.6);
        \draw[ibmcyan, ultra thick] (6.4, -2* 0.6) -- (6.7, -2* 0.6);
        \draw[ibmcyan, ultra thick] (6.4, -3* 0.6) -- (6.7, -3* 0.6);
     }
    \newsavebox{\stagetwo}
    \savebox{\stagetwo}{
        

        \draw[rounded corners, ultra thick, ibmcyan] (1, 0.5) rectangle (6.4, -3.5 * 0.6);
        
        \draw[ibmcyan, ultra thick] (1* 0.6, 0) -- (1, 0);
        \draw[ibmcyan, ultra thick] (1* 0.6, -1 * 0.6) -- (1, -1* 0.6);
        \draw[ibmcyan, ultra thick] (1* 0.6, -2* 0.6) -- (1, -2* 0.6);
        \draw[ibmcyan, ultra thick] (1* 0.6, -3* 0.6) -- (1, -3* 0.6);

        \draw[ibmcyan, ultra thick] (6.4, 0) -- (6.7, 0);
        \draw[ibmcyan, ultra thick] (6.4, -1* 0.6) -- (6.7, -1* 0.6);
        \draw[ibmcyan, ultra thick] (6.4, -2* 0.6) -- (6.7, -2* 0.6);
        \draw[ibmcyan, ultra thick] (6.4, -3* 0.6) -- (6.7, -3* 0.6);
        
        \draw[rounded corners, fill=ibmcyan, ibmcyan, path fading=south] (2.2* 0.6, -0.7* 0.6) rectangle (2.2* 0.6 + 0.4, -2.3* 0.6);
        \draw[rounded corners, fill=ibmred, ibmred] (2.2* 0.6 + 0.6, -0.8* 0.6+0.1) rectangle node[white]{$w_1$} (2.2* 0.6 + 1.2, -1.2* 0.6-0.1); 
        \draw[thick, ibmred] (2.2* 0.6 + 0.6, -1* 0.6) -- (2.2* 0.6 + 0.4, -1* 0.6);
        \draw[thick, ibmred] (2.2* 0.6 + 1.2, -1* 0.6) -- (2.2* 0.6 + 1.2 + 0.2, -1* 0.6);
        \draw[rounded corners, fill=ibmcyan, ibmcyan, path fading=north] (2.2* 0.6 + 1.4, 0.3* 0.6) rectangle (2.2* 0.6 + 1.8, -1.3* 0.6);
        \draw[rounded corners, fill=ibmred, ibmred] (2.2* 0.6 + 2, -0.8* 0.6+0.1) rectangle node[white]{$w_2$}(2.2* 0.6 + 2.6, -1.2* 0.6-0.1); 
        \draw[thick, ibmred] (2.2* 0.6 + 2, -1* 0.6) -- (2.2* 0.6 + 1.8, -1* 0.6);
        \draw[thick, ibmred] (2.2* 0.6 + 2.6, -1* 0.6) -- (2.2* 0.6 + 2.8, -1* 0.6);
        \draw[rounded corners, fill=ibmcyan, ibmcyan, path fading=south] (6-1, -0.7* 0.6) rectangle (6.4-1, -2.3* 0.6);
        \draw[rounded corners, fill=ibmred, ibmred] (6.6-1, -0.8* 0.6+0.1) rectangle node[white]{$w_k$}(7.2-1, -1.2* 0.6-0.1); 
        \draw[thick, ibmred] (6.6-1, -1* 0.6) -- (6.4-1, -1* 0.6);
        \draw[thick, ibmred] (7.4-1, -1* 0.6) -- (7.2-1, -1* 0.6);

        \draw (4.7, -1 * 0.6) node[ibmred] {$\cdots$};

     }
    
    \newsavebox{\stagethree}
    \savebox{\stagethree}{
        
        \draw[rounded corners, ultra thick, ibmcyan] (1, 0.5) rectangle (6.4, -3.5 * 0.6);
        
        \draw[ibmcyan, ultra thick] (1* 0.6, 0) -- (1, 0);
        \draw[ibmcyan, ultra thick] (1* 0.6, -1 * 0.6) -- (1, -1* 0.6);
        \draw[ibmcyan, ultra thick] (1* 0.6, -2* 0.6) -- (1, -2* 0.6);
        \draw[ibmcyan, ultra thick] (1* 0.6, -3* 0.6) -- (1, -3* 0.6);

        \draw[ibmcyan, ultra thick] (6.4, 0) -- (6.7, 0);
        \draw[ibmcyan, ultra thick] (6.4, -1* 0.6) -- (6.7, -1* 0.6);
        \draw[ibmcyan, ultra thick] (6.4, -2* 0.6) -- (6.7, -2* 0.6);
        \draw[ibmcyan, ultra thick] (6.4, -3* 0.6) -- (6.7, -3* 0.6);
        
        \draw[rounded corners, fill=ibmcyan, ibmcyan, path fading=south] (2.2* 0.6, -0.7* 0.6) rectangle (2.2* 0.6 + 0.4, -2.3* 0.6);
        \draw[rounded corners, fill=ibmred, ibmred] (2.2* 0.6 + 0.6, -0.8* 0.6+0.1) rectangle node[white]{$I$} (2.2* 0.6 + 1.2, -1.2* 0.6-0.1); 
        \draw[thick, ibmred] (2.2* 0.6 + 0.6, -1* 0.6) -- (2.2* 0.6 + 0.4, -1* 0.6);
        \draw[thick, ibmred] (2.2* 0.6 + 1.2, -1* 0.6) -- (2.2* 0.6 + 1.2 + 0.2, -1* 0.6);
        \draw[rounded corners, fill=ibmcyan, ibmcyan, path fading=north] (2.2* 0.6 + 1.4, 0.3* 0.6) rectangle (2.2* 0.6 + 1.8, -1.3* 0.6);
        \draw[rounded corners, fill=ibmred, ibmred] (2.2* 0.6 + 2, -0.8* 0.6+0.1) rectangle node[white]{$Y$}(2.2* 0.6 + 2.6, -1.2* 0.6-0.1); 
        \draw[thick, ibmred] (2.2* 0.6 + 2, -1* 0.6) -- (2.2* 0.6 + 1.8, -1* 0.6);
        \draw[thick, ibmred] (2.2* 0.6 + 2.6, -1* 0.6) -- (2.2* 0.6 + 2.8, -1* 0.6);
        \draw[rounded corners, fill=ibmcyan, ibmcyan, path fading=south] (6-1, -0.7* 0.6) rectangle (6.4-1, -2.3* 0.6);
        \draw[rounded corners, fill=ibmred, ibmred] (6.6-1, -0.8* 0.6+0.1) rectangle node[white]{$X$}(7.2-1, -1.2* 0.6-0.1); 
        \draw[thick, ibmred] (6.6-1, -1* 0.6) -- (6.4-1, -1* 0.6);
        \draw[thick, ibmred] (7.4-1, -1* 0.6) -- (7.2-1, -1* 0.6);

        \draw (4.7, -1 * 0.6) node[ibmred] {$\cdots$};
     }

    \newsavebox{\firstcheck}
    \savebox{\firstcheck}{
        \draw[rounded corners, ultra thick, ibmcyan] (1-0.34, 0.5) rectangle (6.4, -3.5 * 0.6);
        
        \draw[ibmcyan, ultra thick] (1* 0.6-0.34, 0) -- (1-0.34, 0);
        \draw[ibmcyan, ultra thick] (1* 0.6-0.34, -1 * 0.6) -- (1-0.34, -1* 0.6);
        \draw[ibmcyan, ultra thick] (1* 0.6-0.34, -2* 0.6) -- (1-0.34, -2* 0.6);
        \draw[ibmcyan, ultra thick] (1* 0.6-0.34, -3* 0.6) -- (1-0.34, -3* 0.6);

        \draw[ibmcyan, ultra thick] (6.4, 0) -- (6.7, 0);
        \draw[ibmcyan, ultra thick] (6.4, -1* 0.6) -- (6.7, -1* 0.6);
        \draw[ibmcyan, ultra thick] (6.4, -2* 0.6) -- (6.7, -2* 0.6);
        \draw[ibmcyan, ultra thick] (6.4, -3* 0.6) -- (6.7, -3* 0.6);
        
        \draw[rounded corners, fill=ibmcyan, ibmcyan] (2.2* 0.6, -0.7* 0.6) rectangle (2.2* 0.6 + 0.4, -2.3* 0.6);
        \draw[rounded corners, fill=ibmred, ibmred] (2.2* 0.6 + 0.6, -0.8* 0.6+0.1) rectangle node[white]{$Y$} (2.2* 0.6 + 1.2, -1.2* 0.6-0.1); 
        \draw[thick, ibmred] (2.2* 0.6 + 0.6, -1* 0.6) -- (2.2* 0.6 + 0.4, -1* 0.6);
        \draw[thick, ibmred] (2.2* 0.6 + 1.2, -1* 0.6) -- (2.2* 0.6 + 1.2 + 0.2, -1* 0.6);
        \draw[rounded corners, fill=ibmcyan, ibmcyan] (2.2* 0.6 + 1.4, 0.3* 0.6) rectangle (2.2* 0.6 + 1.8, -1.3* 0.6);
        \draw[rounded corners, fill=ibmred, ibmred] (2.2* 0.6 + 2, -0.8* 0.6+0.1) rectangle node[white]{$I$}(2.2* 0.6 + 2.6, -1.2 * 0.6-0.1); 
        \draw[thick, ibmred] (2.2* 0.6 + 2, -1* 0.6) -- (2.2* 0.6 + 1.8, -1* 0.6);
        \draw[thick, ibmred] (2.2* 0.6 + 2.6, -1* 0.6) -- (2.2* 0.6 + 2.8, -1* 0.6);
        \draw[rounded corners, fill=ibmcyan, ibmcyan] (6-1, -0.7* 0.6) rectangle (6.4-1, -2.3* 0.6);
        \draw[rounded corners, fill=ibmred, ibmred] (6.6-1, -0.8* 0.6+0.1) rectangle node[white]{$X$}(7.2-1, -1.2* 0.6-0.1); 
        \draw[thick, ibmred] (6.6-1, -1* 0.6) -- (6.4-1, -1* 0.6);
        \draw[thick, ibmred] (7.4-1, -1* 0.6) -- (7.2-1, -1* 0.6);

        \draw[rounded corners, fill=ibmcyan, ibmcyan] (6-1.2-0.36-0.2-0.4 + 0.2, -1.7* 0.6) rectangle (6-1.2-0.36, -3.3* 0.6);
        \draw[rounded corners, fill=ibmcyan, ibmcyan, path fading=north] (6-1.2-0.36-0.2-0.4 - 1.7, -1.7* 0.6) rectangle (6-1.2-0.36-0.2-0.4 - 1.3, -3.3* 0.6);
        
        \draw[ibmgreen, fill=ibmgreen] (6-1.1-0.36, -2 * 0.6 + 0.18) rectangle node[white]{\tiny$X$} (6-1.1, -2 * 0.6 - 0.18);
        \draw[ibmgreen, fill=ibmgreen] (6-1.1-0.36, -1 * 0.6 + 0.18) rectangle node[white]{\tiny$Z$} (6-1.1, -1 * 0.6 - 0.18);
        \draw[ibmgreen, fill=ibmgreen] (6-1.2-0.36-0.3 - 0.56, -3 * 0.6 + 0.18) rectangle node[white]{\tiny$Y$} (6-1.2-0.36-0.2-0.5 + 0.2, -3 * 0.6 - 0.18);
        \draw[ibmgreen, fill=ibmgreen] (6-1.2-0.36-0.2-0.4 - 1.3 + 0.1, -3 * 0.6 + 0.18) rectangle node[white]{\tiny$Z$} (6-1.2-0.36-0.2-0.4 - 1.3 + 0.46, -3 * 0.6 - 0.18);
        \draw[ibmgreen, fill=ibmgreen] (6-1.2-0.36-0.2-0.4 - 1.7 - 0.46, -3 * 0.6 + 0.18) rectangle node[white]{\tiny$Z$} (6-1.2-0.36-0.2-0.4 - 1.7 - 0.1, -3 * 0.6 - 0.18);

        \draw[ibmgreen, fill=ibmgreen] (2.2* 0.6 + 1.4 - 0.46, -0 * 0.6 + 0.18) rectangle node[white]{\tiny$Z$} (2.2* 0.6 + 1.4- 0.1, -0 * 0.6 - 0.18);

        \draw[ibmgreen, fill=ibmgreen] (2.2* 0.6 - 0.46, -1 * 0.6 + 0.18) rectangle node[white]{\tiny$Z$} (2.2* 0.6-0.1, -1 * 0.6 - 0.18);
        \draw[ibmgreen, fill=ibmgreen] (2.2* 0.6 - 0.46, -2 * 0.6 + 0.18) rectangle node[white]{\tiny$Z$} (2.2* 0.6-0.1, -2 * 0.6 - 0.18);
        \draw[ibmgreen, fill=ibmgreen] (2.2* 0.6 - 0.46, -3 * 0.6 + 0.18) rectangle node[white]{\tiny$Z$} (2.2* 0.6-0.1, -3 * 0.6 - 0.18);
        \draw[ibmgreen, fill=ibmgreen] (2.2* 0.6 - 0.46, -0 * 0.6 + 0.18) rectangle node[white]{\tiny$Z$} (2.2* 0.6-0.1, -0 * 0.6 - 0.18);

        \draw (2.2* 0.6 + 1.4 - 0.46 - 0.5, -0 * 0.6) node[ibmgreen] {\bf $\cdots$};
        \draw (6-1.2-0.36-0.2-0.4 - 1.3 + 0.1 + 0.67, -3 * 0.6) node[ibmgreen] {\bf $\cdots$};
        \draw (2.2* 0.6 + 1.4 - 0.46 - 0.7, -3 * 0.6) node[ibmgreen] {\bf $\cdots$};
     }

    \newsavebox{\secondcheck}
    \savebox{\secondcheck}{
        
         \draw[rounded corners, ultra thick, ibmcyan] (1-0.34, 0.5) rectangle (6.4, -3.5 * 0.6);
        
        \draw[ibmcyan, ultra thick] (1* 0.6-0.34, 0) -- (1-0.34, 0);
        \draw[ibmcyan, ultra thick] (1* 0.6-0.34, -1 * 0.6) -- (1-0.34, -1* 0.6);
        \draw[ibmcyan, ultra thick] (1* 0.6-0.34, -2* 0.6) -- (1-0.34, -2* 0.6);
        \draw[ibmcyan, ultra thick] (1* 0.6-0.34, -3* 0.6) -- (1-0.34, -3* 0.6);

        \draw[ibmcyan, ultra thick] (6.4, 0) -- (6.7, 0);
        \draw[ibmcyan, ultra thick] (6.4, -1* 0.6) -- (6.7, -1* 0.6);
        \draw[ibmcyan, ultra thick] (6.4, -2* 0.6) -- (6.7, -2* 0.6);
        \draw[ibmcyan, ultra thick] (6.4, -3* 0.6) -- (6.7, -3* 0.6);
        \draw[rounded corners, fill=ibmcyan, ibmcyan, path fading=south] (2.2* 0.6, -0.7* 0.6) rectangle (2.2* 0.6 + 0.4, -2.3* 0.6);
        \draw[rounded corners, fill=ibmred, ibmred] (2.2* 0.6 + 0.6, -0.8* 0.6+0.1) rectangle node[white]{$I$} (2.2* 0.6 + 1.2, -1.2* 0.6-0.1); 
        \draw[thick, ibmred] (2.2* 0.6 + 0.6, -1* 0.6) -- (2.2* 0.6 + 0.4, -1* 0.6);
        \draw[thick, ibmred] (2.2* 0.6 + 1.2, -1* 0.6) -- (2.2* 0.6 + 1.2 + 0.2, -1* 0.6);
        \draw[rounded corners, fill=ibmcyan, ibmcyan, path fading=north] (2.2* 0.6 + 1.4, 0.3* 0.6) rectangle (2.2* 0.6 + 1.8, -1.3* 0.6);
        \draw[rounded corners, fill=ibmred, ibmred] (2.2* 0.6 + 2, -0.8* 0.6+0.1) rectangle node[white]{$Z$}(2.2* 0.6 + 2.6, -1.2 * 0.6-0.1); 
        \draw[thick, ibmred] (2.2* 0.6 + 2, -1* 0.6) -- (2.2* 0.6 + 1.8, -1* 0.6);
        \draw[thick, ibmred] (2.2* 0.6 + 2.6, -1* 0.6) -- (2.2* 0.6 + 2.8, -1* 0.6);
        \draw[rounded corners, fill=ibmcyan, ibmcyan] (6-1, -0.7* 0.6) rectangle (6.4-1, -2.3* 0.6);
        \draw[rounded corners, fill=ibmred, ibmred] (6.6-1, -0.8* 0.6+0.1) rectangle node[white]{$X$}(7.2-1, -1.2* 0.6-0.1); 
        \draw[thick, ibmred] (6.6-1, -1* 0.6) -- (6.4-1, -1* 0.6);
        \draw[thick, ibmred] (7.4-1, -1* 0.6) -- (7.2-1, -1* 0.6);

        \draw[rounded corners, fill=ibmcyan, ibmcyan] (6-1.2-0.36-0.2-0.4 + 0.2, -1.7* 0.6) rectangle (6-1.2-0.36, -3.3* 0.6);
        \draw[rounded corners, fill=ibmcyan, ibmcyan, path fading=north] (6-1.2-0.36-0.2-0.4 - 1.7, -1.7* 0.6) rectangle (6-1.2-0.36-0.2-0.4 - 1.3, -3.3* 0.6);
        
        \draw[ibmgreen, fill=ibmgreen] (6-1.1-0.36, -2 * 0.6 + 0.18) rectangle node[white]{\tiny$X$} (6-1.1, -2 * 0.6 - 0.18);
        \draw[ibmgreen, fill=ibmgreen] (6-1.1-0.36, -1 * 0.6 + 0.18) rectangle node[white]{\tiny$Z$} (6-1.1, -1 * 0.6 - 0.18);
        \draw[ibmgreen, fill=ibmgreen] (6-1.2-0.36-0.3 - 0.56, -3 * 0.6 + 0.18) rectangle node[white]{\tiny$Y$} (6-1.2-0.36-0.2-0.5 + 0.2, -3 * 0.6 - 0.18);
        \draw[ibmgreen, fill=ibmgreen] (6-1.2-0.36-0.2-0.4 - 1.3 + 0.1, -3 * 0.6 + 0.18) rectangle node[white]{\tiny$Z$} (6-1.2-0.36-0.2-0.4 - 1.3 + 0.46, -3 * 0.6 - 0.18);
        \draw[ibmgreen, fill=ibmgreen] (6-1.2-0.36-0.2-0.4 - 1.7 - 0.46, -3 * 0.6 + 0.18) rectangle node[white]{\tiny$Z$} (6-1.2-0.36-0.2-0.4 - 1.7 - 0.1, -3 * 0.6 - 0.18);
        \draw[ibmgreen, fill=ibmgreen] (2.2* 0.6 - 0.46, -3 * 0.6 + 0.18) rectangle node[white]{\tiny$Z$} (2.2* 0.6-0.1, -3 * 0.6 - 0.18);
        \draw (6-1.2-0.36-0.2-0.4 - 1.3 + 0.1 + 0.67, -3 * 0.6) node[ibmgreen] {\bf $\cdots$};
     }

    \newsavebox{\stagethreehalf}
    \savebox{\stagethreehalf}{
        \draw[ibmblue!30, very thin, rounded corners] (-0.8, 1.6) rectangle (15, -2.7);
        \draw (-0.5, 0) node {\usebox{\firstcheck}};
        \draw (2.5, 1) node[ibmgreen]{\huge$\operatorname{LER}_1$};
        \draw (7.3+3.5, 1) node[ibmgreen]{\huge$\operatorname{LER}_m$};
        \draw(7.3, -1) node[ibmblue]{\Huge$\cdots$};
        \draw (8, 0) node {\usebox{\secondcheck}};
     }

    \newsavebox{\stagefour}
    \savebox{\stagefour}{
       \draw (0, 0) node[circle, fill=ibmcyan](n1) {};
        \draw (0, -1* 0.6) node[circle, fill=ibmcyan](n2) {};
        \draw (0, -2* 0.6) node[circle, fill=ibmcyan](n3) {};
        \draw (0, -3* 0.6) node[circle, fill=ibmcyan](n4) {};
        \draw[ultra thick, ibmcyan] (n1) -- (n2) -- (n3) -- (n4);
        
        \draw (-1, -0.3* 0.6) node[rectangle, fill=ibmpurple](na) {};
        \draw [ibmpurple, ultra thick, dashed] (na) -- (n2);

        \draw[rounded corners, ultra thick, ibmcyan] (1, 0.5) rectangle (6.4, -3.5 * 0.6);
        
        \draw[ibmcyan, ultra thick] (1* 0.6, 0) -- (1, 0);
        \draw[ibmcyan, ultra thick] (1* 0.6, -1 * 0.6) -- (1, -1* 0.6);
        \draw[ibmcyan, ultra thick] (1* 0.6, -2* 0.6) -- (1, -2* 0.6);
        \draw[ibmcyan, ultra thick] (1* 0.6, -3* 0.6) -- (1, -3* 0.6);

        \draw[ibmcyan, ultra thick] (6.4, 0) -- (6.7, 0);
        \draw[ibmcyan, ultra thick] (6.4, -1* 0.6) -- (6.7, -1* 0.6);
        \draw[ibmcyan, ultra thick] (6.4, -2* 0.6) -- (6.7, -2* 0.6);
        \draw[ibmcyan, ultra thick] (6.4, -3* 0.6) -- (6.7, -3* 0.6);
        
        \draw[rounded corners, fill=ibmcyan, ibmcyan, path fading=south] (2.2* 0.6, -0.7* 0.6) rectangle (2.2* 0.6 + 0.4, -2.3* 0.6);
        \draw[rounded corners, fill=ibmcyan, ibmcyan, path fading=north] (2.2* 0.6 + 1.4, 0.3* 0.6) rectangle (2.2* 0.6 + 1.8, -1.3* 0.6);
        \draw[rounded corners, fill=ibmcyan, ibmcyan, path fading=south] (6-1, -0.7* 0.6) rectangle (6.4-1, -2.3* 0.6);
        \draw[ibmpurple, thick] (2.2* 0.6 + 0.4, -1 * 0.6) -- (2.2* 0.6 + 1.4, -1 * 0.6);
        
        \draw (4.2, -1 * 0.6) node[ibmcyan] {$\cdots$};

        \draw[very thick, ibmpurple] (1.2  * 0.6, -0.3 * 0.6) -- (6.9, -0.3 * 0.6);
        \draw[thick, ibmcyan] (6-1-0.3, -1 * 0.6) -- (6-1, -1 * 0.6);
        
        \draw (2.2* 0.6 + 0.9, -0.3 * 0.6) node[circle, fill=ibmpurple, ibmpurple, inner sep=1.5pt] (ctrl1){};
        \draw (6.6-1+0.3, -0.3 * 0.6) node[circle, fill=ibmpurple, ibmpurple, inner sep=1.5pt] (ctrl2){};
        \draw[very thick, ibmpurple] (ctrl1) -- (2.2* 0.6 + 0.9, -1 * 0.6);
        \draw[very thick, ibmpurple] (ctrl2) -- (6.6-1+0.3, -1 * 0.6);
        \draw[fill=ibmpurple, ibmpurple] (1.2, -0) rectangle node[white]{$H$}(2.2-0.4, -0.6 * 0.6); 
        \draw[fill=ibmpurple, ibmpurple] (6.9, -0) rectangle node[white]{$H$}(7.5, -0.6 * 0.6); 
        \draw[ibmpurple, very thick] (7.7, -0) rectangle (8.3, -0.6 * 0.6); 
        \draw[ibmpurple, very thick] (7.83, -0.6* 0.6) arc (180:0:0.3* 0.6);
        \draw[ibmpurple, very thick] (8, -0.6* 0.6) -- (8.1, 0 - 0.2* 0.6);
        \draw[ibmpurple, very thick] (7.5, -0.3 * 0.6) -- (7.7, -0.3 * 0.6);

        \draw (0.8 * 0.6, -0.3 * 0.6) node[ibmpurple] {$\ket{0}$};
        \draw (8.7, -0.3 * 0.6) node[ibmpurple] {$\ket{0}$};
          
        \draw[rounded corners, fill=ibmpurple, ibmpurple] (2.2* 0.6 + 0.6, -0.8* 0.6+0.1) rectangle node[white]{$Y$} (2.2* 0.6 + 1.2, -1.2* 0.6-0.1);
        \draw[thick, ibmblue] (2.2* 0.6 + 2, -1* 0.6) -- (2.2* 0.6 + 1.8, -1* 0.6);
        \draw[rounded corners, fill=ibmpurple, ibmpurple]  (6.6-1, -0.8* 0.6+0.1) rectangle node[white]{$X$}(7.2-1, -1.2* 0.6-0.1); 
        \draw[thick, ibmpurple] (6.6-1, -1* 0.6) -- (6.4-1, -1* 0.6);
        \draw[thick, ibmpurple] (7.4-1, -1* 0.6) -- (7.2-1, -1* 0.6);
     }
    
    




    \draw (0, 0) node{\scalebox{0.75}{\usebox{\stageone}}};
    \draw (6, 0) node{\scalebox{0.75}{\usebox{\stagetwo}}};
    \draw (7, -3) node{\scalebox{0.6}{\usebox{\stagethreehalf}}};
    \draw (-1, -3) node{\scalebox{0.75}{\usebox{\stagefour}}};

    \draw[->, ultra thick, ibmblue] (4.85, 0.45) arc (130:50:1.5);
    \draw (6, 1) node[ibmblue]{\bf identify accessible wires} ;

    \draw[->, ultra thick, ibmblue] (11.5, -1) arc (90:0:1);
    \draw (14, -1) node[ibmblue]{\bf find and score};
    \draw (14, -1.3) node[ibmblue]{\bf checks};

    \draw [->, ultra thick, ibmblue] (6.2, -4.2) arc (-50:-130:1.5);
    \draw (6-0.2, -5) node[ibmblue]{\bf implement the best check};
    
    \draw (0, -5.2) node {};
   \end{tikzpicture}

%% file: figures/spacetime.tex
\begin{tikzpicture}



    \newsavebox{\czgate}
    \savebox{\czgate}{
        \draw (0, 0) node[ibmcyan, fill=ibmcyan, circle, inner sep=2pt] {};
        \draw (0, -1) node[ibmcyan, fill=ibmcyan, circle, inner sep=2pt] {};

        \draw[very thick, ibmcyan] (0, 0) -- (0, -1);
    }
    \newsavebox{\hgate}
    \savebox{\hgate}{
        \draw[ibmcyan, fill=ibmcyan, rounded corners] (-0.3, 0.3) rectangle node[white]{$H$} (0.3, -0.3);
    }

    \newsavebox{\sqrtzgate}
    \savebox{\sqrtzgate}{
        \draw[ibmcyan, fill=ibmcyan, rounded corners] (-0.3, 0.3) rectangle node[white]{$S$} (0.3, -0.3);
    }

    \newsavebox{\wirewithdot}
    \savebox{\wirewithdot}{
        \draw (0.5, 0) node[ibmpurple, fill=ibmpurple, rectangle, inner sep=2pt] {};
        \draw[ibmpurple, very thick] (0, 0) -- (1, 0);
    }
    \newsavebox{\wirewithdotacz}
    \savebox{\wirewithdotacz}{
        \draw (0.4, 0) node[ibmpurple, fill=ibmpurple, rectangle, inner sep=2pt] {};
        \draw[ibmpurple, very thick] (0, 0) -- (1, 0);
    }
    \newsavebox{\wirewithdotbcz}
    \savebox{\wirewithdotbcz}{
        \draw (0.6, 0) node[ibmpurple, fill=ibmpurple, rectangle, inner sep=2pt] {};
        \draw[ibmpurple, very thick] (0, 0) -- (1, 0);
    }
    \newsavebox{\longwire}
    \savebox{\longwire}{
        \draw (1, 0) node[ibmpurple, fill=ibmpurple, rectangle, inner sep=2pt] {};
        \draw[ibmpurple, very thick] (0, 0) -- (2, 0);
    }
    \newsavebox{\inputgate}
    \savebox{\inputgate}{
        \draw (0, 0) node[ibmred, fill=ibmred, circle, inner sep=3.5pt] {\textcolor{white}{\tiny in}};
    }
    \newsavebox{\outputgate}
    \savebox{\outputgate}{
        \draw (0, 0) node[ibmred, fill=ibmred, circle, inner sep=2pt] {\textcolor{white}{\tiny out}};
    }

    \newsavebox{\fullcircuit}
    \savebox{\fullcircuit}{
        \draw (-1, 0) node{\usebox{\wirewithdot}};
        \draw (0, 0) node{\usebox{\wirewithdotbcz}};
        \draw (1, 0) node{\usebox{\longwire}};
        \draw (3, 0) node{\usebox{\longwire}};
    
        \draw (-1, -1) node{\usebox{\longwire}};
        \draw (1, -1) node{\usebox{\wirewithdotacz}};
        \draw (2, -1) node{\usebox{\wirewithdotbcz}};
        \draw (3, -1) node{\usebox{\wirewithdotacz}};
        \draw (4, -1) node{\usebox{\wirewithdot}};
    
        \draw (-1, -2) node{\usebox{\wirewithdot}};
        \draw (0, -2) node{\usebox{\wirewithdotbcz}};
        \draw (1, -2) node{\usebox{\longwire}};
        \draw (3, -2) node{\usebox{\wirewithdot}};
        \draw (4, -2) node{\usebox{\wirewithdotacz}};
        
        \draw (-1, -3) node{\usebox{\longwire}};
        \draw (1, -3) node{\usebox{\wirewithdotacz}};
        \draw (2, -3) node{\usebox{\longwire}};
        \draw (4, -3) node{\usebox{\wirewithdotacz}};
    
        \draw (-1, 0) node{\usebox{\inputgate}};
        \draw (-1, -1) node{\usebox{\inputgate}};
        \draw (-1, -2) node{\usebox{\inputgate}};
        \draw (-1, -3) node{\usebox{\inputgate}};
        \draw (0, 0) node{\usebox{\hgate}};
        \draw (0, -2) node{\usebox{\hgate}};
        \draw (1, 0) node{\usebox{\czgate}};
        \draw (1, -2) node{\usebox{\czgate}};
        \draw (2, -1) node{\usebox{\sqrtzgate}};
        \draw (2, -3) node{\usebox{\hgate}};
        \draw (3, -1) node{\usebox{\czgate}};
        \draw (4, -2) node{\usebox{\czgate}};
        \draw (3, 0) node{\usebox{\sqrtzgate}};
        \draw (4, -1) node{\usebox{\hgate}};
    
        \draw (5, 0) node{\usebox{\outputgate}};
        \draw (5, -1) node{\usebox{\outputgate}};
        \draw (5, -2) node{\usebox{\outputgate}};
        \draw (5, -3) node{\usebox{\outputgate}};
    }

    \newsavebox{\fulldag}
    \savebox{\fulldag}{
        \draw (0, 0) node[circle, draw, ibmred, fill=ibmred!90] (in1) {};
        \draw (0, -1) node[circle, draw, ibmred, fill=ibmred!90] (in2) {};
        \draw (0, -2) node[circle, draw, ibmred, fill=ibmred!90] (in3) {};
        \draw (0, -3) node[circle, draw, ibmred, fill=ibmred!90] (in4) {};

        \draw (1, 0) node[circle, draw, ibmcyan, fill=ibmcyan!90] (g1) {};
        \draw (1, -2) node[circle, draw, ibmcyan, fill=ibmcyan!90] (g2) {};

        \draw (2, -0.5) node[circle, draw, ibmcyan, fill=ibmcyan!90] (g3) {};
        \draw (2, -2.5) node[circle, draw, ibmcyan, fill=ibmcyan!90] (g4) {};

        \draw (3, -1) node[circle, draw, ibmcyan, fill=ibmcyan!90] (g5) {};
        \draw (3, -3) node[circle, draw, ibmcyan, fill=ibmcyan!90] (g6) {};

        \draw (4, 0) node[circle, draw, ibmcyan, fill=ibmcyan!90] (g7) {};
        \draw (4, -1.5) node[circle, draw, ibmcyan, fill=ibmcyan!90] (g8) {};

        \draw (5, -1) node[circle, draw, ibmcyan, fill=ibmcyan!90] (g9) {};
        \draw (5, -2.5) node[circle, draw, ibmcyan, fill=ibmcyan!90] (g10) {};

        \draw (6, 0) node[circle, draw, ibmred, fill=ibmred!90] (out1) {};
        \draw (6, -1) node[circle, draw, ibmred, fill=ibmred!90] (out2) {};
        \draw (6, -2) node[circle, draw, ibmred, fill=ibmred!90] (out3) {};
        \draw (6, -3) node[circle, draw, ibmred, fill=ibmred!90] (out4) {};

        \draw[ibmpurple, ->, thick] (in1) -- (g1);
        \draw[ibmpurple, ->, thick] (in3) -- (g2);

        \draw[ibmpurple, ->, thick] (g1) -- (g3);
        \draw[ibmpurple, ->, thick] (g2) -- (g4);
        \draw[ibmpurple, ->, thick] (in2) -- (g3);
        \draw[ibmpurple, ->, thick] (in4) -- (g4);

        \draw[ibmpurple, ->, thick] (g3) -- (g5);
        \draw[ibmpurple, ->, thick] (g3) -- (g7);
        \draw[ibmpurple, ->, thick] (g4) -- (g6);
        \draw[ibmpurple, ->, thick] (g4) -- (g8);

        \draw[ibmpurple, ->, thick] (g5) -- (g8);

        \draw[ibmpurple, ->, thick] (g8) -- (g9);
        \draw[ibmpurple, ->, thick] (g8) -- (g10);
        \draw[ibmpurple, ->, thick] (g10) -- (out3);
        \draw[ibmpurple, ->, thick] (g10) -- (out4);
        \draw[ibmpurple, ->, thick] (g9) -- (out2);
        \draw[ibmpurple, ->, thick] (g6) -- (g10);

        \draw[ibmpurple, ->, thick] (g7) -- (out1);

    }
    \newsavebox{\fulllinedag}
    \savebox{\fulllinedag}{
        \draw (0, 0) node[circle] (in1) {};
        \draw (0, -1) node[circle] (in2) {};
        \draw (0, -2) node[circle] (in3) {};
        \draw (0, -3) node[circle] (in4) {};

        \draw (1, 0) node[circle] (g1) {};
        \draw (1, -2) node[circle] (g2) {};

        \draw (2, -0.5) node[circle] (g3) {};
        \draw (2, -2.5) node[circle] (g4) {};

        \draw (3, -1) node[circle] (g5) {};
        \draw (3, -3) node[circle] (g6) {};

        \draw (4, 0) node[circle] (g7) {};
        \draw (4, -1.5) node[circle] (g8) {};

        \draw (5, -1) node[circle] (g9) {};
        \draw (5, -2.5) node[circle] (g10) {};

        \draw (6, 0) node[circle] (out1) {};
        \draw (6, -1) node[circle] (out2) {};
        \draw (6, -2) node[circle] (out3) {};
        \draw (6, -3) node[circle] (out4) {};

        \path (in1) -- node[midway, ibmpurple, fill=ibmpurple, rectangle, draw](e1){} (g1);
        \path (in3) -- node[midway, ibmpurple, fill=ibmpurple, rectangle, draw](e2){} (g2);
        \path (g1) -- node[midway, ibmpurple, fill=ibmpurple, rectangle, draw](e3){} (g3);
        \path (g2) -- node[midway, ibmpurple, fill=ibmpurple, rectangle, draw](e4){}(g4);
        \path (in2) -- node[midway, ibmpurple, fill=ibmpurple, rectangle, draw](e5){} (g3);
        \path (in4) -- node[midway, ibmpurple, fill=ibmpurple, rectangle, draw](e6){} (g4);
        \path (g3) -- node[midway, ibmpurple, fill=ibmpurple, rectangle, draw](e7){} (g5);
        \path (g3) -- node[midway, ibmpurple, fill=ibmpurple, rectangle, draw](e8){} (g7);
        \path (g4) -- node[midway, ibmpurple, fill=ibmpurple, rectangle, draw](e9){}(g6);
        \path (g4) -- node[midway, ibmpurple, fill=ibmpurple, rectangle, draw](e10){}(g8);
        \path (g5) -- node[midway, ibmpurple, fill=ibmpurple, rectangle, draw](e11){} (g8);
        \path (g8) -- node[midway, ibmpurple, fill=ibmpurple, rectangle, draw](e12){}  (g9);
        \path (g8) --  node[midway, ibmpurple, fill=ibmpurple, rectangle, draw](e13){}  (g10);
        \path (g10) -- node[midway, ibmpurple, fill=ibmpurple, rectangle, draw](e14){}(out3);
        \path (g10) -- node[midway, ibmpurple, fill=ibmpurple, rectangle, draw](e15){} (out4);
        \path (g9) -- node[midway, ibmpurple, fill=ibmpurple, rectangle, draw](e16){} (out2);
        \path (g6) -- node[midway, ibmpurple, fill=ibmpurple, rectangle, draw](e17){}  (g10);
        \path (g7) -- node[midway, ibmpurple, fill=ibmpurple, rectangle, draw](e18){}  (out1);

        \draw[thick, ->] (e1) -- (e3);
        \draw[thick, ->] (e2) -- (e4);
        \draw[thick, ->] (e3) -- (e7);\draw[thick, ->] (e3) -- (e8);
        \draw[thick, ->] (e4) -- (e9);\draw[thick, ->] (e4) -- (e10);
        \draw[thick, ->] (e5) -- (e7);\draw[thick, ->] (e5) -- (e8);
        \draw[thick, ->] (e6) -- (e9);\draw[thick, ->] (e6) -- (e10);
        \draw[thick, ->] (e7) -- (e11);
        \draw[thick, ->] (e8) -- (e18);
        \draw[thick, ->] (e9) -- (e17);
        \draw[thick, ->] (e10) -- (e12);\draw[thick, ->] (e10) -- (e13);
        \draw[thick, ->] (e11) -- (e12);\draw[thick, ->] (e11) -- (e13);
        \draw[thick, ->] (e12) -- (e16);
        \draw[thick, ->] (e13) -- (e14);\draw [thick, ->](e13) -- (e15);
        \draw[thick, ->] (e17) -- (e14);\draw [thick, ->](e17) -- (e15);

    }

    \draw (0,0) node{a)};
    \draw (2.7, -1.5)  node {$ \Qcircuit @C=1.0em @R=1.2em @!R { 
        &\gate{H} & \ctrl{1} & \qw       & \gate{S}  & \qw &\qw\\
        &\qw      & \ctrl{0} & \gate{S}  &  \ctrl{1} & \gate{H}&\qw\\
        &\gate{H} &\ctrl{1}  & \qw       &  \ctrl{0} & \ctrl{1}&\qw\\
        &\qw      &\ctrl{0}  & \gate{H}  &  \qw      & \ctrl{0}&\qw
     }$
    };
    
    \draw (7-0.8,0) node{b)};
    \draw(8, 0) node{\usebox{\fullcircuit}};
    \draw (14,0) node{c)};
    \draw(14.5,0) node{\usebox{\fulldag}};
    \draw (21.5,0) node{d)};
    \draw(21.5,0) node{\usebox{\fulllinedag}};
    \end{tikzpicture}

%% file: figures/backpropagator.tex
\begin{tikzpicture}
    \draw[rounded corners, fill=black!10] (0, 0) rectangle (10, 10);

    \draw[dashed, fill=gray!10] plot [smooth cycle, tension=0.] 
        coordinates 
        {
            (1, 10)(3.5, 6.7)(4.6, 6.7)(4.6, 4.9)(3.7, 4.9)(3.7, 4)(2.6, 4)(1.5, 0)(0, 0)(0, 10)
        };

    \draw[rounded corners, fill=white] (4, 5) rectangle (4.5, 6.5);
    \draw[rounded corners, fill=white] (4, 3.3) rectangle (4.5, 4.8);
    \draw[rounded corners, fill=white] (5.2, 1.65 + 2.5) rectangle (5.7, 2.4 + 3.25);
    \draw[rounded corners, fill=white, ] (4-1.2,  1.65 + 2.5) rectangle (4.5-1.2, 2.4 + 3.25);
x
    \draw[red] (4.5, 5.3) edge node[above]{$w$} (5.2, 5.3);
    \draw (4.5, 4.8 - 0.3) edge (5.2, 4.8 - 0.3);
    \draw (4.5, 3.3 + 0.3) edge (5.2, 3.3 + 0.3);
    \draw (4.5, 6.5 - 0.3) edge (5.2, 6.5 - 0.3);

    \draw (5.7, 5.3) edge  (6.2, 5.3);
    \draw (5.7, 4.8 - 0.3) edge (6.2, 4.8 - 0.3);
    \draw (3.3, 5.3) edge  (4, 5.3);
    \draw (3.3, 4.8 - 0.3) edge (4, 4.8 - 0.3);
    \draw (3.3, 3.3 + 0.3) edge  (4, 3.3 + 0.3);
    \draw (3.3, 6.5 - 0.3) edge (4, 6.5 - 0.3);

    \draw (1.5, 6) node {\huge{$A_w$}};
    \draw (7.5, 6) node {\huge{$B_w$}};

\end{tikzpicture}

%% file: figures/half_swaps.tex
\begin{tikzpicture}
    \newsavebox{\zbox}
    \newsavebox{\xbox}
    \savebox{\zbox}{
        \draw[ibmgreen, very thick] (0, 0.4) -- (0, 0);
        \draw[] (0, 0.4) node[circle, fill=ibmgreen, inner sep=1pt]{\textcolor{white}{$Z$}};
    }
    \savebox{\xbox}{
        \draw[ibmgreen, very thick] (0, 0.4) -- (0, 0);
        \draw[] (0, 0.4) node[rectangle, fill=ibmgreen, inner sep=2pt]{\textcolor{white}{$X$}};
    }
    \newsavebox{\swapgate}
    \savebox{\swapgate}{
        \draw[very thick, ibmpurple] (45:0.2) -- (45+180:0.2);
        \draw[very thick, ibmpurple] (-45:0.2) -- (-45+180:0.2);

        \draw[very thick, ibmpurple, yshift=-0.9cm] (45:0.2) -- (45+180:0.2);
        \draw[very thick, ibmpurple, yshift=-0.9cm] (-45:0.2) -- (-45+180:0.2);
        \draw[very thick, ibmpurple] (0, 0) -- (0, -0.9);
    }
    \newsavebox{\halfswapgate}
    \savebox{\halfswapgate}{
        \draw[very thick, ibmpurple] (45:0.2) -- (45+180:0.2);
        \draw[very thick, ibmpurple] (-45:0.2) -- (-45+180:0.2);
        \draw[very thick, ibmpurple] (-0.14, 0.14) rectangle (0.14, -0.14);

        \draw[very thick, ibmpurple, yshift=-0.9cm] (45:0.2) -- (45+180:0.2);
        \draw[very thick, ibmpurple, yshift=-0.9cm] (-45:0.2) -- (-45+180:0.2);
        \draw[very thick, ibmpurple, yshift=-0.9cm] (-0.14, 0.14) rectangle (0.14, -0.14);

        \draw[very thick, ibmpurple] (0, 0) -- (0, -0.9);
    }
  
    \newsavebox{\qbits}
    \savebox{\qbits}{
        \draw[ibmpurple, fill=ibmpurple, very thick] (0, 1.3) node[fill=ibmpurple, circle, inner sep=4pt] (a){};
        \draw[fill=ibmpurple, very thick] (0, 0.4) node[fill=ibmpurple, circle, inner sep=4pt] (b){};
        \draw[fill=ibmcyan, very thick] (0, -0.5) node[fill=ibmcyan, circle, inner sep=4pt](c) {};
        \draw[very  thick, ibmpurple] (a) --(b)--(c);
        \draw[very thick, ibmcyan, path fading=south, dashed] (c) -- (0, -1.5);
    }
    \newsavebox{\swapless}
    \savebox{\swapless}{
    \draw (-2, 0) node{\usebox{\qbits}};
        \draw[very thick, ibmpurple] (-1.2, 0.4) -- (9.2, 0.4);
        \draw[very thick, ibmpurple] (-1.2, 1.3) -- (9.2, 1.3);
        \draw[rounded corners, fill=ibmcyan, ibmcyan, path fading=south]  (0, 0) rectangle (8, -1.6);
        \draw[very thick, ibmcyan] (-0.2, -0.5) -- (8.2, -0.5);

        \draw (1.6, 0.4) node[ibmpurple, fill=ibmpurple, circle, inner sep=2pt] {};
        \draw[ibmpurple, very thick] (1.6, 0.4) -- (1.6, -0.5);
        \draw (1.6 * 3, 0.4) node[ibmpurple, fill=ibmpurple, circle, inner sep=2pt] {};
        \draw[ibmpurple, very thick] (1.6 * 3, 0.4) -- (1.6 * 3, -0.5);
        \draw (1.6 * 2, 1.3) node[ibmpurple, fill=ibmpurple, circle, inner sep=2pt] {};
        \draw[ibmpurple, very thick] (1.6 * 2, 1.3) -- (1.6 * 2, -0.5);
        \draw (1.6 * 4, 1.3) node[ibmpurple, fill=ibmpurple, circle, inner sep=2pt] {};
        \draw[ibmpurple, very thick] (1.6 * 4, 1.3) -- (1.6 * 4, -0.5);

        \draw[fill=ibmpurple, rounded corners, ibmpurple] (1.6 - 0.3, -0.5 + 0.3) rectangle node[white]{$P_1$} (1.6 + 0.3, -0.5 - 0.3);
        \draw[fill=ibmpurple, rounded corners, ibmpurple] (1.6 * 3- 0.3, -0.5 + 0.3) rectangle node[white]{$P_2$} (1.6 * 3 + 0.3, -0.5 - 0.3);
        \draw[fill=ibmpurple, rounded corners, ibmpurple] (1.6 * 2 - 0.3, -0.5 + 0.3) rectangle node[white]{$Q_1$} (1.6*2 + 0.3, -0.5 - 0.3);
        \draw[fill=ibmpurple, rounded corners, ibmpurple] (1.6 * 4- 0.3, -0.5 + 0.3) rectangle node[white]{$Q_2$} (1.6 * 4 + 0.3, -0.5 - 0.3);
        \draw[fill=ibmpurple, ibmpurple] (-0.2 - 0.3, 0.4 + 0.3) rectangle node[white]{$H$} (-0.2 + 0.3, 0.4 - 0.3);
        \draw[fill=ibmpurple, ibmpurple] (-0.2- 0.3, 1.3 + 0.3) rectangle node[white]{$H$} (-0.2 + 0.3, 1.3 - 0.3);
        \draw[fill=ibmpurple, ibmpurple] (8.2 - 0.3, 0.4 + 0.3) rectangle node[white]{$H$} (8.2 + 0.3, 0.4 - 0.3);
        \draw[fill=ibmpurple, ibmpurple] (8.2 - 0.3, 1.3 + 0.3) rectangle node[white]{$H$} (8.2 + 0.3, 1.3 - 0.3);

       \draw (8.85, 1.3) node{\usebox{\zbox}};
       \draw (7.3, 1.3) node{\usebox{\xbox}};
       \draw (3 * 1.6, 1.3) node{\usebox{\xbox}};
       \draw (1 * 1.6, 1.3) node{\usebox{\xbox}};
       \draw (-0.85, 1.3) node{\usebox{\zbox}};
    }
    \newsavebox{\withswap}
    \savebox{\withswap}{
    \draw (-2, 0) node{\usebox{\qbits}};
        \draw[very thick, ibmpurple] (-1.2, 0.4) -- (9.2, 0.4);
        \draw[very thick, ibmpurple] (-1.2, 1.3) -- (9.2, 1.3);
        \draw[rounded corners, fill=ibmcyan, ibmcyan, path fading=south]  (0, 0) rectangle (8, -1.6);
        \draw[very thick, ibmcyan] (-0.2, -0.5) -- (8.2, -0.5);

        \draw (1.6, 0.4) node[ibmpurple, fill=ibmpurple, circle, inner sep=2pt] {};
        \draw[ibmpurple, very thick] (1.6, 0.4) -- (1.6, -0.5);
        \draw (1.6 * 3, 0.4) node[ibmpurple, fill=ibmpurple, circle, inner sep=2pt] {};
        \draw[ibmpurple, very thick] (1.6 * 3, 0.4) -- (1.6 * 3, -0.5);
        \draw (1.6 * 2, 0.4) node[ibmpurple, fill=ibmpurple, circle, inner sep=2pt] {};
        \draw[ibmpurple, very thick] (1.6 * 2, 0.4) -- (1.6 * 2, -0.5);
        \draw (1.6 * 4, 0.4) node[ibmpurple, fill=ibmpurple, circle, inner sep=2pt] {};
        \draw[ibmpurple, very thick] (1.6 * 4, 0.4) -- (1.6 * 4, -0.5);

        \draw[fill=ibmpurple, rounded corners, ibmpurple] (1.6 - 0.3, -0.5 + 0.3) rectangle node[white]{$P_1$} (1.6 + 0.3, -0.5 - 0.3);
        \draw[fill=ibmpurple, rounded corners, ibmpurple] (1.6 * 3- 0.3, -0.5 + 0.3) rectangle node[white]{$P_2$} (1.6 * 3 + 0.3, -0.5 - 0.3);
        \draw[fill=ibmpurple, rounded corners, ibmpurple] (1.6 * 2 - 0.3, -0.5 + 0.3) rectangle node[white]{$Q_1$} (1.6*2 + 0.3, -0.5 - 0.3);
        \draw[fill=ibmpurple, rounded corners, ibmpurple] (1.6 * 4- 0.3, -0.5 + 0.3) rectangle node[white]{$Q_2$} (1.6 * 4 + 0.3, -0.5 - 0.3);
        
        \draw (1.6 * 1.5, 1.3) node{\usebox{\swapgate}};
        \draw (1.6 * 2.5, 1.3) node{\usebox{\swapgate}};
        \draw (1.6 * 3.5, 1.3) node{\usebox{\swapgate}};
        \draw[fill=ibmpurple, ibmpurple] (-0.2 - 0.3, 0.4 + 0.3) rectangle node[white]{$H$} (-0.2 + 0.3, 0.4 - 0.3);
        \draw[fill=ibmpurple, ibmpurple] (-0.2- 0.3, 1.3 + 0.3) rectangle node[white]{$H$} (-0.2 + 0.3, 1.3 - 0.3);
        \draw[fill=ibmpurple, ibmpurple] (8.2 - 0.3, 0.4 + 0.3) rectangle node[white]{$H$} (8.2 + 0.3, 0.4 - 0.3);
        \draw[fill=ibmpurple, ibmpurple] (8.2 - 0.3, 1.3 + 0.3) rectangle node[white]{$H$} (8.2 + 0.3, 1.3 - 0.3);

          \draw (8.85, 0.4) node{\usebox{\zbox}};
          \draw (7.3, 0.4) node{\usebox{\xbox}};
          \draw (3.75 * 1.6, 0.4) node{\usebox{\xbox}};
          \draw (3 * 1.6, 1.3) node{\usebox{\xbox}};
          \draw (2.25* 1.6, 0.4) node{\usebox{\xbox}};
          \draw (1.75* 1.6, 0.4) node{\usebox{\xbox}};
          \draw (1.5* 0.8 - 0.1 , 1.3) node{\usebox{\xbox}};
          \draw (-0.85, 1.3) node{\usebox{\zbox}};
       
    }
    \newsavebox{\withhalfswap}
    \savebox{\withhalfswap}{
    \draw (-2, 0) node{\usebox{\qbits}};
        \draw[very thick, ibmpurple] (-1.2, 0.4) -- (9.2, 0.4);
        \draw[very thick, ibmpurple] (-1.2, 1.3) -- (9.2, 1.3);
        \draw[rounded corners, fill=ibmcyan, ibmcyan, path fading=south]  (0, 0) rectangle (8, -1.6);
        \draw[very thick, ibmcyan] (-0.2, -0.5) -- (8.2, -0.5);

        \draw (1.6, 0.4) node[ibmpurple, fill=ibmpurple, circle, inner sep=2pt] {};
        \draw[ibmpurple, very thick] (1.6, 0.4) -- (1.6, -0.5);
        \draw (1.6 * 3, 0.4) node[ibmpurple, fill=ibmpurple, circle, inner sep=2pt] {};
        \draw[ibmpurple, very thick] (1.6 * 3, 0.4) -- (1.6 * 3, -0.5);
        \draw (1.6 * 2, 0.4) node[ibmpurple, fill=ibmpurple, circle, inner sep=2pt] {};
        \draw[ibmpurple, very thick] (1.6 * 2, 0.4) -- (1.6 * 2, -0.5);
        \draw (1.6 * 4, 0.4) node[ibmpurple, fill=ibmpurple, circle, inner sep=2pt] {};
        \draw[ibmpurple, very thick] (1.6 * 4, 0.4) -- (1.6 * 4, -0.5);

        \draw[fill=ibmpurple, rounded corners, ibmpurple] (1.6 - 0.3, -0.5 + 0.3) rectangle node[white]{$P_1$} (1.6 + 0.3, -0.5 - 0.3);
        \draw[fill=ibmpurple, rounded corners, ibmpurple] (1.6 * 3- 0.3, -0.5 + 0.3) rectangle node[white]{$P_2$} (1.6 * 3 + 0.3, -0.5 - 0.3);
        \draw[fill=ibmpurple, rounded corners, ibmpurple] (1.6 * 2 - 0.3, -0.5 + 0.3) rectangle node[white]{$Q_1$} (1.6*2 + 0.3, -0.5 - 0.3);
        \draw[fill=ibmpurple, rounded corners, ibmpurple] (1.6 * 4- 0.3, -0.5 + 0.3) rectangle node[white]{$Q_2$} (1.6 * 4 + 0.3, -0.5 - 0.3);
        
        \draw (1.6 * 1.5, 1.3) node{\usebox{\halfswapgate}};
        \draw (1.6 * 2.5, 1.3) node{\usebox{\halfswapgate}};
        \draw (1.6 * 3.5, 1.3) node{\usebox{\halfswapgate}};
        \draw (0.5 * 1.6, 1.3) node{\usebox{\czgate}};

        \draw[fill=ibmpurple, ibmpurple] (-0.2 - 0.3, 0.4 + 0.3) rectangle node[white]{$H$} (-0.2 + 0.3, 0.4 - 0.3);
        \draw[fill=ibmpurple, ibmpurple] (-0.2- 0.3, 1.3 + 0.3) rectangle node[white]{$H$} (-0.2 + 0.3, 1.3 - 0.3);
        \draw[fill=ibmpurple, ibmpurple] (8.2 - 0.3, 0.4 + 0.3) rectangle node[white]{$H$} (8.2 + 0.3, 0.4 - 0.3);
        \draw[fill=ibmpurple, ibmpurple] (8.2 - 0.3, 1.3 + 0.3) rectangle node[white]{$H$} (8.2 + 0.3, 1.3 - 0.3);

          \draw (8.85, 0.4) node{\usebox{\zbox}};
          \draw (7.3, 0.4) node{\usebox{\xbox}};
          \draw (3.75 * 1.6, 0.4) node{\usebox{\xbox}};
          \draw (3 * 1.6, 1.3) node{\usebox{\xbox}};
          \draw (2.25* 1.6, 0.4) node{\usebox{\xbox}};
          \draw (1.75* 1.6, 0.4) node{\usebox{\xbox}};
          \draw (1.6, 1.3) node{\usebox{\xbox}};
          \draw (3.25 * 1.6, 0.4) node{\usebox{\zbox}};
          \draw (2.75 * 1.6, 0.4) node{\usebox{\zbox}};
          \draw (1.25 * 1.6, 0.4) node{\usebox{\zbox}};
          \draw (0.75 * 1.6, 0.4) node{\usebox{\zbox}};
          \draw (0.25 * 1.6, 1.3) node{\usebox{\xbox}};
          \draw (-0.85, 1.3) node{\usebox{\zbox}};
       
    }
    \newsavebox{\withhalfswaptwo}
    \savebox{\withhalfswaptwo}{
        \draw[very thick, ibmpurple] (-1.2, 0.4) -- (9.2, 0.4);
        \draw[very thick, ibmpurple] (-1.2, 1.3) -- (9.2, 1.3);
        \draw[rounded corners, fill=ibmcyan, ibmcyan, path fading=south]  (0, 0) rectangle (8, -1.6);
        \draw[very thick, ibmcyan] (-0.2, -0.5) -- (8.2, -0.5);

        \draw (1.6, 0.4) node[ibmpurple, fill=ibmpurple, circle, inner sep=2pt] {};
        \draw[ibmpurple, very thick] (1.6, 0.4) -- (1.6, -0.5);
        \draw (1.6 * 3, 0.4) node[ibmpurple, fill=ibmpurple, circle, inner sep=2pt] {};
        \draw[ibmpurple, very thick] (1.6 * 3, 0.4) -- (1.6 * 3, -0.5);
        \draw (1.6 * 2, 0.4) node[ibmpurple, fill=ibmpurple, circle, inner sep=2pt] {};
        \draw[ibmpurple, very thick] (1.6 * 2, 0.4) -- (1.6 * 2, -0.5);
        \draw (1.6 * 4, 0.4) node[ibmpurple, fill=ibmpurple, circle, inner sep=2pt] {};
        \draw[ibmpurple, very thick] (1.6 * 4, 0.4) -- (1.6 * 4, -0.5);

        \draw[fill=ibmpurple, rounded corners, ibmpurple] (1.6 - 0.3, -0.5 + 0.3) rectangle node[white]{$P_1$} (1.6 + 0.3, -0.5 - 0.3);
        \draw[fill=ibmpurple, rounded corners, ibmpurple] (1.6 * 3- 0.3, -0.5 + 0.3) rectangle node[white]{$P_2$} (1.6 * 3 + 0.3, -0.5 - 0.3);
        \draw[fill=ibmpurple, rounded corners, ibmpurple] (1.6 * 2 - 0.3, -0.5 + 0.3) rectangle node[white]{$Q_1$} (1.6*2 + 0.3, -0.5 - 0.3);
        \draw[fill=ibmpurple, rounded corners, ibmpurple] (1.6 * 4- 0.3, -0.5 + 0.3) rectangle node[white]{$Q_2$} (1.6 * 4 + 0.3, -0.5 - 0.3);
        
        \draw (1.6 * 1.5, 1.3) node{\usebox{\halfswapgate}};
        \draw (1.6 * 2.5, 1.3) node{\usebox{\halfswapgate}};
        \draw (1.6 * 3.5, 1.3) node{\usebox{\halfswapgate}};
        \draw (4.5 * 1.6, 1.3) node{\usebox{\czgate}};

        \draw[fill=ibmpurple, ibmpurple] (-0.2 - 0.3, 0.4 + 0.3) rectangle node[white]{$H$} (-0.2 + 0.3, 0.4 - 0.3);
        \draw[fill=ibmpurple, ibmpurple] (-0.2- 0.3, 1.3 + 0.3) rectangle node[white]{$H$} (-0.2 + 0.3, 1.3 - 0.3);
        \draw[fill=ibmpurple, ibmpurple] (8.2 - 0.3, 0.4 + 0.3) rectangle node[white]{$H$} (8.2 + 0.3, 0.4 - 0.3);
        \draw[fill=ibmpurple, ibmpurple] (8.2 - 0.3, 1.3 + 0.3) rectangle node[white]{$H$} (8.2 + 0.3, 1.3 - 0.3);

        \draw (-0.85, 1.3) node{\usebox{\zbox}};
          \draw (8.85, 0.4) node{\usebox{\zbox}};
          \draw (4.75 * 1.6, 0.4) node{\usebox{\xbox}};
          \draw (4.25 * 1.6, 0.4) node{\usebox{\xbox}};
          \draw (3.75 * 1.6, 0.4) node{\usebox{\xbox}};
          \draw (3 * 1.6, 1.3) node{\usebox{\xbox}};
          \draw (2.25* 1.6, 0.4) node{\usebox{\xbox}};
          \draw (1.75* 1.6, 0.4) node{\usebox{\xbox}};
          \draw (1.5* 0.8 - 0.1 , 1.3) node{\usebox{\xbox}};
          \draw (4 * 1.6, 1.3) node{\usebox{\zbox}};
          \draw (2 * 1.6, 1.3) node{\usebox{\zbox}};
       
    }
    \draw(0, 0) node{\usebox{\swapless}};
    \draw(0, -4) node{\usebox{\withswap}};
    \draw(0, -8) node{\usebox{\withhalfswap}};
    \draw(0, -10) node{};
    \end{tikzpicture}

%% file: letter_format.bbl
\begin{thebibliography}{10}
\expandafter\ifx\csname url\endcsname\relax
  \def\url#1{\texttt{#1}}\fi
\expandafter\ifx\csname urlprefix\endcsname\relax\def\urlprefix{URL }\fi
\providecommand{\bibinfo}[2]{#2}
\providecommand{\eprint}[2][]{\url{#2}}

\bibitem{temme2017error}
\bibinfo{author}{Temme, K.}, \bibinfo{author}{Bravyi, S.} \& \bibinfo{author}{Gambetta, J.~M.}
\newblock \bibinfo{title}{Error mitigation for short-depth quantum circuits}.
\newblock \emph{\bibinfo{journal}{Physical Review Letters}} \textbf{\bibinfo{volume}{119}}, \bibinfo{pages}{180509} (\bibinfo{year}{2017}).

\bibitem{cai2023quantum}
\bibinfo{author}{Cai, Z.} \emph{et~al.}
\newblock \bibinfo{title}{Quantum error mitigation}.
\newblock \emph{\bibinfo{journal}{Reviews of Modern Physics}} \textbf{\bibinfo{volume}{95}}, \bibinfo{pages}{045005} (\bibinfo{year}{2023}).

\bibitem{gottesman2013fault}
\bibinfo{author}{Gottesman, D.}
\newblock \bibinfo{title}{Fault-tolerant quantum computation with constant overhead}.
\newblock \emph{\bibinfo{journal}{arXiv preprint arXiv:1310.2984}}  (\bibinfo{year}{2013}).

\bibitem{bravyi2024high}
\bibinfo{author}{Bravyi, S.} \emph{et~al.}
\newblock \bibinfo{title}{High-threshold and low-overhead fault-tolerant quantum memory}.
\newblock \emph{\bibinfo{journal}{Nature}} \textbf{\bibinfo{volume}{627}}, \bibinfo{pages}{778--782} (\bibinfo{year}{2024}).

\bibitem{tsubouchi2023universal}
\bibinfo{author}{Tsubouchi, K.}, \bibinfo{author}{Sagawa, T.} \& \bibinfo{author}{Yoshioka, N.}
\newblock \bibinfo{title}{Universal cost bound of quantum error mitigation based on quantum estimation theory}.
\newblock \emph{\bibinfo{journal}{Physical Review Letters}} \textbf{\bibinfo{volume}{131}}, \bibinfo{pages}{210601} (\bibinfo{year}{2023}).

\bibitem{barron2024provable}
\bibinfo{author}{Barron, S.~V.} \emph{et~al.}
\newblock \bibinfo{title}{Provable bounds for noise-free expectation values computed from noisy samples}.
\newblock \emph{\bibinfo{journal}{Nature Computational Science}} \bibinfo{pages}{1--11} (\bibinfo{year}{2024}).

\bibitem{bacon_sparse_2017}
\bibinfo{author}{Bacon, D.}, \bibinfo{author}{Flammia, S.~T.}, \bibinfo{author}{Harrow, A.~W.} \& \bibinfo{author}{Shi, J.}
\newblock \bibinfo{title}{Sparse quantum codes from quantum circuits}.
\newblock \emph{\bibinfo{journal}{IEEE Transactions on Information Theory}} \textbf{\bibinfo{volume}{63}}, \bibinfo{pages}{2464--2479} (\bibinfo{year}{2017}).

\bibitem{gottesman2022opportunities}
\bibinfo{author}{Gottesman, D.}
\newblock \bibinfo{title}{Opportunities and challenges in fault-tolerant quantum computation}.
\newblock \emph{\bibinfo{journal}{arXiv preprint arXiv:2210.15844}}  (\bibinfo{year}{2022}).

\bibitem{delfosse_spacetime_2023}
\bibinfo{author}{Delfosse, N.} \& \bibinfo{author}{Paetznick, A.}
\newblock \bibinfo{title}{Spacetime codes of {Clifford} circuits}.
\newblock \emph{\bibinfo{journal}{arXiv preprint arXiv:2304.05943}}  (\bibinfo{year}{2023}).

\bibitem{mcewen2023relaxing}
\bibinfo{author}{McEwen, M.}, \bibinfo{author}{Bacon, D.} \& \bibinfo{author}{Gidney, C.}
\newblock \bibinfo{title}{Relaxing hardware requirements for surface code circuits using time-dynamics}.
\newblock \emph{\bibinfo{journal}{Quantum}} \textbf{\bibinfo{volume}{7}}, \bibinfo{pages}{1172} (\bibinfo{year}{2023}).

\bibitem{bombin2024unifying}
\bibinfo{author}{Bombin, H.}, \bibinfo{author}{Litinski, D.}, \bibinfo{author}{Nickerson, N.}, \bibinfo{author}{Pastawski, F.} \& \bibinfo{author}{Roberts, S.}
\newblock \bibinfo{title}{Unifying flavors of fault tolerance with the {ZX} calculus}.
\newblock \emph{\bibinfo{journal}{Quantum}} \textbf{\bibinfo{volume}{8}}, \bibinfo{pages}{1379} (\bibinfo{year}{2024}).

\bibitem{roffe_protecting_2018}
\bibinfo{author}{Roffe, J.}, \bibinfo{author}{Headley, D.}, \bibinfo{author}{Chancellor, N.}, \bibinfo{author}{Horsman, D.} \& \bibinfo{author}{Kendon, V.}
\newblock \bibinfo{title}{Protecting quantum memories using coherent parity check codes}.
\newblock \emph{\bibinfo{journal}{Quantum Science and Technology}} \textbf{\bibinfo{volume}{3}}, \bibinfo{pages}{035010} (\bibinfo{year}{2018}).

\bibitem{debroy_extended_2020}
\bibinfo{author}{Debroy, D.~M.} \& \bibinfo{author}{Brown, K.~R.}
\newblock \bibinfo{title}{Extended flag gadgets for low-overhead circuit verification}.
\newblock \emph{\bibinfo{journal}{Physical Review A}} \textbf{\bibinfo{volume}{102}}, \bibinfo{pages}{052409} (\bibinfo{year}{2020}).

\bibitem{van2023single}
\bibinfo{author}{van~den Berg, E.} \emph{et~al.}
\newblock \bibinfo{title}{Single-shot error mitigation by coherent {Pauli} checks}.
\newblock \emph{\bibinfo{journal}{Physical Review Research}} \textbf{\bibinfo{volume}{5}}, \bibinfo{pages}{033193} (\bibinfo{year}{2023}).

\bibitem{gonzales_quantum_2023}
\bibinfo{author}{Gonzales, A.}, \bibinfo{author}{Shaydulin, R.}, \bibinfo{author}{Saleem, Z.~H.} \& \bibinfo{author}{Suchara, M.}
\newblock \bibinfo{title}{Quantum error mitigation by {Pauli} check sandwiching}.
\newblock \emph{\bibinfo{journal}{Scientific Reports}} \textbf{\bibinfo{volume}{13}}, \bibinfo{pages}{2122} (\bibinfo{year}{2023}).

\bibitem{gottesman_heisenberg_1998}
\bibinfo{author}{Gottesman, D.}
\newblock \bibinfo{title}{The heisenberg representation of quantum computers}.
\newblock \emph{\bibinfo{journal}{arXiv preprint quant-ph/9807006}}  (\bibinfo{year}{1998}).

\bibitem{huang2020predicting}
\bibinfo{author}{Huang, H.-Y.}, \bibinfo{author}{Kueng, R.} \& \bibinfo{author}{Preskill, J.}
\newblock \bibinfo{title}{Predicting many properties of a quantum system from very few measurements}.
\newblock \emph{\bibinfo{journal}{Nature Physics}} \textbf{\bibinfo{volume}{16}}, \bibinfo{pages}{1050--1057} (\bibinfo{year}{2020}).

\bibitem{jozsa2013classical}
\bibinfo{author}{Jozsa, R.} \& \bibinfo{author}{Nest, M. V.~d.}
\newblock \bibinfo{title}{Classical simulation complexity of extended {Clifford} circuits}.
\newblock \emph{\bibinfo{journal}{arXiv preprint arXiv:1305.6190}}  (\bibinfo{year}{2013}).

\bibitem{bouland2017complexity}
\bibinfo{author}{Bouland, A.}, \bibinfo{author}{Fitzsimons, J.~F.} \& \bibinfo{author}{Koh, D.~E.}
\newblock \bibinfo{title}{Complexity classification of conjugated {Clifford} circuits}.
\newblock \emph{\bibinfo{journal}{arXiv preprint arXiv:1709.01805}}  (\bibinfo{year}{2017}).

\bibitem{ghosh2023complexity}
\bibinfo{author}{Ghosh, S.}, \bibinfo{author}{Deshpande, A.}, \bibinfo{author}{Hangleiter, D.}, \bibinfo{author}{Gorshkov, A.~V.} \& \bibinfo{author}{Fefferman, B.}
\newblock \bibinfo{title}{Complexity phase transitions generated by entanglement}.
\newblock \emph{\bibinfo{journal}{Physical Review Letters}} \textbf{\bibinfo{volume}{131}}, \bibinfo{pages}{030601} (\bibinfo{year}{2023}).

\bibitem{gottesman1999demonstrating}
\bibinfo{author}{Gottesman, D.} \& \bibinfo{author}{Chuang, I.~L.}
\newblock \bibinfo{title}{Demonstrating the viability of universal quantum computation using teleportation and single-qubit operations}.
\newblock \emph{\bibinfo{journal}{Nature}} \textbf{\bibinfo{volume}{402}}, \bibinfo{pages}{390--393} (\bibinfo{year}{1999}).

\bibitem{PhysRevLett.97.150504}
\bibinfo{author}{Van~den Nest, M.}, \bibinfo{author}{Miyake, A.}, \bibinfo{author}{D\"ur, W.} \& \bibinfo{author}{Briegel, H.~J.}
\newblock \bibinfo{title}{Universal resources for measurement-based quantum computation}.
\newblock \emph{\bibinfo{journal}{Physical Review Letters}} \textbf{\bibinfo{volume}{97}}, \bibinfo{pages}{150504} (\bibinfo{year}{2006}).
\newblock \urlprefix\url{https://link.aps.org/doi/10.1103/PhysRevLett.97.150504}.

\bibitem{horsman2012surface}
\bibinfo{author}{Horsman, D.}, \bibinfo{author}{Fowler, A.~G.}, \bibinfo{author}{Devitt, S.} \& \bibinfo{author}{Van~Meter, R.}
\newblock \bibinfo{title}{Surface code quantum computing by lattice surgery}.
\newblock \emph{\bibinfo{journal}{New Journal of Physics}} \textbf{\bibinfo{volume}{14}}, \bibinfo{pages}{123011} (\bibinfo{year}{2012}).

\bibitem{hastings2021dynamically}
\bibinfo{author}{Hastings, M.~B.} \& \bibinfo{author}{Haah, J.}
\newblock \bibinfo{title}{Dynamically generated logical qubits}.
\newblock \emph{\bibinfo{journal}{Quantum}} \textbf{\bibinfo{volume}{5}}, \bibinfo{pages}{564} (\bibinfo{year}{2021}).

\bibitem{delfosse_simulation_2023}
\bibinfo{author}{Delfosse, N.} \& \bibinfo{author}{Paetznick, A.}
\newblock \bibinfo{title}{Simulation of noisy {Clifford} circuits without fault propagation}.
\newblock \emph{\bibinfo{journal}{arXiv preprint arXiv:2309.15345}}  (\bibinfo{year}{2023}).

\bibitem{gidney2021stim}
\bibinfo{author}{Gidney, C.}
\newblock \bibinfo{title}{Stim: a fast stabilizer circuit simulator}.
\newblock \emph{\bibinfo{journal}{Quantum}} \textbf{\bibinfo{volume}{5}}, \bibinfo{pages}{497} (\bibinfo{year}{2021}).

\bibitem{maslov2018shorter}
\bibinfo{author}{Maslov, D.} \& \bibinfo{author}{Roetteler, M.}
\newblock \bibinfo{title}{Shorter stabilizer circuits via {Bruhat} decomposition and quantum circuit transformations}.
\newblock \emph{\bibinfo{journal}{IEEE Transactions on Information Theory}} \textbf{\bibinfo{volume}{64}}, \bibinfo{pages}{4729--4738} (\bibinfo{year}{2018}).

\bibitem{Kissinger_2022}
\bibinfo{author}{Kissinger, A.} \& \bibinfo{author}{van~de Wetering, J.}
\newblock \bibinfo{title}{Simulating quantum circuits with {ZX}-calculus reduced stabiliser decompositions}.
\newblock \emph{\bibinfo{journal}{Quantum Science and Technology}} \textbf{\bibinfo{volume}{7}}, \bibinfo{pages}{044001} (\bibinfo{year}{2022}).
\newblock \urlprefix\url{http://dx.doi.org/10.1088/2058-9565/ac5d20}.

\bibitem{PhysRevLett.106.230501}
\bibinfo{author}{Flammia, S.~T.} \& \bibinfo{author}{Liu, Y.-K.}
\newblock \bibinfo{title}{Direct fidelity estimation from few {Pauli} measurements}.
\newblock \emph{\bibinfo{journal}{Physical Review Letters}} \textbf{\bibinfo{volume}{106}}, \bibinfo{pages}{230501} (\bibinfo{year}{2011}).

\bibitem{mckay2023benchmarking}
\bibinfo{author}{McKay, D.~C.} \emph{et~al.}
\newblock \bibinfo{title}{Benchmarking quantum processor performance at scale}.
\newblock \emph{\bibinfo{journal}{arXiv preprint arXiv:2311.05933}}  (\bibinfo{year}{2023}).

\bibitem{harper2019fault}
\bibinfo{author}{Harper, R.} \& \bibinfo{author}{Flammia, S.~T.}
\newblock \bibinfo{title}{Fault-tolerant logical gates in the ibm quantum experience}.
\newblock \emph{\bibinfo{journal}{Physical review letters}} \textbf{\bibinfo{volume}{122}}, \bibinfo{pages}{080504} (\bibinfo{year}{2019}).

\bibitem{gupta2024encoding}
\bibinfo{author}{Gupta, R.~S.} \emph{et~al.}
\newblock \bibinfo{title}{Encoding a magic state with beyond break-even fidelity}.
\newblock \emph{\bibinfo{journal}{Nature}} \textbf{\bibinfo{volume}{625}}, \bibinfo{pages}{259--263} (\bibinfo{year}{2024}).

\bibitem{Bluvstein_2023}
\bibinfo{author}{Bluvstein, D.} \emph{et~al.}
\newblock \bibinfo{title}{Logical quantum processor based on reconfigurable atom arrays}.
\newblock \emph{\bibinfo{journal}{Nature}} \textbf{\bibinfo{volume}{626}}, \bibinfo{pages}{58–65} (\bibinfo{year}{2023}).
\newblock \urlprefix\url{http://dx.doi.org/10.1038/s41586-023-06927-3}.

\bibitem{reichardt2024demonstration}
\bibinfo{author}{Reichardt, B.~W.} \emph{et~al.}
\newblock \bibinfo{title}{Demonstration of quantum computation and error correction with a tesseract code}.
\newblock \emph{\bibinfo{journal}{arXiv preprint arXiv:2409.04628}}  (\bibinfo{year}{2024}).

\bibitem{reichardt2024logical}
\bibinfo{author}{Reichardt, B.~W.} \emph{et~al.}
\newblock \bibinfo{title}{Logical computation demonstrated with a neutral atom quantum processor}.
\newblock \emph{\bibinfo{journal}{arXiv preprint arXiv:2411.11822}}  (\bibinfo{year}{2024}).

\bibitem{self2024protecting}
\bibinfo{author}{Self, C.~N.}, \bibinfo{author}{Benedetti, M.} \& \bibinfo{author}{Amaro, D.}
\newblock \bibinfo{title}{Protecting expressive circuits with a quantum error detection code}.
\newblock \emph{\bibinfo{journal}{Nature Physics}} \textbf{\bibinfo{volume}{20}}, \bibinfo{pages}{219--224} (\bibinfo{year}{2024}).

\bibitem{PhysRevResearch.6.043020}
\bibinfo{author}{D'Arcangelo, M.} \emph{et~al.}
\newblock \bibinfo{title}{Leveraging analog quantum computing with neutral atoms for solvent configuration prediction in drug discovery}.
\newblock \emph{\bibinfo{journal}{Phys. Rev. Res.}} \textbf{\bibinfo{volume}{6}}, \bibinfo{pages}{043020} (\bibinfo{year}{2024}).
\newblock \urlprefix\url{https://link.aps.org/doi/10.1103/PhysRevResearch.6.043020}.

\bibitem{decross2024computational}
\bibinfo{author}{DeCross, M.} \emph{et~al.}
\newblock \bibinfo{title}{The computational power of random quantum circuits in arbitrary geometries}.
\newblock \emph{\bibinfo{journal}{arXiv preprint arXiv:2406.02501}}  (\bibinfo{year}{2024}).

\bibitem{chao2018quantum}
\bibinfo{author}{Chao, R.} \& \bibinfo{author}{Reichardt, B.~W.}
\newblock \bibinfo{title}{Quantum error correction with only two extra qubits}.
\newblock \emph{\bibinfo{journal}{Physical review letters}} \textbf{\bibinfo{volume}{121}}, \bibinfo{pages}{050502} (\bibinfo{year}{2018}).

\bibitem{anker2024flag}
\bibinfo{author}{Anker, B.} \& \bibinfo{author}{Marvian, M.}
\newblock \bibinfo{title}{Flag gadgets based on classical codes}.
\newblock \emph{\bibinfo{journal}{PRX Quantum}} \textbf{\bibinfo{volume}{5}}, \bibinfo{pages}{040340} (\bibinfo{year}{2024}).

\bibitem{li2020projection}
\bibinfo{author}{Li, G.} \emph{et~al.}
\newblock \bibinfo{title}{Projection-based runtime assertions for testing and debugging quantum programs}.
\newblock \emph{\bibinfo{journal}{Proceedings of the ACM on Programming Languages}} \textbf{\bibinfo{volume}{4}}, \bibinfo{pages}{1--29} (\bibinfo{year}{2020}).

\bibitem{tsubouchi2025symmetric}
\bibinfo{author}{Tsubouchi, K.}, \bibinfo{author}{Mitsuhashi, Y.}, \bibinfo{author}{Takagi, R.} \& \bibinfo{author}{Yoshioka, N.}
\newblock \bibinfo{title}{Symmetric channel verification for purifying noisy quantum channels}.
\newblock \emph{\bibinfo{journal}{arXiv preprint arXiv:2503.13114}}  (\bibinfo{year}{2025}).

\bibitem{bonet2018low}
\bibinfo{author}{Bonet-Monroig, X.}, \bibinfo{author}{Sagastizabal, R.}, \bibinfo{author}{Singh, M.} \& \bibinfo{author}{O'Brien, T.}
\newblock \bibinfo{title}{Low-cost error mitigation by symmetry verification}.
\newblock \emph{\bibinfo{journal}{Physical Review A}} \textbf{\bibinfo{volume}{98}}, \bibinfo{pages}{062339} (\bibinfo{year}{2018}).

\end{thebibliography}


\begin{thebibliography}{10}
\expandafter\ifx\csname url\endcsname\relax
  \def\url#1{\texttt{#1}}\fi
\expandafter\ifx\csname urlprefix\endcsname\relax\def\urlprefix{URL }\fi
\providecommand{\bibinfo}[2]{#2}
\providecommand{\eprint}[2][]{\url{#2}}

\bibitem{PhysRevLett.106.230501}
\bibinfo{author}{Flammia, S.~T.} \& \bibinfo{author}{Liu, Y.-K.}
\newblock \bibinfo{title}{Direct fidelity estimation from few {Pauli} measurements}.
\newblock \emph{\bibinfo{journal}{Physical Review Letters}} \textbf{\bibinfo{volume}{106}}, \bibinfo{pages}{230501} (\bibinfo{year}{2011}).

\bibitem{merkel2025clifford}
\bibinfo{author}{Merkel, S.} \emph{et~al.}
\newblock \bibinfo{title}{When {Clifford} benchmarks are sufficient; estimating application performance with scalable proxy circuits}.
\newblock \emph{\bibinfo{journal}{arXiv preprint arXiv:2503.05943}}  (\bibinfo{year}{2025}).

\bibitem{dutkiewicz2024error}
\bibinfo{author}{Dutkiewicz, A.} \emph{et~al.}
\newblock \bibinfo{title}{Error mitigation and circuit division for early fault-tolerant quantum phase estimation}.
\newblock \emph{\bibinfo{journal}{arXiv preprint arXiv:2410.05369}}  (\bibinfo{year}{2024}).

\bibitem{barron2024provable}
\bibinfo{author}{Barron, S.~V.} \emph{et~al.}
\newblock \bibinfo{title}{Provable bounds for noise-free expectation values computed from noisy samples}.
\newblock \emph{\bibinfo{journal}{Nature Computational Science}} \bibinfo{pages}{1--11} (\bibinfo{year}{2024}).

\bibitem{tsubouchi2023universal}
\bibinfo{author}{Tsubouchi, K.}, \bibinfo{author}{Sagawa, T.} \& \bibinfo{author}{Yoshioka, N.}
\newblock \bibinfo{title}{Universal cost bound of quantum error mitigation based on quantum estimation theory}.
\newblock \emph{\bibinfo{journal}{Physical Review Letters}} \textbf{\bibinfo{volume}{131}}, \bibinfo{pages}{210601} (\bibinfo{year}{2023}).

\bibitem{filippov2023scalable}
\bibinfo{author}{Filippov, S.}, \bibinfo{author}{Leahy, M.}, \bibinfo{author}{Rossi, M.~A.} \& \bibinfo{author}{Garc{\'\i}a-P{\'e}rez, G.}
\newblock \bibinfo{title}{Scalable tensor-network error mitigation for near-term quantum computing}.
\newblock \emph{\bibinfo{journal}{arXiv preprint arXiv:2307.11740}}  (\bibinfo{year}{2023}).

\bibitem{temme2017error}
\bibinfo{author}{Temme, K.}, \bibinfo{author}{Bravyi, S.} \& \bibinfo{author}{Gambetta, J.~M.}
\newblock \bibinfo{title}{Error mitigation for short-depth quantum circuits}.
\newblock \emph{\bibinfo{journal}{Physical Review Letters}} \textbf{\bibinfo{volume}{119}}, \bibinfo{pages}{180509} (\bibinfo{year}{2017}).

\bibitem{peng2020simulating}
\bibinfo{author}{Peng, T.}, \bibinfo{author}{Harrow, A.~W.}, \bibinfo{author}{Ozols, M.} \& \bibinfo{author}{Wu, X.}
\newblock \bibinfo{title}{Simulating large quantum circuits on a small quantum computer}.
\newblock \emph{\bibinfo{journal}{Physical Review Letters}} \textbf{\bibinfo{volume}{125}}, \bibinfo{pages}{150504} (\bibinfo{year}{2020}).

\bibitem{bremner2017achieving}
\bibinfo{author}{Bremner, M.~J.}, \bibinfo{author}{Montanaro, A.} \& \bibinfo{author}{Shepherd, D.~J.}
\newblock \bibinfo{title}{Achieving quantum supremacy with sparse and noisy commuting quantum computations}.
\newblock \emph{\bibinfo{journal}{Quantum}} \textbf{\bibinfo{volume}{1}}, \bibinfo{pages}{8} (\bibinfo{year}{2017}).

\bibitem{bouland2018quantum}
\bibinfo{author}{Bouland, A.}, \bibinfo{author}{Fefferman, B.}, \bibinfo{author}{Nirkhe, C.} \& \bibinfo{author}{Vazirani, U.}
\newblock \bibinfo{title}{Quantum supremacy and the complexity of random circuit sampling}.
\newblock \emph{\bibinfo{journal}{arXiv preprint arXiv:1803.04402}}  (\bibinfo{year}{2018}).

\bibitem{bravyi2021classical}
\bibinfo{author}{Bravyi, S.}, \bibinfo{author}{Gosset, D.} \& \bibinfo{author}{Movassagh, R.}
\newblock \bibinfo{title}{Classical algorithms for quantum mean values}.
\newblock \emph{\bibinfo{journal}{Nature Physics}} \textbf{\bibinfo{volume}{17}}, \bibinfo{pages}{337--341} (\bibinfo{year}{2021}).

\bibitem{gottesman2013fault}
\bibinfo{author}{Gottesman, D.}
\newblock \bibinfo{title}{Fault-tolerant quantum computation with constant overhead}.
\newblock \emph{\bibinfo{journal}{arXiv preprint arXiv:1310.2984}}  (\bibinfo{year}{2013}).

\bibitem{bravyi2024high}
\bibinfo{author}{Bravyi, S.} \emph{et~al.}
\newblock \bibinfo{title}{High-threshold and low-overhead fault-tolerant quantum memory}.
\newblock \emph{\bibinfo{journal}{Nature}} \textbf{\bibinfo{volume}{627}}, \bibinfo{pages}{778--782} (\bibinfo{year}{2024}).

\bibitem{bravyi2010tradeoffs}
\bibinfo{author}{Bravyi, S.}, \bibinfo{author}{Poulin, D.} \& \bibinfo{author}{Terhal, B.}
\newblock \bibinfo{title}{Tradeoffs for reliable quantum information storage in 2d systems}.
\newblock \emph{\bibinfo{journal}{Physical Review Letters}} \textbf{\bibinfo{volume}{104}}, \bibinfo{pages}{050503} (\bibinfo{year}{2010}).

\bibitem{gottesman2022opportunities}
\bibinfo{author}{Gottesman, D.}
\newblock \bibinfo{title}{Opportunities and challenges in fault-tolerant quantum computation}.
\newblock \emph{\bibinfo{journal}{arXiv preprint arXiv:2210.15844}}  (\bibinfo{year}{2022}).

\bibitem{bacon_sparse_2017}
\bibinfo{author}{Bacon, D.}, \bibinfo{author}{Flammia, S.~T.}, \bibinfo{author}{Harrow, A.~W.} \& \bibinfo{author}{Shi, J.}
\newblock \bibinfo{title}{Sparse quantum codes from quantum circuits}.
\newblock \emph{\bibinfo{journal}{IEEE Transactions on Information Theory}} \textbf{\bibinfo{volume}{63}}, \bibinfo{pages}{2464--2479} (\bibinfo{year}{2017}).

\bibitem{delfosse_spacetime_2023}
\bibinfo{author}{Delfosse, N.} \& \bibinfo{author}{Paetznick, A.}
\newblock \bibinfo{title}{Spacetime codes of {Clifford} circuits}.
\newblock \emph{\bibinfo{journal}{arXiv preprint arXiv:2304.05943}}  (\bibinfo{year}{2023}).

\bibitem{debroy_extended_2020}
\bibinfo{author}{Debroy, D.~M.} \& \bibinfo{author}{Brown, K.~R.}
\newblock \bibinfo{title}{Extended flag gadgets for low-overhead circuit verification}.
\newblock \emph{\bibinfo{journal}{Physical Review A}} \textbf{\bibinfo{volume}{102}}, \bibinfo{pages}{052409} (\bibinfo{year}{2020}).

\bibitem{gonzales_quantum_2023}
\bibinfo{author}{Gonzales, A.}, \bibinfo{author}{Shaydulin, R.}, \bibinfo{author}{Saleem, Z.~H.} \& \bibinfo{author}{Suchara, M.}
\newblock \bibinfo{title}{Quantum error mitigation by {Pauli} check sandwiching}.
\newblock \emph{\bibinfo{journal}{Scientific Reports}} \textbf{\bibinfo{volume}{13}}, \bibinfo{pages}{2122} (\bibinfo{year}{2023}).

\bibitem{van2023single}
\bibinfo{author}{van~den Berg, E.} \emph{et~al.}
\newblock \bibinfo{title}{Single-shot error mitigation by coherent {Pauli} checks}.
\newblock \emph{\bibinfo{journal}{Physical Review Research}} \textbf{\bibinfo{volume}{5}}, \bibinfo{pages}{033193} (\bibinfo{year}{2023}).

\bibitem{tsubouchi2025symmetric}
\bibinfo{author}{Tsubouchi, K.}, \bibinfo{author}{Mitsuhashi, Y.}, \bibinfo{author}{Takagi, R.} \& \bibinfo{author}{Yoshioka, N.}
\newblock \bibinfo{title}{Symmetric channel verification for purifying noisy quantum channels}.
\newblock \emph{\bibinfo{journal}{arXiv preprint arXiv:2503.13114}}  (\bibinfo{year}{2025}).

\bibitem{mcewen2023relaxing}
\bibinfo{author}{McEwen, M.}, \bibinfo{author}{Bacon, D.} \& \bibinfo{author}{Gidney, C.}
\newblock \bibinfo{title}{Relaxing hardware requirements for surface code circuits using time-dynamics}.
\newblock \emph{\bibinfo{journal}{Quantum}} \textbf{\bibinfo{volume}{7}}, \bibinfo{pages}{1172} (\bibinfo{year}{2023}).

\bibitem{641542}
\bibinfo{author}{Vardy, A.}
\newblock \bibinfo{title}{The intractability of computing the minimum distance of a code}.
\newblock \emph{\bibinfo{journal}{IEEE Transactions on Information Theory}} \textbf{\bibinfo{volume}{43}}, \bibinfo{pages}{1757--1766} (\bibinfo{year}{1997}).

\bibitem{1055873}
\bibinfo{author}{Berlekamp, E.}, \bibinfo{author}{McEliece, R.} \& \bibinfo{author}{van Tilborg, H.}
\newblock \bibinfo{title}{On the inherent intractability of certain coding problems (corresp.)}.
\newblock \emph{\bibinfo{journal}{IEEE Transactions on Information Theory}} \textbf{\bibinfo{volume}{24}}, \bibinfo{pages}{384--386} (\bibinfo{year}{1978}).

\bibitem{10.1007/978-3-030-52482-1_11}
\bibinfo{author}{de~Brugi{\`e}re, T.~G.}, \bibinfo{author}{Baboulin, M.}, \bibinfo{author}{Valiron, B.}, \bibinfo{author}{Martiel, S.} \& \bibinfo{author}{Allouche, C.}
\newblock \bibinfo{title}{Quantum cnot circuits synthesis for nisq architectures using the syndrome decoding problem}.
\newblock In \bibinfo{editor}{Lanese, I.} \& \bibinfo{editor}{Rawski, M.} (eds.) \emph{\bibinfo{booktitle}{Reversible Computation}}, \bibinfo{pages}{189--205} (\bibinfo{publisher}{Springer International Publishing}, \bibinfo{address}{Cham}, \bibinfo{year}{2020}).

\bibitem{Goubault_de_Brugi_re_2025}
\bibinfo{author}{Goubault~de Brugière, T.}, \bibinfo{author}{Martiel, S.} \& \bibinfo{author}{Vuillot, C.}
\newblock \bibinfo{title}{A graph-state based synthesis framework for {Clifford} isometries}.
\newblock \emph{\bibinfo{journal}{Quantum}} \textbf{\bibinfo{volume}{9}}, \bibinfo{pages}{1589} (\bibinfo{year}{2025}).
\newblock \urlprefix\url{http://dx.doi.org/10.22331/q-2025-01-14-1589}.

\bibitem{PhysRevLett.127.200505}
\bibinfo{author}{Piveteau, C.}, \bibinfo{author}{Sutter, D.}, \bibinfo{author}{Bravyi, S.}, \bibinfo{author}{Gambetta, J.~M.} \& \bibinfo{author}{Temme, K.}
\newblock \bibinfo{title}{Error mitigation for universal gates on encoded qubits}.
\newblock \emph{\bibinfo{journal}{Phys. Rev. Lett.}} \textbf{\bibinfo{volume}{127}}, \bibinfo{pages}{200505} (\bibinfo{year}{2021}).
\newblock \urlprefix\url{https://link.aps.org/doi/10.1103/PhysRevLett.127.200505}.

\bibitem{PhysRevLett.97.150504}
\bibinfo{author}{Van~den Nest, M.}, \bibinfo{author}{Miyake, A.}, \bibinfo{author}{D\"ur, W.} \& \bibinfo{author}{Briegel, H.~J.}
\newblock \bibinfo{title}{Universal resources for measurement-based quantum computation}.
\newblock \emph{\bibinfo{journal}{Physical Review Letters}} \textbf{\bibinfo{volume}{97}}, \bibinfo{pages}{150504} (\bibinfo{year}{2006}).
\newblock \urlprefix\url{https://link.aps.org/doi/10.1103/PhysRevLett.97.150504}.

\bibitem{PhysRevA.69.022316}
\bibinfo{author}{Van~den Nest, M.}, \bibinfo{author}{Dehaene, J.} \& \bibinfo{author}{De~Moor, B.}
\newblock \bibinfo{title}{Graphical description of the action of local {Clifford} transformations on graph states}.
\newblock \emph{\bibinfo{journal}{Physical Review A}} \textbf{\bibinfo{volume}{69}}, \bibinfo{pages}{022316} (\bibinfo{year}{2004}).
\newblock \urlprefix\url{https://link.aps.org/doi/10.1103/PhysRevA.69.022316}.

\bibitem{hein2006entanglementgraphstatesapplications}
\bibinfo{author}{Hein, M.} \emph{et~al.}
\newblock \bibinfo{title}{Entanglement in graph states and its applications} (\bibinfo{year}{2006}).
\newblock \urlprefix\url{https://arxiv.org/abs/quant-ph/0602096}.
\newblock \eprint{quant-ph/0602096}.

\bibitem{rank_width_random}
\bibinfo{author}{Lee, C.}, \bibinfo{author}{Lee, J.} \& \bibinfo{author}{Oum, S.-i.}
\newblock \bibinfo{title}{Rank-width of random graphs}.
\newblock \emph{\bibinfo{journal}{Journal of Graph Theory}} \textbf{\bibinfo{volume}{70}}, \bibinfo{pages}{339--347} (\bibinfo{year}{2012}).
\newblock \urlprefix\url{https://onlinelibrary.wiley.com/doi/abs/10.1002/jgt.20620}.
\newblock \eprint{https://onlinelibrary.wiley.com/doi/pdf/10.1002/jgt.20620}.

\bibitem{Aaronson_2004}
\bibinfo{author}{Aaronson, S.} \& \bibinfo{author}{Gottesman, D.}
\newblock \bibinfo{title}{Improved simulation of stabilizer circuits}.
\newblock \emph{\bibinfo{journal}{Physical Review A}} \textbf{\bibinfo{volume}{70}} (\bibinfo{year}{2004}).
\newblock \urlprefix\url{http://dx.doi.org/10.1103/PhysRevA.70.052328}.

\bibitem{rank_width_lower_bound}
\bibinfo{author}{Bey{\ss}, M.}
\newblock \bibinfo{title}{Fast algorithm for rank-width}.
\newblock In \bibinfo{editor}{Ku{\v{c}}era, A.}, \bibinfo{editor}{Henzinger, T.~A.}, \bibinfo{editor}{Ne{\v{s}}et{\v{r}}il, J.}, \bibinfo{editor}{Vojnar, T.} \& \bibinfo{editor}{Anto{\v{s}}, D.} (eds.) \emph{\bibinfo{booktitle}{Mathematical and Engineering Methods in Computer Science}}, \bibinfo{pages}{82--93} (\bibinfo{publisher}{Springer Berlin Heidelberg}, \bibinfo{address}{Berlin, Heidelberg}, \bibinfo{year}{2013}).

\bibitem{maslov2018shorter}
\bibinfo{author}{Maslov, D.} \& \bibinfo{author}{Roetteler, M.}
\newblock \bibinfo{title}{Shorter stabilizer circuits via {Bruhat} decomposition and quantum circuit transformations}.
\newblock \emph{\bibinfo{journal}{IEEE Transactions on Information Theory}} \textbf{\bibinfo{volume}{64}}, \bibinfo{pages}{4729--4738} (\bibinfo{year}{2018}).

\bibitem{PhysRevLett.107.210404}
\bibinfo{author}{da~Silva, M.~P.}, \bibinfo{author}{Landon-Cardinal, O.} \& \bibinfo{author}{Poulin, D.}
\newblock \bibinfo{title}{Practical characterization of quantum devices without tomography}.
\newblock \emph{\bibinfo{journal}{Physical Review Letters}} \textbf{\bibinfo{volume}{107}}, \bibinfo{pages}{210404} (\bibinfo{year}{2011}).

\bibitem{stehlik2021tunable}
\bibinfo{author}{Stehlik, J.} \emph{et~al.}
\newblock \bibinfo{title}{Tunable coupling architecture for fixed-frequency transmon superconducting qubits}.
\newblock \emph{\bibinfo{journal}{Physical Review Letters}} \textbf{\bibinfo{volume}{127}}, \bibinfo{pages}{080505} (\bibinfo{year}{2021}).

\bibitem{javadi2024quantum}
\bibinfo{author}{Javadi-Abhari, A.} \emph{et~al.}
\newblock \bibinfo{title}{Quantum computing with qiskit}.
\newblock \emph{\bibinfo{journal}{arXiv preprint arXiv:2405.08810}}  (\bibinfo{year}{2024}).

\bibitem{nation2021scalable}
\bibinfo{author}{Nation, P.~D.}, \bibinfo{author}{Kang, H.}, \bibinfo{author}{Sundaresan, N.} \& \bibinfo{author}{Gambetta, J.~M.}
\newblock \bibinfo{title}{Scalable mitigation of measurement errors on quantum computers}.
\newblock \emph{\bibinfo{journal}{PRX Quantum}} \textbf{\bibinfo{volume}{2}}, \bibinfo{pages}{040326} (\bibinfo{year}{2021}).

\bibitem{wack2021scale}
\bibinfo{author}{Wack, A.} \emph{et~al.}
\newblock \bibinfo{title}{Scale, quality, and speed: three key attributes to measure the performance of near-term quantum computers}.
\newblock \emph{\bibinfo{journal}{arXiv preprint arXiv:2110.14108}}  (\bibinfo{year}{2021}).

\bibitem{Hahn1950}
\bibinfo{author}{Hahn, E.~L.}
\newblock \bibinfo{title}{Spin echoes}.
\newblock \emph{\bibinfo{journal}{Physical Review}} \textbf{\bibinfo{volume}{80}}, \bibinfo{pages}{580--594} (\bibinfo{year}{1950}).
\newblock \urlprefix\url{https://link.aps.org/doi/10.1103/PhysRev.80.580}.

\bibitem{Viola1998}
\bibinfo{author}{Viola, L.} \& \bibinfo{author}{Lloyd, S.}
\newblock \bibinfo{title}{Dynamical suppression of decoherence in two-state quantum systems}.
\newblock \emph{\bibinfo{journal}{Physical Review A}} \textbf{\bibinfo{volume}{58}}, \bibinfo{pages}{2733--2744} (\bibinfo{year}{1998}).
\newblock \urlprefix\url{https://link.aps.org/doi/10.1103/PhysRevA.58.2733}.

\bibitem{seif2024suppressing}
\bibinfo{author}{Seif, A.} \emph{et~al.}
\newblock \bibinfo{title}{Suppressing correlated noise in quantum computers via context-aware compiling}.
\newblock In \emph{\bibinfo{booktitle}{2024 ACM/IEEE 51st Annual International Symposium on Computer Architecture (ISCA)}}, \bibinfo{pages}{310--324} (\bibinfo{organization}{IEEE}, \bibinfo{year}{2024}).

\bibitem{coote2025resource}
\bibinfo{author}{Coote, P.} \emph{et~al.}
\newblock \bibinfo{title}{Resource-efficient context-aware dynamical decoupling embedding for arbitrary large-scale quantum algorithms}.
\newblock \emph{\bibinfo{journal}{PRX Quantum}} \textbf{\bibinfo{volume}{6}}, \bibinfo{pages}{010332} (\bibinfo{year}{2025}).

\end{thebibliography}
